\documentclass{article}

\usepackage[english]{babel}

\usepackage[letterpaper,top=2cm,bottom=2cm,left=3cm,right=3cm,marginparwidth=1.75cm]{geometry}

\usepackage{amsmath,amsfonts,amssymb,amsthm}
\usepackage{graphicx}
\usepackage[colorlinks=true, allcolors=blue]{hyperref}
\usepackage{natbib}
\usepackage{csquotes}
\usepackage{booktabs}
\usepackage{xcolor}
\usepackage{pifont}
\usepackage{hyperref}
\usepackage{caption}
\usepackage{subcaption}
\usepackage{bbm}
\usepackage{bm}
\usepackage{stackrel}
\usepackage{booktabs}

\usepackage{chngcntr} 
\counterwithin{figure}{section}
\counterwithin{table}{section}

\usepackage{setspace} 
\newcommand{\Q}{\mathbb{Q}}

\newcommand{\R}{\mathbb{R}}

\newcommand{\N}{\mathbb{N}}
\newcommand{\Prob}{\mathbb{P}}
\newcommand{\EP}{\mathbb{E}^{\mathbb{P}}}
\newcommand{\EQ}{\mathbb{E}^{\mathbb{Q}}}

\newcommand{\tin}{t \in [0, T]}
\newcommand{\rBrackets}[1]{\left( #1 \right)}
\newcommand{\sBrackets}[1]{\left[ #1 \right]}
\newcommand{\cBrackets}[1]{\left\{ #1 \right\}}

\newcommand{\psibar}{\overline{\psi}}
\newcommand{\cunder}{\underline{c}}
\newcommand{\Vunder}{\underline{V}}

\newcommand{\Ztilde}{\widetilde{Z}}
\newcommand{\Ztildet}[1]{\widetilde{Z}({#1})}

\newcommand{\intzeroT}{\int \limits_{0}^{T}}
\newcommand{\intttoT}{\int \limits_{t}^{T}}
\newcommand{\piStar}{\pi^{\ast}}
\newcommand{\pitil}{\widetilde{\pi}}
\newcommand{\pitilast}{\widetilde{\pi}^\ast}
\newcommand{\pibar}{\overline{\pi}}
\newcommand{\pibarast}{\overline{\pi}^\ast}

\newcommand{\cStar}{c^{\ast}}
\newcommand{\psiStar}{\psi^{\ast}}
\newcommand{\Xtil}{\widetilde{X}}
\newcommand{\XStar}{X^{\ast}}
\newcommand{\Xtilast}{\widetilde{X}^\ast}
\newcommand{\xtil}{\widetilde{x}}
\newcommand{\xbar}{\overline{x}}

\newcommand{\Xbarast}{\overline{X}^\ast}
\newcommand{\vast}{v^\ast}
\newcommand{\dbar}{\overline{d}}

\DeclareMathOperator*{\argmax}{arg\,max}

\allowdisplaybreaks 

\numberwithin{equation}{section}
\theoremstyle{definition}
\newtheorem{assumption}{Assumption}[section]

\newtheorem{lemma}[assumption]{Lemma}

\newtheorem{corollary}[assumption]{Corollary}
\newtheorem{proposition}[assumption]{Proposition}

\title{Framework for asset-liability management with fixed-term securities}
\author{Yevhen Havrylenko\thanks{Primary affiliation: Institute of Insurance Science, Ulm University, Helmholtzstrasse 20, 89081, Ulm, Germany.  \href{yevhen.havrylenko@uni-ulm.de}{yevhen.havrylenko@uni-ulm.de}}
\textsuperscript{\,,\,}\thanks{Secondary affiliation: Department of Mathematical Sciences, University of Copenhagen, Universitetsparken 5, 2100, Copenhagen, Denmark. E-mail address: \href{yh@math.ku.dk}{yh@math.ku.dk}}
}

\begin{document}
\date{}
\maketitle

\begin{abstract}
We consider an optimal investment-consumption problem for a utility-maximizing investor who has access to assets with different liquidity and whose consumption rate as well as terminal wealth are subject to lower-bound constraints. Assuming utility functions that satisfy standard conditions, we develop a methodology for deriving the optimal strategies in semi-closed form. Our methodology is based on the generalized martingale approach and the decomposition of the problem into subproblems. We illustrate our approach by deriving explicit formulas for agents with power-utility functions and discuss potential extensions of the proposed framework. In numerical studies, we substantiate how the parameters of our framework impact the optimal proportion of initial capital allocated to the illiquid asset, the monetary value that the investor subjectively assigns to the fixed-term asset, and the potential of the illiquid asset to increase terminal the terminal value of liabilities without loss in the investor's expected utility. 
\end{abstract}

\bigskip


\bigskip




\section{Introduction}\label{sec:introduction}

\quad\,\,\textbf{Motivation.} Both retail investors and institutional investors have access to assets with different liquidity. Furthermore, both types of investors may want to manage their financial risks so that consumption rate (dividend-payout rate for institutions) and terminal wealth are larger than or equal to the respective lower bounds predetermined at the beginning of an investment period. 

Indeed, insurance companies have very complex asset-allocation processes that involve both liquid (dynamically traded) and illiquid (fixed-term) assets (\cite{Albrecher2018}), can have regulatory constraints on terminal portfolio value (\cite{Boonen2017}), management rules (\cite{Albrecher2018}), etc. Examples of liquid assets are stocks, bonds, commodity futures, etc. Instances of illiquid assets are loans, land, private equity investments, etc. Often illiquid assets help generate higher returns due to a so-called illiquidity premium. Since insurance companies must ensure that their assets are sufficient to meet their liabilities to customers, they effectively have a constraint on their portfolio value at the time when their liabilities are due. Moreover, many insurance companies pay dividends, the stability of which is important for the reputation of those companies. As a result, the size of dividends should not fall below some minimal acceptable level (e.g., average dividend size across the insurance industry), as otherwise shareholders may become dissatisfied, which may have a negative impact on the share price of insurance companies. 

Similarly, retail investors also have access to both liquid assets (e.g., stocks, exchange-traded funds, etc.) and illiquid assets (e.g., housing, private equity, etc.). Retail investors may have specific goals regarding the amount of wealth they want to accumulate by the end of their investment horizons, e.g., the beginning of their retirement, the time when their children start attending higher-education institutions, etc.. Naturally, people prefer to have a standard of living that is gradually improving or at least not deteriorating. This implies that retail investors have a subsistence level of consumption of their wealth.

To the best of our knowledge, there is no academic literature on dynamic portfolio optimization (PO) with fixed-term assets under simultaneous downside constraints on terminal wealth and consumption rate. This fact and the practical importance of the topic motivate us to study optimal decisions in the setting described above.

\textbf{Goal.} We aim to construct an analytically tractable asset-liability management framework for utility-maximizing investors who control their consumption and investment strategies, have access to liquid (dynamically traded) and illiquid (fixed-term) assets, and have the following risk-management constraints:
\begin{enumerate}
    \item the terminal value of assets is higher than or equal to the terminal value of liabilities almost surely (lower bound on terminal wealth);
    \item the consumption (dividend) rate is larger than or equal to a fixed minimal level almost surely (lower bound on intermediate consumption)
\end{enumerate}
We achieve our goal by extending the existing PO literature by combining and generalizing various PO techniques.

\textbf{Literature overview.} We organize the overview of the literature in two blocks: sources on dynamic PO with constraints on wealth as well as consumption, sources on PO with illiquid assets. To keep the article concise, we do not attempt to provide an extensive overview, but focus only on the references most relevant for this research.  

\textit{PO with constraints on wealth and consumption.} The seminal papers on dynamic PO are \cite{Merton1969} and \cite{Merton1971}, where the classic continuous-time dynamic PO problem of a utility-maximizing investor in a Black-Scholes market was first considered and solved via stochastic optimal control approach in a setting without constraints on terminal wealth, consumption, or investment strategies. These articles are extended in \cite{Tepla2001}  and \cite{Korn2005}, where a terminal wealth constraint is added to the setting so that the investor's terminal portfolio value does not fall below some guaranteed level almost surely. Both papers use the martingale approach (MA), but the former paper considers a deterministic lower bound on terminal wealth, whereas the latter paper assumes a stochastic lower bound. Other researchers extend the results of \cite{Tepla2001}  and \cite{Korn2005} by generalizing constraints on terminal wealth, e.g.,  \cite{Basak2001}, \cite{Boyle2007}, \cite{Chen2018b}, \cite{Chen2018c}.

\cite{Lakner2006} considers an investor with both utility from consumption and utility from wealth. Assuming a complete Black-Scholes market, the authors derive the optimal investment-consumption strategies under the lower-bound constraints that prevent intermediate consumption and terminal wealth from falling below predetermined thresholds. To derive optimal controls in their setting, they apply MA and an initial-wealth separation technique, which was found useful already in unconstrained utility-maximization problems (see, e.g., Chapter 3 in \cite{Karatzas1998}). Other researchers extended \cite{Lakner2006} in several directions such as making the lower bound on consumption time-dependent (\cite{Ye2008}) or making the investor's relative-risk aversion time dependent (\cite{Lichtenstern2021}). We are not aware of any publications on the generalization of the framework of \cite{Lakner2006} to the presence of illiquid assets, which is one of the research gaps that we fill.

\textit{PO with both liquid and fixed-term assets.} There are various ways to model the illiquidity of assets in a PO problem. One approach is to assume that illiquid assets are traded at discrete time points, i.e., not continuously in contrast to liquid assets. This approach is used in \cite{Garleanu2009} and \cite{Ang2014}, where the authors assume that an illiquid risky asset cannot be traded during time intervals of uncertain duration. The former paper considers an economy with one liquid risk-free asset and one illiquid risky asset. For a continuum of decision makers with \textbf{c}onstant \textbf{a}bsolute \textbf{r}isk \textbf{a}version \textit{(CARA) utility from intermediate consumption} and an infinite investment horizon, \cite{Garleanu2009} uses dynamic programming approach (DPA) and approximation techniques to derive an equilibrium in the market. However, \cite{Ang2014} considers a financial market with one risk-free, one continuously-traded risky asset and one illiquid risky asset that can be traded only at random time points which are Poisson distributed. Also using DPA, the researchers derive the optimal investment-consumption strategies for an agent who has an \textit{infinite investment horizon} and whose preferences are modeled by either a \textbf{c}onstant \textbf{r}elative \textbf{r}isk \textbf{a}version \textit{(CRRA) utility from intermediate consumption or an Epstein-Zin utility}.  \cite{Desmettre2016} also considers an investment universe with three assets, but there an investor maximizes the expected utility of \textit{wealth} at the end of a \textit{finite investment} horizon and the illiquid asset is modeled as a \textit{fixed-term asset} that can be traded only at the beginning of the investment horizon. To solve this problem, \cite{Desmettre2016} generalizes MA to a two-step procedure. In the first step, a PO problem with an arbitrary but fixed position in the fixed-term asset is solved. In the second step, the optimal investment in fixed-term assets is found by optimizing the value function (conditional on the position in the illiquid asset) from the first step. \cite{Chen2023} extends the results of \cite{Desmettre2016} to the case of an investor with an S-shaped utility function and analyzes the optimal asset allocation from the perspective of a pension fund that invests in an indexed bond, a liquid stock, and an illiquid fixed-term asset. 

Another way to model illiquidity in PO is to include transaction costs, as in \cite{Bichuch2018}. Using DPA and an asymptotic-expansion technique, the authors derive optimal investment strategies for long-term investors who maximize their expected CRRA utility. Assuming that an investor has a lower bounded utility function and using the superharmonic function approach to impulse control problems, \cite{Belak2019} derive optimal investment strategies in a factor model with constant and proportional transaction costs. There are other publications on portfolio optimization with transaction costs. For more information, see \cite{Bichuch2018} or \cite{Belak2019} and the references therein.

In our paper, we use the approach to modeling illiquid assets in PO as in \cite{Desmettre2016}. We extend that paper in two directions. First, we allow the investor to gain utility from intermediate consumption in addition to utility from terminal wealth. Second, we incorporate lower-bound constraints on the investor's consumption rate and terminal wealth.

\textbf{Contribution.} Our theoretical contribution to the PO literature is the development of an analytically tractable PO framework for decision makers who have constraints on terminal wealth as well as  consumption rate and who have access to illiquid (fixed-term) as well as liquid assets. Our framework is based on the synthesis of ideas from \cite{Desmettre2016}, \cite{Korn2005}, and \cite{Lakner2006}. We model fixed-term assets as in \cite{Desmettre2016} and extend that article as described in the previous paragraph. When dealing with lower-bound constraints on consumption rate and terminal wealth, we generalize techniques from \cite{Korn2005} and \cite{Lakner2006}.

As for our practical contributions, we illustrate how various parameters of our framework influence the optimal capital allocation to illiquid assets, the subjective monetary value of fixed-term assets for the decision-maker, and the potential of illiquid assets to increase terminal liabilities of the investor without loss in the decision-maker's expected utility. This information may help both retail and institutional investors assess the tangible benefits of including fixed-term assets in their portfolios and understand the trade-offs involved in meeting financial guarantees. Furthermore, we give an example of how our framework can be used by institutional investors, like life insurers and pension funds, to learn more about their clients' risk appetites by analyzing their choices in stylized scenarios involving fixed-term assets.

\textbf{Paper structure.} In Section \ref{sec:PO_framework}, we formally describe our asset-liability framework. In Section \ref{sec:methodology}, we develop a methodology for the derivation of the optimal consumption-investment strategies in our framework. In Section \ref{sec:explicit_formulas_for_power_U}, we illustrate the analytical tractability of our theoretical methodology by deriving in semi-closed form the optimal decisions of an economic agent with power-utility functions. We discuss potential extensions of our framework in Section \ref{sec:discussion_on_extensions}. Section \ref{sec:numerical_studies} contains the results of our numerical studies for a decision maker with power-utility functions. Finally, Section \ref{sec:conclusion} concludes. Appendix \ref{app:proofs_general} contains the proofs of theoretical results from Section \ref{sec:methodology}. Appendix \ref{app:proofs_power_U} contains the proofs of theoretical results from Section \ref{sec:explicit_formulas_for_power_U}.

\section{Portfolio optimization framework}\label{sec:PO_framework}

\quad\,\,\,For the simplicity of the exposition, we consider a financial market with $1$ liquid risky asset and $1$ illiquid risky asset. However, it is straightforward to generalize our setting and the methodology for deriving optimal decisions in our framework to $n \in \N$ liquid risky assets and $m \in \N$ illiquid risky assets, as we briefly explain in Section \ref{sec:discussion_on_extensions}.

Let $W^{\Prob} = (W^{\Prob}(t))_{\tin}$ be a $1$-dimensional Wiener process on a filtered complete probability space $\left( \Omega, \mathcal{F}, \Prob, (\mathcal{F}(t))_{\tin}\right)$, where $\Omega$ is a sample space,  $\mathcal{F} = \mathcal{F}(T)$ is a sigma algebra on $\Omega$, $\Prob$ is the real-world probability measure, $\mathcal{F}(t)$, $\tin$, is the natural filtration generated by $W^{\Prob}(s),\, s \in [0, t],$ and augmented by the null sets.

The financial market contains three assets. Two of these assets are continuously traded without frictions, i.e., without transaction costs, bid-ask spreads, taxes, etc. The first asset is a risk-free asset whose price process we denote by $S_0(t)$, $\tin$. The second asset is a risky one whose price process is denoted by $S_1(t)$, $\tin$. The third asset is a fixed-term asset and can be traded only at time $t = 0$.

Under $\Prob$, the price dynamics of the risk-free asset, also referred to as a bank account, is given by the following ordinary differential equation (ODE):
\begin{equation*}
    dS_0(t) = rS_0(t)dt,\,\,\, S_0(0) = 1,\\
\end{equation*}
where $r > 0$ is a constant interest rate.

Under $\Prob$, the price process $S_1(t)$ of the liquid risky asset (also referred to as a risky fund) evolves according to the following stochastic differential equation (SDE):
\begin{equation*}
	dS_1(t)  = S_1(t)\rBrackets{ \mu  dt + \sigma  dW^{\Prob}(t)},\,\,\, S_1(0) = s_1,
\end{equation*}
where $\mu \in \R$ is a constant drift parameter, $\sigma > 0$ denotes a constant volatility of the liquid asset, and $s_1 > 0$ is the liquid risky asset's initial price. We denote the market price of risk by $\gamma :=\sigma^{-1}(\mu - r)$.

According to Theorem 3.26 on page 147 in \cite{Korn2014}, in the aforementioned market there exists a unique risk-neutral probability measure (also known as an equivalent martingale measure) that we denote by $\Q$:
\begin{equation}\label{eq:Qtilde_RN_derivative}
	\left. \frac{d \Q}{d \Prob}\right|_{\mathcal{F}(t)} := Z(t) :=  \exp \rBrackets{-\frac{1}{2} \gamma^2 t - \gamma W^{\Prob}(t)}.
\end{equation}

Further, we denote the associated pricing kernel by $\widetilde{Z} = \rBrackets{\widetilde{Z}(t)}_{\tin}$, which is also known as the state price density or the deflator. It is defined as follows:

\begin{equation}\label{eq:explicit_pricing_kernel}
\Ztildet{t} :=  \exp \left( -\left(r +  \frac{1}{2} \gamma^2 \right)t - \gamma W^{\Prob}(t) \right),\,\,\, t \in [0, T]
\end{equation}
and, thus, satisfies the following SDE under $\Prob$:
\begin{equation}\label{eq:SDE_pricing_kernel}
	d\Ztildet{t} = -\Ztildet{t} \rBrackets{r dt +  \gamma  d W^{\Prob}(t)}, \quad \Ztildet{0} = 1.
\end{equation}

Similarly to \cite{Desmettre2016}, we assume that a decision maker has access to a fixed-term asset that has a stochastic payoff $F(T)$ at time $T$. This asset can be purchased at a price of $F(0) = F_0$ at time $t = 0$. Importantly, the fixed-term asset cannot be traded after $t = 0$ due to contractual limitations such as lock-in periods (e.g. infrastructure projects) or the absence of an established liquid market (e.g., private-equity investments), cf. \cite{Desmettre2016}. To ensure the analytical tractability of our PO framework, we assume that $F(T)$ is $\mathcal{F}(T)$-measurable and $F(T) > 0$ $\Prob$-almost surely (a.s.). As long as these two conditions are satisfied, the probability distribution of $F(T)$ can be flexible. We do not impose any conditions on the risk-return profile of the fixed-term asset\footnote{In practice, illiquid investments usually have a higher return (due to a so-called illiquidity premium) and can have a lower volatility than those of liquid assets.}. However, intuitively, if the fixed-term asset is not \enquote{sufficiently attractive} in comparison with liquid assets, then a rational investor should not invest in $F$. Mathematically, this condition (non-redundancy of $F$) translates into the following inequality: $F(0) \leq \EQ\sBrackets{\exp\rBrackets{-rT} F(T)} = \EP\sBrackets{\Ztilde(T) F(T)}$. When $F(0) = \EP\sBrackets{\Ztilde(T) F(T)}$, then a rational investor is indifferent to the presence of $F$ in the investment universe. If $F(0) > \EP\sBrackets{\Ztilde(T) F(T)}$, a rational investor does not purchase $F$.

We denote by $v_0$ the initial wealth of the investor and by $\psi$ the \textit{absolute quantity} of the fixed-term asset that the investor buys at price $F(0) = F_0$ at time $t = 0$. We refer to $\psi$ as the illiquid investment strategy. We assume that:
\begin{equation}\label{eq:psi_admissibility_set}
    0 \leq \psi \leq \psi^{\max} = {v_0}/{F(0)}.
\end{equation}
This condition prevents the investor from taking a short position in the fixed-term asset and ensures solvency, i.e., the decision maker cannot buy more fixed-term assets than the total initial capital allows. We denote by $x_0 \in [0, v_0]$ the initial capital that remains after buying the fixed-term asset and that is used for liquid assets.

Let $\pi = (\pi(t))_{\tin}$ be the relative portfolio process (liquid investment strategy) with respect to the liquid risky asset $S_1$, $c = (c(t))_{\tin}$ be the consumption-rate process of the decision maker, $X(t):=X^{x_0, \rBrackets{\pi, \psi, c}}(t)$ be the value of the liquid portfolio of the investor at time $\tin$. For a retail investor, $c$ can be thought of as the rate at which the person withdraws money from the investment portfolio to cover lifestyle expenses. For an institutional investor, $c$ can be interpreted as the rate at which dividends are paid out and the operational expenses of the company are covered.

We assume that the decision maker does not inject additional capital and does not withdraw any capital from liquid wealth $X = (X(t))_{\tin}$ except for consumption $c$. Due to $\pi_0(t) = 1 - \pi(t),\, \forall \tin$, $X(t)$ has the following SDE: 
   \begin{equation}\label{eq:SDE_wealth_process}
   dX(t) = X(t)\rBrackets{\rBrackets{r + \pi(t)(\mu - r)}dt +\pi(t)  \sigma dW^{\Prob}(t)} - c(t)dt, \qquad X(0) = x_0.
   \end{equation}
   
We impose further regularity conditions on $(\pi, c)$:
\begin{equation}\label{eq:X_regularity_conditions}
	X(t) \geq 0\quad \forall \tin\quad \text{and}\quad  \int \limits_0^T \rBrackets{X(t)\pi(t)}^2\,dt < +\infty\qquad \Prob\text{-a.s.}.
\end{equation}

The condition that the liquid wealth is non-negative during the whole investment period is called a \textit{solvency requirement} in \cite{Desmettre2016} and is also assumed in \cite{Munk2000}. It prevents arbitrage opportunities, since trading has to be financed solely from the liquid capital and the economic agent has to declare insolvency as soon as the liquid portfolio value drops below $0$.

Let $\mathcal{A}_u(v_0)$ denote the set of \textbf{a}dmissible \textbf{u}nconstrained investment-consumption strategies for a decision maker with initial capital $v_0$, i.e., $(\pi, \psi, c)$ are such that  \eqref{eq:psi_admissibility_set} holds, $\pi$ and $c$ are progressively measurable,  \eqref{eq:SDE_wealth_process} has a unique solution that satisfies \eqref{eq:X_regularity_conditions}.

Let $V^{v_0, (\pi, \psi, c)}(t)$ be the value of the total portfolio consisting of both liquid and fixed-term assets at time $\tin$.

As mentioned in Section \ref{sec:introduction}, the investor has two risk-management constraints. The first constraint is a lower bound on the consumption process $c$, namely:
\begin{equation}\label{eq:consumption_constraint}
    c(t) \geq \underline{c} \quad \forall \,\tin,
\end{equation}
where $\underline{c} > 0$ represents a subsistence level for consumption (e.g., the consumption rate required by the lowest standard of living acceptable for a retail investor, minimal acceptable dividend-payout rate for an insurance company, etc.). Let $K_c := \left\{ c \, | \, c(t) \geq \underline{c} \,\,\forall \,\tin \right\}$ be the set of consumption processes that satisfy \eqref{eq:consumption_constraint}.

The second constraint is a lower bound on the terminal value $V^{v_0, (\pi, \psi, c)}(T)$ of the investor's total portfolio, namely:
\begin{equation}\label{eq:VaR_constraint}
    V^{v_0, (\pi, \psi, c)}(T) \geq \underline{V},
\end{equation}
where $\underline{V} > 0$ represents the threshold for terminal portfolio value (e.g., the minimal amount of wealth at retirement acceptable for a retail investor, the expected terminal value of liabilities of an insurance company, etc.) at time $T$. Let $K_V = \{ V^{v_0, (\pi, \psi, c)}(T)\,|\,V^{v_0, (\pi, \psi, c)}(T) \geq \underline{V} \}$ be the set of terminal portfolio values that satisfy \eqref{eq:VaR_constraint}.

As in \cite{Lakner2006}, we assume that the decision maker's preferences towards consumption are described by an objective function $U_1:[0, T] \times (0,+\infty) \rightarrow \R$, such that for each $\tin$ the function $U_1(t, \cdot)$ is a standard utility function, i.e., strictly increasing, strictly concave, continuously differentiable, and satisfying Inada conditions. Similarly, we assume that the investor's preferences toward terminal wealth are described by an objective function $U_2:[0, T] \times (0,+\infty) \rightarrow \R$, such that for each $\tin$ the function $U_2(t, \cdot)$ is a standard utility function.

Denoting the set of \textbf{a}dmissible \textbf{c}onstrained investment-consumption strategies by
\begin{equation*}
    \mathcal{A}(v_0; K_c, K_V) := \left\{ (\pi, \psi, c) \in \mathcal{A}_u(v_0)\,|\, c \in K_c, V^{v_0, (\pi, \psi, c)}(T) \in K_V \right\},
\end{equation*}
we finally state our complete framework as the following optimization problem: 
\begin{equation}\label{OP:DS2016_with_V_c_lower_bounds}\tag{$P^{\pi, \psi, c}_{\underline{V}, \underline{c}}$}
\begin{aligned}
&\max_{\pi, \psi, c} \quad \EP\sBrackets{\int \limits_{0}^{T} U_1\rBrackets{t, c(t)}dt + U_2\rBrackets{T, V^{v_0, (\pi, \psi, c)}(T)}}\\
& \text{ s.t.} \qquad {(\pi, \psi, c) \in \mathcal{A}(v_0; K_c, K_V)}.
\end{aligned}
\end{equation}

We denote by $(\piStar, \psiStar, \cStar) \in \mathcal{A}(v_0; K_c, K_V)$ a solution to \eqref{OP:DS2016_with_V_c_lower_bounds} and by $\mathcal{V}(v_0;K_c, K_V)$ the value function (VF) associated with \eqref{OP:DS2016_with_V_c_lower_bounds}:
\begin{equation}\label{eq:original_value_function}
    \mathcal{V}(v_0;K_c, K_V) := \EP\sBrackets{\int \limits_{0}^{T} U_1\rBrackets{t, \cStar(t)}dt + U_2\rBrackets{T, V^{v_0, (\piStar, \psiStar, \cStar)}(T)}}.
\end{equation}

\section{Methodology for deriving optimal strategies}\label{sec:methodology}

\quad\,\,\,Problem \eqref{OP:DS2016_with_V_c_lower_bounds} has two components: the utility of consumption (UoC) subproblem and the utility of wealth (UoW) subproblem. Due to that, we employ the initial-wealth separation technique used in \cite{Lakner2006}, where a similar but simpler -- due to the absence of a fixed-term asset -- PO problem is solved. They:
\begin{enumerate}
    \item split the investor's initial capital $v_0$ into two arbitrary but fixed parts such that one part is allocated to the UoC subproblem, another part is allocated to the UoW subproblem;
    \item solve each subproblem separately;
    \item derive the solution to the original problem by combining the solutions to the subproblems with the respective initial capital determined by the optimal split of the investor's initial capital $v_0$.
\end{enumerate}
As we shall see, the above-sketched separation approach also works for \eqref{OP:DS2016_with_V_c_lower_bounds}. However, the methodology for solving individual subproblems that arise from \eqref{OP:DS2016_with_V_c_lower_bounds} requires generalizing theoretical techniques from \cite{Korn2005}, \cite{Lakner2006}, and \cite{Desmettre2016} due to the presence of a fixed-term asset $F$. In the following, we formally describe how to solve \eqref{OP:DS2016_with_V_c_lower_bounds} and what is novel in our solution methodology compared to the approaches in those three articles.

\subsection{Maximization of utility from consumption subject to a lower-bound constraint}
Let $v_1 \in [0, v_0]$ be an arbitrarily fixed initial capital that the investor allocates to his/her/their UoC subproblem. The UoC subproblem associated with \eqref{OP:DS2016_with_V_c_lower_bounds} is given by:
\begin{equation}\label{OP:UoC_problem}\tag{$P^{\pi, \psi, c}_{1}$}
        \mathcal{V}_1(v_1; K_c) := \max_{(\pi, \psi, c) \in \mathcal{A}_u(v_1)} \EP \sBrackets{\int \limits_{t = 0}^{T} U_1\rBrackets{t, c(t)}dt }\quad \text{s.t.}\quad c \in K_c,
\end{equation}
where $\mathcal{V}_1(v_1; K_c)$ is the value function of this subproblem. 

In contrast to the UoC subproblem in \cite{Lakner2006}, \eqref{OP:UoC_problem} has a fixed-term asset $F$ and the related control variable $\psi$. Therefore, we solve \eqref{OP:UoC_problem} in two steps:
\begin{enumerate}
    \item derive $\rBrackets{\pi^{\ast}_1(t;v_1, \psibar_1), c^\ast_1(t;v_1, \psibar_1)}_{\tin}$ optimal for $\rBrackets{P^{\pi, \psibar_1, c}_{1}}$, where $\psibar_1$ is an arbitrarily fixed admissible $\psi$ in \eqref{OP:UoC_problem};
    \item determine the optimal $\psi^{\ast}_1$ by maximizing the VF associated with Problem $\rBrackets{P^{\pi, \psibar_1, c}_{1}}$ with respect to (w.r.t.) $\psibar_1$.
\end{enumerate}
We denote by $U'_1$ the derivative of $U_1$ w.r.t. the second argument. Let $I_1: [0, T] \times (0, +\infty) \mapsto [0, T] \times [\underline{c}, +\infty)$ be the generalized inverse function of $U'_1$:
\begin{equation}\label{eq:I_1_general_form}
   I_1(t, z) = \min \{ c \geq \underline{c}: U'_1(t, c) \leq z\},\quad\text{for each}\, \tin.
\end{equation}

From now on till the end of Section \ref{sec:methodology}, we make the following assumption, unless stated otherwise:
\begin{equation}\label{eq:UoCS_assumption_regarding_lambda}
    \EP \sBrackets{ \intzeroT \Ztilde(t) I_1\rBrackets{t, \lambda_1\Ztilde(t)}\,dt} < +\infty \quad \forall \lambda_1 \in (0,+\infty).
\end{equation}
This is a standard assumption (see also, e.g., Assumption 7.1 in Section 3.7 in \cite{Karatzas1998}) needed to ensure the existence of the optimal Lagrangian multiplier in \eqref{eq:UoC_lambdaStar} below. The optimal Lagrange multiplier is an integral part of the solution to \eqref{OP:UoC_problem}.

Before solving \eqref{OP:UoC_problem}, we first establish the minimal amount of capital needed to ensure that a solution to \eqref{OP:UoC_problem} exists.

\begin{lemma}[Minimal initial capital for UoC subproblem]\label{lem:v_1_min}
     The set of strategies admissible for \eqref{OP:UoC_problem} is non-empty if and only if $v_1 \geq v_1^{\min} := \cunder\rBrackets{1 - \exp\rBrackets{ - r T}} / r$.
\end{lemma}

We denote by $\mathcal{C}^{1,2}\rBrackets{[0, T] \times \R}$ the set of real-valued functions that are once continuously differentiable w.r.t. the first argument and twice continuously differentiable w.r.t. the second argument. The next proposition provides the solution to \eqref{OP:UoC_problem} and is a generalization of Proposition 3.1 in \cite{Lakner2006} to the case where an investor has access to a fixed-term asset.

\begin{proposition}[Solution to UoC subproblem]\label{prop:UoC_problem_solution}
    Consider Problem \eqref{OP:UoC_problem} with initial capital $v_1$.
    
    If $v_1 < v_1^{\min}$, then \eqref{OP:UoC_problem} does not admit a solution.
    
    If $v_1 = v_1^{\min}$, then $\pi_1^\ast(t;v_1) = 0$, $\psiStar_1 = 0$, $c_1^\ast(t;v_1) = \cunder$, $X_1^\ast(t;v_1) = \cunder \rBrackets{1 - \exp\rBrackets{ - r (T - t)}} / r$.
    
    If $v_1 > v_1^{\min}$, then under Assumption \eqref{eq:UoCS_assumption_regarding_lambda}:
    \begin{enumerate}
        \item[(i)] the optimal consumption process is given by
        \begin{equation}\label{eq:UoC_cStar}
            c_1^\ast(t;v_1) = I_1\rBrackets{t, \lambda^\ast_1(v_1) \Ztilde(t)},\,\tin,
        \end{equation}
        where $\lambda^\ast_1:=\lambda^\ast_1(v_1)$ is the unique solution to the budget-constraint equation:
        \begin{equation}\label{eq:UoC_lambdaStar}
            \EP\sBrackets{\intzeroT \Ztilde(t) c_1^\ast(t; v_1)\,dt} = v_1;
        \end{equation}
        \item[(ii)] the optimal position in the fixed-term asset is given by:
        \begin{equation}\label{eq:UoC_psiStar}
            \psi_1^\ast = 0;
        \end{equation}
        \item[(iii)] the optimal portfolio value is given by:
        \begin{equation}\label{eq:UoC_XStar}
            X_1^\ast(t;v_1) = \EQ \sBrackets{\int \limits_{t}^{T}\exp\rBrackets{-r (s - t)} c^\ast_1(s;v_1)\,ds | \mathcal{F}(t)}, 
        \end{equation}
        and the optimal terminal value of the total portfolio is given by
        \begin{equation}\label{eq:UoC_VStar}
            V^\ast_1(T; v_1) = X_1^\ast(T;v_1) = 0;
        \end{equation}
        \item[(iv)] the optimal relative portfolio process exists and can be found from equating the SDE for $X_1^\ast(t;v_1)$ in \eqref{eq:UoC_XStar} with the SDE \eqref{eq:SDE_wealth_process} where $\pi_1 = \pi_1^{\ast}$. If $X_1^\ast(t;v_1) = g_1(t, W^{\Prob}(t))$ for some $g_1(t, w) \in \mathcal{C}^{1,2}\rBrackets{[0, T] \times \R}$ with $g_1(0, 0) = v_1$, then: 
        \begin{equation}\label{eq:UoC_piStar}
            \pi_1^\ast(t;v_1) = \frac{1}{X_1^\ast(t;v_1)} \frac{1}{\sigma} \frac{\partial}{\partial w}g_1(t, W^{\Prob}(t)).
        \end{equation}
    \end{enumerate}
\end{proposition}

\subsection{Maximization of utility from wealth subject to a lower-bound constraint}\label{subsec:UoW}
Let $v_2 \in [0, v_0]$ be an arbitrarily fixed initial capital that the investor allocates to the UoW subproblem. Then the UoW subproblem associated with \eqref{OP:DS2016_with_V_c_lower_bounds} is given by:
\begin{equation}\label{OP:UoW_original}\tag{$P^{\pi, \psi, c}_{2}$}
        \mathcal{V}_2(v_2; K_V) := \max_{ (\pi, \psi, c) \in \mathcal{A}_u(v_2) } \EP \sBrackets{U_2\rBrackets{T, V^{v_2, (\pi, \psi, c)}(T)} }\, \text{s.t.}\, V^{v_2, (\pi, \psi, c)}(T) \in K_V,
\end{equation}
where $\mathcal{V}_2(v_2; K_V)$ is the value function associated with \eqref{OP:UoW_original}.

In contrast to the setting in \cite{Korn2005}, \eqref{OP:UoW_original} has two additional controls, namely the consumption rate and the illiquid investment strategy. Regarding the former control variable, we can use the same argument as \cite{Lakner2006} uses to show that the optimal consumption rate is zero in their UoW problem. In particular, we have $c_2^{\ast}(t; v_2) = 0$, since also in our setting $U_2(\cdot, \cdot)$ explicitly depends only on terminal time and wealth, whereas the consumption rate is not its argument. Regarding the illiquid investment strategy, we combine the two-step approach of \cite{Desmettre2016} with the methodology used in \cite{Korn2005}.

Fix an arbitrary admissible $\psibar_2$, which implies that the decision maker invests $\psibar_2 F(0)$ in the fixed-term asset. The remaining capital for liquid assets is denoted by $x_2 := x_2(\psibar_2) = v_2 - \psibar_2 F(0)$. When we refer to the value $X^{x_2, (\pi,\psibar_2, c)}(t)$ of the liquid portfolio at time $\tin$, we will use the shorter notation $X^{x_2, \pi}(t)$, because $x_2$ contains the information about $\psibar_2$ and $c = 0$ as argued in the previous paragraph. The constraints on terminal wealth
\begin{equation*}
    X^{x_2, \pi}(T) + \psibar_2 F(T) \geq \Vunder\quad\text{and}\quad X^{x_2, \pi}(T) \geq 0
\end{equation*}
are equivalent to:
\begin{equation}\label{eq:B_definition}
 X^{x_2, \pi}(T)  \geq \max \cBrackets{\Vunder - \psibar_2 F(T),0} =: \rBrackets{\Vunder - \psibar_2 F(T)}^{+} =: B^{\psibar_2} =: B. 
\end{equation}
Therefore, the UoW problem for a fixed $\psibar_2$ is given by:
\begin{equation}\label{OP:UoW_fixed_psi}\tag{$\overline{P}^{\pi, \psibar_2}_{2}$} 
\begin{aligned}
    \overline{\mathcal{V}}_2(x_2):= \max_{\pi} \, \EP \sBrackets{U_2\rBrackets{T, X^{x_2, \pi}(T) + \psibar_2 F(T)}}\quad \text{s.t.} \quad \pi \in \overline{\mathcal{A}}(x_2; B), 
\end{aligned}     
\end{equation}
where the admissibility set is given by:
\begin{equation}\label{eq:def_A_bar_in_UoW_subproblem}
    \begin{aligned}
        \overline{\mathcal{A}}(x_2; B) := \biggl\{ \pi\text{ is progr. meas.} \,\biggl|\, & X^{x_2, \pi}(T) \geq B \quad \Prob\text{-a.s.}, \intzeroT  \rBrackets{\pi(t)X^{x_2, \pi}(t)}^2\,dt < +\infty,\\
        & \mathbb{E}^{\Prob}\sBrackets{\max \left\{ -U_2 \rBrackets{T, V^{v_2, (\pi, \psi, c)}(T)}, 0 \right\} }  < +\infty  \biggr\}.
    \end{aligned}
\end{equation}
The last condition in \eqref{eq:def_A_bar_in_UoW_subproblem} ensures that the expected utility of terminal wealth is well defined. Common utility functions with $U(T, \cdot) \geq 0$ trivially satisfy this condition. Let us denote by $\pibarast_2 = \rBrackets{\pibarast_2(t;x_2)}_{\tin}$ the investment strategy optimal for \eqref{OP:UoW_fixed_psi} and by $\Xbarast_2 = \rBrackets{\Xbarast_2(t;x_2)}_{\tin}$ the respective optimal liquid wealth process. 

Comparing the main problem (P) in \cite{Korn2005} and our \eqref{OP:UoW_fixed_psi}, we see that they differ only by the presence of an additional $\mathcal{F}(T)$-measurable term $\psibar_2 F(T)$ inside the utility function in \eqref{OP:UoW_fixed_psi}. Problem \eqref{OP:UoW_fixed_psi} has a non-empty set of admissible strategies if the initial capital $x_2$ is sufficient to meet the terminal wealth target:
\begin{equation}\label{eq:UoW_fixed_psi_minimal_capital}
    x_2 = x_2(\psibar_2) \geq \EP\sBrackets{\Ztilde(T)B}=: x_B(\psibar_2) =: x_B.
\end{equation}
The quantity $x_B$ is the minimal price needed to replicate the lower bound $B= \rBrackets{\Vunder - \psibar_2 F(T)}^{+}$ on the terminal wealth in \eqref{OP:UoW_fixed_psi}.  

To solve \eqref{OP:UoW_fixed_psi}, we link it to the following auxiliary unconstrained problem:
\begin{equation}\label{OP:UoW_auxiliary_unconstrained}\tag{$\widetilde{P}^{\pi, \psibar_2}_{2}$} 
    \begin{aligned}
        \max_{\pi} \mathbb{E}\sBrackets{U_2\rBrackets{T, X^{\xtil_2, \pi}(T) + \psibar_2 F(T) + B}}\quad \text{s.t.} \quad \pi \in \overline{\mathcal{A}}(\xtil_2; 0),
    \end{aligned}
\end{equation}
where:
\begin{equation}\label{eq:xtil_xB_relation}
    \xtil_2 := \xtil_2(\psibar_2) := x_2(\psibar_2) - x_B(\psibar_2) = v_2 - \psibar_2 F_0 - x_B(\psibar_2).
\end{equation}

Using the definition of $B$, we obtain that
\begin{equation*}
\psibar_2 F(T) + B = \psibar_2 F(T) + \max \cBrackets{\Vunder - \psibar_2 F(T),0} = \max \cBrackets{\Vunder, \psibar_2 F(T)}.    
\end{equation*}
Denoting the random utility function\footnote{It is denoted by $U_Z(x)$ on p. 316 in \cite{Korn2005} and by $\widetilde{U}(x)$ on p. 317 there} by
\begin{equation}\label{eq:def_U_tilde}
    \widetilde{U}(x):=U_2(T, x + \max \cBrackets{\Vunder, \psibar_2 F(T)}),
\end{equation}
we can write \eqref{OP:UoW_auxiliary_unconstrained} in a more compact form that is consistent with the notation of \cite{Korn2005} (see (P) on page 315, (10) and (11) on page 317 there):
\begin{equation}\label{OP:UoW_auxiliary_unconstrained_compact}\tag{$\widetilde{P}^{\pi}_{2}$}
        \widetilde{\mathcal{V}}_2(\xtil_2):= \max_{\pi \in \overline{\mathcal{A}}(\xtil_2; 0)} \mathbb{E}\sBrackets{\widetilde{U}\rBrackets{X^{\xtil_2, \pi}(T)}}.
\end{equation}

We denote by $I_2(\cdot)$ the inverse function of $U_2'(T, \cdot)$. From now till the end of Section \ref{sec:methodology}, we assume that the following condition holds, unless stated otherwise:
\begin{equation}\label{eq:UoWS_assumption_regarding_lambda}
    \EP \sBrackets{ \Ztilde(T) I_2\rBrackets{\lambda_2\Ztilde(t)} } < +\infty \quad \forall \lambda_2 \in (0,+\infty).
\end{equation}
Similarly to \eqref{eq:UoCS_assumption_regarding_lambda}, \eqref{eq:UoWS_assumption_regarding_lambda} is a standard assumption for portfolio optimization problems. It is needed for ensuring the existence of a unique Lagrange multiplier that characterizes the solution to \eqref{OP:UoW_auxiliary_unconstrained_compact} (e.g., cf. Assumption 7.2 in Section 3.7 in \cite{Karatzas1998}, Assumption (2) in \cite{Korn2005}).

Furthermore, we denote by $\widetilde{I}_2(\cdot)$ the extended inverse function of $\widetilde{U}'(\cdot)$:
\begin{equation}\label{eq:def_I_2_tilde}
    \widetilde{I}_2(y) = \max \cBrackets{{I}_2(y) - \max \cBrackets{\Vunder, \psibar_2 F(T)}, 0},
\end{equation}
where we extended the domain of the inverse function to all non-negative $y$ (cf. Remark 2 in \cite{Korn2005} or Definition 3.2 in \cite{Desmettre2016}). The function $\widetilde{I}_2(\cdot)$ plays a crucial role in the solution to \eqref{OP:UoW_auxiliary_unconstrained}, which is stated in the next proposition. 

\begin{proposition}[Solution to the auxiliary unconstrained problem \eqref{OP:UoW_auxiliary_unconstrained_compact}]\label{prop:solution_to_UoW_auxiliary_unconstrained}
    Consider the problem \eqref{OP:UoW_auxiliary_unconstrained_compact} with $\widetilde{x}_2 > 0$. Then under Assumption \eqref{eq:UoWS_assumption_regarding_lambda}:
    \begin{enumerate}
        \item the optimal terminal wealth is given by:
        \begin{equation}\label{eq:UoW_auxiliary_unconstrained_problem_X_T}
            \Xtilast_2(T;\xtil_2) = \widetilde{B}^\ast =  \widetilde{I}_2\rBrackets{\widetilde{\lambda}_2^\ast(\widetilde{x}_2)\Ztilde(T)},
        \end{equation}
        where $\widetilde{\lambda}_2^\ast = \widetilde{\lambda}_2^\ast\rBrackets{\xtil_2}$ is the unique solution to the budget-constraint equation:
        \begin{equation}\label{eq:UoW_auxiliary_unconstrained_problem_lambda}
            \EP\sBrackets{\Ztilde(T)\widetilde{B}^\ast}  = \xtil_2;
        \end{equation}
        \item the optimal intermediate wealth equals:
        \begin{equation}\label{eq:UoW_auxiliary_unconstrained_problem_X}
            \Xtilast_2(t;\xtil_2) = \EP\sBrackets{\widetilde{B}^\ast{\Ztilde(T)} / {\Ztilde(t)|\mathcal{F}(t)}};
        \end{equation}
        \item the optimal investment strategy exists and can be obtained from equating the SDE for $\Xtilast_2(t; \xtil_2)$ in \eqref{eq:UoW_auxiliary_unconstrained_problem_X} with the SDE \eqref{eq:SDE_wealth_process} where $\pi_1 = \pitilast_2$. If $\Xtilast_2(t;\widetilde{x}_2) = \widetilde{g}_2(t, W^{\Prob}(t))$ for some $\widetilde{g}_2(t, w) \in \mathcal{C}^{1,2}\rBrackets{[0, T] \times \R}$ with $\widetilde{g}_2(0, 0) = \widetilde{x}_2$, then: 
        \begin{equation}\label{eq:UoW_auxiliary_unconstrained_problem_pi}
            \pitilast_2(t;\widetilde{x}_2) = \frac{1}{\Xtilast_2(t;\widetilde{x}_2)} \frac{1}{\sigma} \frac{\partial}{\partial w}\widetilde{g}_2(t, W^{\Prob}(t)).
        \end{equation}
    \end{enumerate}
\end{proposition}

Theorem 2 in \cite{Desmettre2022} provides another approach to deriving an alternative representation of $\pitilast_2$. This approach is similar to the approach stated in Theorem 21 in \cite{Korn1997} in the sense that both theorems rely on the comparison of coefficients. However, Theorem 2 in \cite{Desmettre2022} represents the optimal strategy in terms of the pricing kernel $\Ztilde$ and the fixed-term asset $F$, assuming that its price at any $\tin$ can be observed.

It follows from \eqref{eq:def_I_2_tilde} and \eqref{eq:UoW_auxiliary_unconstrained_problem_X_T} that the terminal value of the optimal liquid portfolio can be zero. As explained in \cite{Desmettre2016}, this is caused by the fact that $\widetilde{U}$ does not satisfy the Inada conditions in general, and the investor's marginal utility can be finite because of his/her/their long position in the fixed-term asset.

One of the ingredients for the link between \eqref{OP:UoW_auxiliary_unconstrained_compact} and \eqref{OP:UoW_fixed_psi} is pricing and replicating the contingent claim $B = \rBrackets{\Vunder - \psibar_2 F(T)}^{+}$ defined in \eqref{eq:B_definition}. Due to the assumption that $F(T)$ is $\mathcal{F}(T)$-measurable and the liquid financial market is complete, the fair value of $B$ at time $\tin$ is unique and is given by (cf. Chapter 3 in \cite{Korn2014}):
\begin{equation}\label{eq:contingent_claim_B_fair_value}
    X_B(t;x_B) = \EP\sBrackets{B \Ztilde(T) / \Ztilde(t)|\mathcal{F}(t)},\quad \forall \tin. 
\end{equation}

The respective replicating strategy $\pi_B(t; x_B)$ can be found using the comparison of coefficients, as in Theorem 21 in \cite{Korn1997} or Theorem 2 in \cite{Desmettre2022}.

The next proposition establishes the link between \eqref{OP:UoW_auxiliary_unconstrained_compact} and \eqref{OP:UoW_fixed_psi}. This theoretical result generalizes Theorem 5 in \cite{Korn2005} to the case when a utility function in problem $(P)$ in \cite{Korn2005} is random.

\begin{proposition}[Solution to \eqref{OP:UoW_fixed_psi}]\label{prop:solution_to_UoW_fixed_psi}
    Consider \eqref{OP:UoW_fixed_psi} with initial capital $x_2$. Let $x_B$ be given by \eqref{eq:UoW_fixed_psi_minimal_capital}, $X_B(t;x_B)$ be given by \eqref{eq:contingent_claim_B_fair_value} and let $\pi_B(t; x_B)$ be the respective strategy that replicates $X_B(t;x_B)$.
  
    If $x_2 <  x_B$, then \eqref{OP:UoW_fixed_psi} does not admit a solution.

    If $x_2 = x_B$, then $\Xbarast_2(t; x_2) = X_B(t;x_B)$ and $\pibarast_2(t;x_2) = \pi_B(t; x_B)$ for $\forall \tin$.
    
    If $x_2 > x_B$, consider \eqref{OP:UoW_auxiliary_unconstrained_compact} with initial wealth $\widetilde{x}_2 > 0$ given by \eqref{eq:xtil_xB_relation}. Denote the optimal strategy 
    and the optimal wealth of this auxiliary unconstrained problem by $\pitilast_2$ and $\Xtilast_2$ respectively. Then under Assumption \eqref{eq:UoWS_assumption_regarding_lambda}:
    \begin{enumerate}
        \item the optimal terminal wealth is given by:
        \begin{equation}\label{eq:UoW_fixed_psi_X}
            \overline{X}_2^\ast(T; x_2) = B^\ast = B + \widetilde{B}^\ast,
        \end{equation}
        where $B$ is the lower bound on liquid terminal wealth as per \eqref{eq:B_definition} and $\widetilde{B}^\ast$ is the optimal terminal wealth \eqref{eq:UoW_auxiliary_unconstrained_problem_X_T} in the respective auxiliary unconstrained problem \eqref{OP:UoW_auxiliary_unconstrained_compact}. 
        \item the optimal wealth at $\tin$ is given by:
        \begin{equation}\label{eq:UoW_auxiliary_unconstrained_problem_X_t}
            \Xbarast_2(t; x_2) = X_B(t; x_B) + \Xtilast_2\rBrackets{t;\xtil_2}.
        \end{equation}
        \item the optimal amount of money invested in the liquid risky asset at $\tin$ is given by
        \begin{equation}\label{eq:UoW_fixed_psi_pi}
           \pibarast_2(t; x_2) = \frac{\pi_B(t; x_B) X_B(t; x_B) + \pitilast_2(t; \xtil_2) \Xtilast_2\rBrackets{t;\xtil_2}}{ \Xbarast_2(t; x_2)}.
        \end{equation}
    \end{enumerate}
\end{proposition}

In the next lemma, we establish the minimal amount $v_2^{\min}$ of initial capital needed to ensure that \eqref{OP:UoW_original} admits a solution.

\begin{lemma}[Minimal initial capital for UoW subproblem]\label{lem:v_2_min}
    The set of strategies admissible for \eqref{OP:UoW_original} is non-empty if and only if $v_2 \geq v_2^{\min} := \exp\rBrackets{-r T} \Vunder$.
\end{lemma}

The minimal capital requirement in Lemma \ref{lem:v_2_min} is the same as in \cite{Lakner2006}. However, the proof of our result is different due to the presence of the fixed-term asset, which makes the lower bound of the terminal liquid portfolio value random.

In the next proposition, we provide the solution to \eqref{OP:UoW_original}.
\begin{proposition}[Solution to UoW subproblem]\label{prop:UoW_problem_solution}
    Consider \eqref{OP:UoW_original} with initial capital $v_2$.
    
    If $v_2 < v_2^{\min}$, then \eqref{OP:UoW_original} does not admit a solution.

    If $v_2 = v_2^{\min}$, then $\pi_2^\ast(t;v_2)=0$, $\psi_2^\ast = 0$, $c_2^\ast(t;v_2) = 0$, $X_2^\ast(t;v_2) = \exp\rBrackets{ -r (T - t)}\Vunder = v_2 \exp\rBrackets{r t}$ for $\forall \tin$, and $V^\ast_2(T; v_2) = \Vunder$.

    If $v_2 > v_2^{\min}$, then under Assumption \eqref{eq:UoWS_assumption_regarding_lambda}: 
    \begin{enumerate}
        \item[(i)] the optimal consumption process is given by
        \begin{equation}\label{eq:UoW_cStar}
            c_2^\ast(t;v_2) = 0,\,\forall\tin;
        \end{equation}
        \item[(ii)] the optimal position in the fixed-term asset is given by:
        \begin{equation}\label{eq:UoW_psiStar}
            \psi_2^\ast = \argmax_{\psibar_2 \in \sBrackets{0, \frac{v_2}{F(0)}}} \widetilde{\mathcal{V}}_2\rBrackets{\xtil_2\rBrackets{\psibar_2}},
        \end{equation}
        where $\widetilde{\mathcal{V}}_2\rBrackets{\cdot}$ is the value function of \eqref{OP:UoW_auxiliary_unconstrained};
        \item[(iii)] the optimal liquid portfolio value is given by
        \begin{equation}\label{eq:UoW_XStar}
            X_2^\ast(t;v_2) =  \Xbarast_2(t; x_2\rBrackets{\psi_2^\ast}) = X_B(t; x_B\rBrackets{\psi_2^\ast}) + \Xtilast_2(t; \xtil_2\rBrackets{\psi_2^\ast}),\,\forall\tin,
        \end{equation}
        and the optimal terminal value of the total portfolio is given by
        \begin{equation}\label{eq:UoW_VStar}
            V^\ast_2(T; v_2) = \max \cBrackets{\Vunder, \psiStar_2 F(T), I_2\rBrackets{\widetilde{\lambda}_2^{\ast} \Ztildet{T}}};
        \end{equation}
        \item[(iv)] the optimal relative portfolio process is given by
        \begin{equation}\label{eq:UoW_piStar}
            \pi_2^\ast(t;v_2) = \frac{\pi_B(t; x_B\rBrackets{\psi_2^\ast}) X_B(t;x_B\rBrackets{\psi_2^\ast}) + \pitilast_2(t;\xtil_2\rBrackets{\psi_2^\ast}) \Xtilast_2\rBrackets{t;\xtil_2\rBrackets{\psi^\ast_2}}}{X_2^\ast(t;v_2)}.
        \end{equation}
    \end{enumerate}
\end{proposition}

\textbf{Remark.} Using \eqref{eq:UoW_VStar}, we can write \eqref{eq:UoW_psiStar} more explicitly as follows:
\begin{equation*}
    \psi_2^\ast = \argmax \limits_{\psibar_2 \in \sBrackets{0, \frac{v_2}{F(0)}}} \EP\sBrackets{U_2\rBrackets{T, \max \cBrackets{\Vunder, \psibar_2 \cdot F(T), I_2\rBrackets{\widetilde{\lambda}_2^{\ast}\rBrackets{\xtil_2\rBrackets{\psibar_2}} \cdot \Ztildet{T}}} }}.
\end{equation*}
Letting $\Vunder \downarrow 0$ in the above formula, we recover the results in Corollary 3.4 of \cite{Desmettre2016} for $\alpha = 0$.

\subsection{Solution to the main problem via the optimal combination of solutions to subproblems}
In Proposition \ref{prop:UoC_problem_solution} and Proposition \ref{prop:UoW_problem_solution} we provide solutions to UoC subproblem \eqref{OP:UoC_problem} and UoW subproblem \eqref{OP:UoW_original} respectively. The solutions to these subproblems depend on the investor's initial capital allocated to them. In the next proposition we establish the link between the VF in the UoC subproblem, the VF in the UoW subproblem, and the VF in the orignial problem \eqref{OP:DS2016_with_V_c_lower_bounds}.

\begin{proposition}\label{prop:maximization_problem_for_vStar_1_vStar_2}
    Assume that $v_0 \geq v_0^{\min}:= v_1^{\min} + v_2^{\min}$ holds. Then:
    \begin{equation}\label{eq:value_functions_merging}
        \mathcal{V}(v_0; K_c, K_V) = \max_{(v_1, v_2) \in \Lambda(v_0)} \cBrackets{\mathcal{V}_1(v_1; K_c) + \mathcal{V}_2(v_2; K_V)},
    \end{equation}
    where $\Lambda(v_0) = \left\{ (v_1, v_2)^\top\,:\,v_1 \geq v_1^{\min},\,v_2 \geq v_2^{\min},\,v_1 + v_2 = v_0\right\}$.
\end{proposition}

\begin{lemma}\label{lem:vStar_1_vStar_2}
    The optimal distribution of the investor's initial capital is given by:
    \begin{align}
        &\vast_1 = v_1^{\min} \mathbbm{1}_{\bar{v}_1 \in (0, v_1^{\min})} + \bar{v}_1 \mathbbm{1}_{\bar{v}_1 \in [v_1^{\min}, v_0 - v_2^{\min})} + \rBrackets{v_0 - v_2^{\min}} \mathbbm{1}_{\bar{v}_1 \in [v_0 - v_2^{\min}, +\infty)}; \label{eq:vStar_1}\\
        &\vast_2 = v_0 - \vast_1, \label{eq:vStar_2}
    \end{align}
    where $\bar{v}_1$ solves the following equation:
    \begin{equation}
        \mathcal{V}_1'(v_1; K_c) - \mathcal{V}_2'(v_0 - v_1; K_V) = 0\label{eq:vBar_1}
    \end{equation}
    with $\mathcal{V}_1'$ and $\mathcal{V}_2'$ denoting the derivatives w.r.t. the respective initial wealth.
\end{lemma}

Finally, in the next proposition we show how the investor's initial capital should be distributed among the UoC and UoW subproblems and how the corresponding solutions to the subproblems should be combined to obtain the solution to the original problem \eqref{OP:DS2016_with_V_c_lower_bounds}. 

\begin{proposition}[Solution to \eqref{OP:DS2016_with_V_c_lower_bounds}]\label{prop:solution_to_DS2016_with_V_c_lower_bounds}
    Assume that $v_0 \geq v_0^{\min}:= v_1^{\min} + v_2^{\min}$ and that Assumptions \eqref{eq:UoCS_assumption_regarding_lambda} and \eqref{eq:UoWS_assumption_regarding_lambda} are satisfied. Let $\vast_1$ and $\vast_2$ be given by \eqref{eq:vStar_1} and \eqref{eq:vStar_2} respectively. Then for the problem \eqref{OP:DS2016_with_V_c_lower_bounds}:
    \begin{enumerate}
        \item[(i)] the value function is given by
        \begin{equation}\label{eq:value_function}
            \mathcal{V}(v_0;K_c, K_V) = \mathcal{V}_1(v_1^\ast; K_c) + \mathcal{V}_2(v_2^\ast; K_V);
        \end{equation}
        \item[(ii)] the value of the optimal liquid portfolio at $\tin$ is given by
        \begin{equation}\label{eq:XStar_t}
            X^\ast(t; v_0) = X_1^\ast(t;v^{\ast}_1) + X_2^{\ast}(t;v^{\ast}_2), \quad \forall\tin;
        \end{equation}
        \item[(iii)] the optimal consumption rate at $\tin$ is given by
        \begin{equation}\label{eq:cStar}
            c^\ast(t; v_0) = c^\ast_1(t;v_1^\ast), \quad \forall\tin;
        \end{equation}
        \item[(iv)] the optimal position in the fixed-term asset is given by
        \begin{equation}\label{eq:psiStar}
            \psi^{\ast} = \psi^{\ast}_2(v_2^\ast);
        \end{equation}
        \item[(v)] the optimal portfolio strategy at $\tin$ is given by
        \begin{equation}\label{eq:piStar}
            \pi^{\ast}(t; v_0) = \frac{\pi^{\ast}_1(t;v_1^\ast) \cdot X^{\ast}_1(t;v^\ast_1) + \pi^{\ast}_2(t;v_2^\ast)\cdot X^{\ast}_2(t;v^\ast_2)}{X^\ast(t; v_0)}, \quad \forall\tin;
        \end{equation}
        \item[(vi)] the terminal value of the optimal total portfolio equals:
        \begin{equation}\label{eq:VStar_T}
            V^\ast(T; v_0) = V_2^\ast(T; v^{\ast}_2) = X_2^\ast(T; v^{\ast}_2) + \psiStar F(T).
        \end{equation}
    \end{enumerate}
\end{proposition}

\subsection{Summary of our solution methodology and theoretical contributions}
Next we summarize the approach to deriving the optimal consumption-investment strategies within the framework \eqref{OP:DS2016_with_V_c_lower_bounds}:
\begin{enumerate}
    \item Derive $\rBrackets{\pi_1^\ast(t;v_1), \psi_1^\ast(v_1), c_1^\ast(t;v_1)}_{\tin}$ optimal for the UoC subproblem \eqref{OP:UoC_problem} by using the generalized MA and obtain the respective optimal wealth $X^{\ast}_1(t;v^{\ast}_1), \tin,$ and value function $\mathcal{V}_1(v_1; K_c)$:
    \begin{enumerate}
        \item solve UoC subproblem for an arbitrarily fixed admissible $\psibar_1$;
        \item maximize w.r.t. $\psibar_1$ the value function from Step 1.(a).
    \end{enumerate}
    \item Derive $\rBrackets{\pi_2^\ast(t;v_2), \psi_2^\ast(v_2), c_2^\ast(t;v_2)}_{\tin}$ optimal for the UoW subproblem \eqref{OP:UoW_original} by applying the generalized MA and obtain the respective optimal wealth $X^{\ast}_2(t;v^{\ast}_2), \tin,$ and value function $\mathcal{V}_2(v_2; K_V)$:
    \begin{enumerate}
        \item solve UoW subproblem for an arbitrarily fixed admissible $\psibar_2$:
            \begin{enumerate}
                \item replicate the lower bound on terminal portfolio value;
                \item solve the auxiliary unconstrained problem \eqref{OP:UoW_auxiliary_unconstrained}.
            \end{enumerate}
        \item maximize w.r.t. $\psibar_2$ the conditional value function Step 2.(a).
    \end{enumerate} 
    \item Optimally combine the solutions obtained in Step 1 and Step 2 to get the solution to \eqref{OP:DS2016_with_V_c_lower_bounds}
    \begin{align*}
        \pi^{\ast}(t;v_0) & = \rBrackets{\pi^{\ast}_1(t;v_1^\ast) \cdot X^{\ast}_1(t;v^\ast_1) + \pi^{\ast}_2(t;v_2^\ast) \cdot X^{\ast}_2(t;v^\ast_2)} / \rBrackets{X^{\ast}_1(t;v^{\ast}_1) + X^{\ast}_2(t;v^{\ast}_2)}, \quad \forall\tin; \\
        \psi^{\ast}(v_0) & = \psi^{\ast}_1(v_1^\ast) + \psi^{\ast}_2(v_2^\ast);\\
        c^\ast(t;v_0) & = c^\ast_1(t;v_1^\ast) + c^\ast_2(t;v_2^\ast), \quad \forall\tin,
    \end{align*}
    where $v_1^\ast$ is determined by finding the solution $\bar{v}_1$ to:
    \begin{equation*}
        \mathcal{V}_1'(v_1; K_c) = \mathcal{V}_2'(v_0 - v_1; K_V).
    \end{equation*}
    and projecting $\bar{v}_1$ on the admissibility set $[v_1^{\min}, v_0 - v_2^{\min}]$, whereas $v_2^\ast$ is found via $v_2^\ast = v_0 - v_1^\ast$.
\end{enumerate}

In comparison to \cite{Lakner2006}, we added the opportunity of the decision maker to invest in a fixed-term asset to the investment universe, which makes UoC subproblem in Step 1 and UoW subproblem in Step 2 more general and more difficult to solve. The core idea of \cite{Lakner2006} to solve the UoC subproblem and UoW subproblem separately and then optimally combine the solutions is still valid.

In comparison to \cite{Desmettre2016}, we added to the investor's objective function also utility of intermediate consumption and we integrated risk-management constraints, i.e., lower bound on intermediate consumption and on terminal wealth. The generalized martingale approach proposed in \cite{Desmettre2016} has been adjusted to our setting. It means that UoC and UoW subproblems can be solved in two steps by deriving solutions conditional on $\psi_{k}F(T)$, $k \in \cBrackets{1, 2}$, and then maximizing the respective (conditional) value function. 

Our methodology for solving the UoW subproblem \eqref{OP:UoW_fixed_psi} in Step 2 extends the technique of \cite{Korn2005} by allowing the original utility in the wealth-constrained problem to be random due to the presence of a fixed-term asset.

\section{Explicit formulas for power utility functions}\label{sec:explicit_formulas_for_power_U}
In this section, we derive explicit formulas for an economic agent with power-utility functions and a fixed-term asset whose price process is either a deterministic one or a geometric Brownian motion. 

We assume that the agent's preferences are described by power-utility functions:
\begin{align}
    U_1(c) &= \frac{c^{p_1}}{p_1}; \label{eq:U_1_power_utility}\\
    U_2(v) &= \frac{v^{p_2}}{p_2}, \label{eq:U_2_power_utility}
\end{align}
where $p_1 \in (-\infty, 1) \setminus \{ 0\}$ and $p_2 \in (-\infty, 1) \setminus \{ 0\}$ are the relative-risk aversion (RRA) parameters. We also assume that the fixed-term asset price has the following dynamics:
\begin{equation}\label{eq:F_SDE_case_studies}
    d F(t) = F(t)\rBrackets{\mu_{F} dt + \sigma_{F}dW^{\Prob}(t)},\quad F(0) = F_0 > 0,
\end{equation}
where $\mu_F \in \mathbb{R}$, $\sigma_F \geq 0$. Trivially, \eqref{eq:F_SDE_case_studies}
satisfies the assumptions of $\mathcal{F}(T)$-measurability and $F(T) > 0$ $\Prob$-a.s..

The decision-making problem \eqref{OP:DS2016_with_V_c_lower_bounds} becomes then:
\begin{equation}\label{OP:DS2016_with_V_c_lower_bounds_power}
\begin{aligned}
\mathcal{V}(v_0; K_c, K_V) :=&\max_{\pi, \psi, c} \quad \EP\sBrackets{\int \limits_{t = 0}^{T} \frac{\rBrackets{c(t)}^{p_1}}{p_1}dt + \frac{\rBrackets{V^{v_0, (\pi, \psi, c)}(T)}^{p_2}}{p_2}}\\
& \text{ s.t.} \qquad {(\pi, \psi, c) \in \mathcal{A}(v_0; K_c, K_V)}
\end{aligned}
\end{equation}

As described in Section \ref{sec:methodology}, \eqref{OP:DS2016_with_V_c_lower_bounds} can be decomposed into two subproblems \eqref{OP:UoC_problem} and \eqref{OP:UoW_original}, which in the case of power utilities become
\begin{equation}\label{OP:UoC_subproblem_power}
        \mathcal{V}_1(v_1; K_c) := \max_{(\pi, \psi, c) \in \mathcal{A}_u(v_1)} \EP \sBrackets{\int \limits_{0}^{T} \frac{\rBrackets{c(t)}^{p_1}}{p_1}dt }\quad \text{s.t.}\quad c \in K_c,
\end{equation}
and
\begin{equation}\label{OP:UoW_subproblem_power}
        \mathcal{V}_2(v_2; K_V) := \max_{ (\pi, \psi, c) \in \mathcal{A}_u(v_2) } \EP \sBrackets{\frac{\rBrackets{V^{v_2, (\pi, \psi, c)}(T)}^{p_2}}{p_2} }\, \text{s.t.}\, V^{v_2, (\pi, \psi, c)}(T) \in K_V,
\end{equation}
respectively.

Even though the techniques for solving  \eqref{OP:UoC_subproblem_power} as well as \eqref{OP:UoW_subproblem_power} are the same when $\sigma_F > 0$ and when $\sigma_F = 0$, we distinguish between these cases for two reasons. First, when $\sigma_F = 0$, the (in general) stochastic nature of $F$ disappears and $F$ can be interpreted as another safe asset (e.g., a deposit in a bank whose default probability is negligible), which is illiquid but can have a higher rate of return $\mu_F$ than the risk-free rate $r$ of the liquid safe asset $S_0$. Second, the explicit formulas for $\sigma_F = 0$ are simpler, while the case $\sigma_F > 0$ requires further separation among three parametric regions, as we demonstrate later. Therefore, in Subsection \ref{subsec:explicit_formulas_for_power_U_stochastic_F} we provide explicit formulas for $F$ with a stochastic component (i.e., $\sigma_F > 0$) and in Subsection \ref{subsec:explicit_formulas_for_power_U_deterministic_F} we treat the case of deterministic $F$ (i.e., $\sigma_F = 0$).

\subsection{Fixed-term asset with a random component}\label{subsec:explicit_formulas_for_power_U_stochastic_F}
Throughout this subsection, we assume that $F$ is described by \eqref{eq:F_SDE_case_studies} with $\sigma_F > 0$. As mentioned in Section \ref{sec:PO_framework}, the fixed-term asset is non-redundant if:
\begin{equation}\label{ineq:general-non-redundancy-of-F}
    F(0) \leq \EP\sBrackets{\Ztilde(T) F(T)}.
\end{equation}
For $F$ with SDE \eqref{eq:F_SDE_case_studies}, we can show via  straightforward calculations that \eqref{ineq:general-non-redundancy-of-F} is equivalent to:
\begin{equation}\label{ineq:case-study-non-redundancy-of-F}
    \frac{\mu_F - r}{\sigma_F} \geq \frac{\mu - r}{\sigma} = \gamma.
\end{equation}

Condition \eqref{ineq:case-study-non-redundancy-of-F} means that the Sharpe ratio of the illiquid asset is larger than or equal to the Sharpe ratio of the liquid risky asset, which coincides with the liquid market price of risk. This condition is consistent with Assumption (8) in \cite{Ang2014}, where the researchers assume that the Sharpe ratio of the illiquid asset is greater than or equal to the Sharpe ratio of the liquid asset.

Before deriving the investment-consumption strategies optimal for \eqref{OP:DS2016_with_V_c_lower_bounds_power}, we state now two auxiliary lemmas that we will use in multiple proofs.

\begin{lemma}\label{lem:auxiliary_E_of_Z_T}
    Let $a \in (0, +\infty),\, b  \in (0, +\infty)$, $k \in \R$, and $0 \leq t < s \leq T$. Then:
    \begin{equation}
        \begin{aligned}
            \EP\sBrackets{\rBrackets{\Ztilde(s)}^k \mathbbm{1}_{\cBrackets{\Ztilde(s) \in (a, b)}}|\mathcal{F}(t)} & = \rBrackets{\Ztilde(t)}^k \exp \rBrackets{-k\rBrackets{r + \frac{\gamma^2}{2}}(s - t) + \frac{k^2 \gamma^2}{2}(s - t)} \\
            & \quad \cdot \rBrackets{\Phi\rBrackets{\dbar(a) + k \gamma \sqrt{s - t}} - \Phi\rBrackets{\dbar(b) + k \gamma \sqrt{s - t}}} \cdot \mathbbm{1}_{\cBrackets{a < b}}\\
            & =: \xi(s, t, \Ztilde(t), k, a, b; \gamma, r),
        \end{aligned}
    \end{equation}
    where
    \begin{equation}\label{eq:dbar}
    \begin{aligned}
        \dbar(x) &:= \dbar(x; s, t, \Ztilde(t), r, \gamma) := \frac{\ln \rBrackets{\Ztilde(t) / x} - \rBrackets{r + \frac{\gamma^2}{2}}(s - t)}{\gamma \sqrt{s - t}}; \\
        \dbar(0) &:= \lim_{x \downarrow 0} \dbar(x) = +\infty \quad \text{and}\quad \dbar(+\infty):= \lim_{x \uparrow +\infty} \dbar(x) = -\infty,
    \end{aligned}
\end{equation}
and $\Phi\rBrackets{\cdot}$ is the cumulative distribution function of a standard normal random variable.
\end{lemma}

\begin{lemma}\label{lem:F_Ztilde_relation}
    Let $t_1 \in \R, t_2 \in \R$ be such that $0 \leq t_1 \leq t_2 \leq T$. Then:
    \begin{equation} \label{eq:link_btw_F_and_Ztilde}
    F(t_2) = F(t_1) h(t_1, t_2) \rBrackets{\frac{\Ztilde(t_2)}{\Ztilde(t_1)}}^{-\frac{\sigma_F}{\gamma}},
\end{equation}
    where
    \begin{equation}\label{eq:h_t_1_t_2}
        h(t_1, t_2) := \exp \rBrackets{\rBrackets{\mu_F - \frac{1}{2}\sigma_F^2 - \frac{\sigma_F r}{\gamma} - \frac{1}{2}\sigma_F \gamma} (t_2 - t_1)}.
\end{equation}
\end{lemma}

Applying Proposition \ref{prop:UoC_problem_solution}, we can solve the UoC subproblem \eqref{OP:UoC_subproblem_power}. 

\begin{corollary}[Solution to \eqref{OP:UoC_subproblem_power}]\label{cor:UoC_problem_solution_power}
    Consider Problem \eqref{OP:UoC_problem} with initial capital $v_1$.
    
    If $v_1 < v_1^{\min}$, then \eqref{OP:UoC_subproblem_power} does not admit a solution.
    
    If $v_1 = v_1^{\min}$, then $\pi_1^\ast(t;v_1) = 0$, $\psiStar_1 = 0$, $c_1^\ast(t;v_1) = \cunder$, $X_1^\ast(t;v_1) = \cunder \rBrackets{1 - \exp\rBrackets{ - r (T - t)}} / r$.
    
    If $v_1 > v_1^{\min}$, then:
    \begin{enumerate}
        \item[(i)] the optimal consumption process is given by
        \begin{equation}\label{eq:UoC_cStar_power}
            c_1^\ast(t;v_1) = \rBrackets{\lambda^\ast_1 \Ztilde(t)}^{\frac{1}{p_1 - 1}}\mathbbm{1}_{\cBrackets{ \Ztilde(t) \leq \cunder^{p_1 - 1} / \lambda^\ast_1}} + \cunder\mathbbm{1}_{\cBrackets{ \Ztilde(t) > \cunder^{p_1 - 1} / \lambda^\ast_1}},\,\tin,
        \end{equation}
        where $\lambda^\ast_1 := \lambda^\ast_1(v_1)$ is the unique solution to the equation:
        \begin{equation}\label{eq:UoC_lambdaStar_power}
        \begin{aligned}
            v_1 = & \lambda_1^{\frac{1}{p_1 - 1}} \intzeroT \xi\rBrackets{t, 0, 1, \frac{p_1}{p_1 - 1}, 0, \frac{\cunder^{p_1 - 1}}{\lambda_1}; \gamma, r}\, dt \\
            & + \cunder  \intzeroT \xi\rBrackets{t, 0, 1, 0, \frac{\cunder^{p_1 - 1}}{\lambda_1}, +\infty; \gamma, r}\, dt;
        \end{aligned}
        \end{equation}
        \item[(ii)] the optimal position in the fixed-term asset is given by:
        \begin{equation}\label{eq:UoC_psiStar_power}
            \psi_1^\ast = 0;
        \end{equation}
        \item[(iii)] the value of the optimal liquid portfolio is given by:
        \begin{equation}\label{eq:UoC_XStar_power}
            \begin{aligned}
                X_1^\ast(t;v_1) = & \rBrackets{\lambda_1^\ast}^{\frac{1}{p_1 - 1}} \rBrackets{\Ztilde(t)}^{-1} \intttoT \xi\rBrackets{s, t, \Ztilde(t), \frac{p_1}{p_1 - 1}, 0, \frac{\cunder^{p_1 - 1}}{\lambda_1^\ast}; \gamma, r}\, ds \\
                & + \cunder \rBrackets{\Ztilde(t)}^{-1} \intttoT \xi\rBrackets{s, t, \Ztilde(t), 0, \frac{\cunder^{p_1 - 1}}{\lambda_1^\ast}, +\infty; \gamma, r}\, ds;
            \end{aligned}
        \end{equation}
        and the terminal value of the optimal total portfolio is given by
        \begin{equation}\label{eq:UoC_VStar_power}
            V^\ast_1(T; v_1) = 0;
        \end{equation}
        \item[(iv)] the optimal relative portfolio process exists and is given by
        \begin{equation}\label{eq:UoC_piStar_power}
            \pi_1^\ast(t;v_1) = -\frac{\gamma \Ztilde(t)}{\sigma g_1(t, \Ztilde(t))} \frac{\partial g_1(t, \Ztilde(t))}{\partial \widetilde{z}},
        \end{equation}
        where $g_1(t, \Ztilde(t))$ is the right-hand side in \eqref{eq:UoC_XStar_power} seen as a function of $t$ and $\Ztilde(t)$.
    \end{enumerate}
\end{corollary}

\begin{proposition}\label{prop:put_on_psiBar_F}
    Let $\psibar_2 \geq 0$ be arbitrarily fixed and $F$ be given by \eqref{eq:F_SDE_case_studies}. Consider a put option $B := (\Vunder - \psibar_2 F(T))^{+}$ with $F$ given by \eqref{eq:F_SDE_case_studies}. Then:
    \begin{enumerate}
        \item[(i)] the initial capital needed to replicate $B$ by trading $S_0$ and $S_1$ equals:
        \begin{equation}\label{eq:X_0_B_explicit}
            \begin{aligned}
                x_B = x_B(\psibar_2) &= \Vunder \xi \rBrackets{T, 0, 1, 1, \rBrackets{\psibar_2 F_0 h(0, T) / \Vunder}^{\frac{\gamma}{\sigma_F}}, +\infty; \gamma, r} \\
                & \quad - \psibar F_0 h(0, T) \xi \rBrackets{T, 0, 1, 1 - \frac{\sigma_F}{\gamma}, \rBrackets{\psibar_2 F_0 h(0, T) / \Vunder}^{\frac{\gamma}{\sigma_F}}, +\infty; \gamma, r};
            \end{aligned}
        \end{equation}
        \item[(ii)] the value process of the portfolio that replicates  is given by:
        \begin{align}
                X_B\rBrackets{t;x_B} &= \Vunder \rBrackets{\Ztildet{t}}^{-1} \xi \rBrackets{T, t, \Ztildet{t}, 1, \Ztildet{t}\rBrackets{\psibar_2 F(t) h(t, T) / \Vunder}^{\frac{\gamma}{\sigma_F}}, +\infty; \gamma, r}  \label{eq:X_t_B_explicit} \\
                & \quad -  \psibar_2 F(t) h(t, T) \rBrackets{\Ztildet{t}}^{\frac{\sigma_F}{\gamma}-1} \xi \rBrackets{T, t, \Ztildet{t}, 1 - \frac{\sigma_F}{\gamma}, \Ztildet{t}\rBrackets{\psibar_2 F(t) h(t, T) / \Vunder}^{\frac{\gamma}{\sigma_F}}, +\infty; \gamma, r};\notag
        \end{align}
        \item[(iii)] the investment strategy that replicates $B$ is given by:
        \begin{equation} \label{eq:pi_B_explicit}
            \pi_B(t; x_B) = -\frac{\gamma \Ztilde(t)}{\sigma g_B(t, \Ztilde(t))} \frac{\partial g_B(t, \Ztilde(t))}{\partial \widetilde{z}},
        \end{equation}
        where $g_B(t, \Ztilde(t))$ is the right-hand side in \eqref{eq:X_t_B_explicit} represented as a seen of $t$ and $\Ztilde(t)$.
    \end{enumerate}
\end{proposition}

Applying Proposition \ref{prop:UoW_problem_solution}, we can solve the UoW subproblem \eqref{OP:UoW_subproblem_power}. In the calculations, it turns out that the relation among the parameters $\sigma_F$, $\gamma$ and $p_2$ has a subtle yet structurally important impact, which is why we define:
\begin{equation}\label{eq:D_defining_cases}
    D(\gamma, \sigma_F, p_2):= 1 + \sigma_F (p_2 - 1) / \gamma 
\end{equation}
and distinguish among the following three cases (parametric regions):
\begin{equation}\label{def_parametric_cases}
\begin{aligned}
     \text{Case I:} \quad &  D(\gamma, \sigma_F, p_2) > 0;\\
     \text{Case II:} \quad & D(\gamma, \sigma_F, p_2)  = 0;\\
     \text{Case III:} \quad & D(\gamma, \sigma_F, p_2) < 0.
\end{aligned}
\end{equation}

Since $\gamma > 0$, $ \sigma_F \geq 0$ and $(1 - p_2) > 0$, the above cases depend on the relation between the RRA-adjusted market price of risk in the liquid market and the volatility of the fixed-term asset:
\begin{equation*}
    \frac{\gamma}{1 - p_2} \stackrel[<]{>}{=} \sigma_F.
\end{equation*}

\begin{corollary}[Solution to \eqref{OP:UoW_subproblem_power}]\label{cor:UoW_problem_solution_power}
    Consider \eqref{OP:UoW_subproblem_power} with initial capital $v_2$.
    
    If $v_2 < v_2^{\min}$, then \eqref{OP:UoW_subproblem_power} does not admit a solution.

    If $v_2 = v_2^{\min}$, then $\pi_2^\ast(t;v_2)=0$, $\psi_2^\ast = 0$, $c_2^\ast(t;v_2) = 0$, $X_2^\ast(t;v_2) = \exp\rBrackets{ -r (T - t)}\Vunder = v_2 \exp\rBrackets{r t}$ for $\forall \tin$, and $V^\ast_2(T; v_2) = \Vunder$.

    If $v_2 > v_2^{\min}$, then:
    \begin{enumerate}
        \item[(i)] the optimal consumption process is given by
        \begin{equation}\label{eq:UoW_cStar_power}
            c_2^\ast(t;v_2) = 0,\,\forall\tin;
        \end{equation}
        \item[(ii)] the optimal position in the fixed-term asset is given by:
        \begin{equation}\label{eq:UoW_psiStar_power}
            \psi_2^\ast := \psi_2^\ast(v_2) = \argmax_{\psibar_2 \in \sBrackets{0, \frac{v_2}{F_0}}} \widetilde{\mathcal{V}}_2\rBrackets{\xtil_2\rBrackets{\psibar_2}},
        \end{equation}
        where $\xtil_2(\psibar_2) = v_2 - \psibar_2 F_0 - x_B\rBrackets{\psibar_2}$, $x_B\rBrackets{\psibar_2}$ is given by \eqref{eq:X_0_B_explicit}, and
        \begin{equation} \label{eq:conditional_VF_2}
        \widetilde{\mathcal{V}}_2\rBrackets{\xtil_2(\psibar_2)} =
        \left\{
            \begin{aligned}
                 & \frac{1}{p_2}\Vunder^{p_2} \xi \Bigl( T, 0, 1, 0, \Vunder^{p_2 -1} / \widetilde{\lambda}^\ast_2, +\infty; \gamma, r \Bigr) \\
                & \qquad + \frac{1}{p_2}\rBrackets{\widetilde{\lambda}^\ast_2}^{\frac{p_2}{p_2 - 1}} \xi \Bigl( T, 0, 1, \frac{p_2}{p_2 - 1}, 0,\Vunder^{p_2 -1} / \widetilde{\lambda}^\ast_2; \gamma, r \Bigr),\quad  \text{if } \psibar_2 = 0; \\
                & \widetilde{\mathcal{V}}_{2}^{I}\rBrackets{\xtil_2(\psibar_2)}, \qquad \text{if } \psibar_2 \in \left(0, \frac{v_2}{F_0} \right] \text{ and } D(\gamma, \sigma_F, p_2) > 0; \\
                & \widetilde{\mathcal{V}}_{2}^{II}\rBrackets{\xtil_2(\psibar_2)}, \qquad \text{if } \psibar_2 \in \left(0, \frac{v_2}{F_0} \right] \text{ and } D(\gamma, \sigma_F, p_2) = 0; \\
                & \widetilde{\mathcal{V}}_{2}^{III}\rBrackets{\xtil_2(\psibar_2)}, \qquad \text{if } \psibar_2 \in \left(0, \frac{v_2}{F_0} \right] \text{ and } D(\gamma, \sigma_F, p_2) < 0;
            \end{aligned} 
            \right.
        \end{equation}
        with $D(\cdot)$ defined in \eqref{eq:D_defining_cases} and  
        \begin{align*}
            \text{Case I:}\, \widetilde{\mathcal{V}}_2^{I}(\xtil_2(\psibar_2)) = & \frac{1}{p_2}\Vunder^{p_2} \xi \Bigl( T, 0, 1, 0, \max \Bigl\{ \rBrackets{{\psibar_2 F_0 h(0, T)}/{\Vunder}}^{\frac{\gamma}{\sigma_F}}, \Vunder^{p_2 -1} / \widetilde{\lambda}^\ast_2 \Bigr\}, +\infty; \gamma, r \Bigr) \\
            & + \frac{1}{p_2} \rBrackets{\psibar_2 F_0 h(0, T)}^{p_2} \xi \biggl(T, 0, 1, -\frac{\sigma_F p_2}{\gamma},\\
            & \rBrackets{\rBrackets{\psibar_2 F_0 h(0, T)}^{p_2 - 1} / \widetilde{\lambda}^\ast_2}^{1 / D(\gamma,\sigma_F, p_2)}, \rBrackets{\psibar_2 F_0 h(0, T) / \Vunder}^{\frac{\gamma}{\sigma_F}}; \gamma, r \biggr) \\
            & + \frac{1}{p_2} \rBrackets{\widetilde{\lambda}^\ast_2}^{\frac{p_2}{p_2 - 1}} \xi \biggl(T, 0, 1, \frac{p_2}{p_2 - 1}, 0, \\
            & \min \biggl\{ \Vunder^{p_2 -1} / \widetilde{\lambda}^\ast_2, \rBrackets{\rBrackets{\psibar_2 F_0 h(0, T)}^{p_2 - 1} / \widetilde{\lambda}^\ast_2}^{1 / D(\gamma,\sigma_F, p_2)} \biggr\}; \gamma, r \biggr);
        \end{align*}
        \begin{align*}
            \text{Case II:}\, \widetilde{\mathcal{V}}_{2}^{II}(\xtil_2\rBrackets{\psibar_2}) = & \frac{1}{p_2}\Vunder^{p_2} \xi \Bigl( T, 0, 1, 0, \max \Bigl\{ \rBrackets{{\psibar_2 F_0 h(0, T)}/{\Vunder}}^{\frac{\gamma}{\sigma_F}}, \Vunder^{p_2 -1} / \widetilde{\lambda}^\ast_2 \Bigr\}, +\infty; \gamma, r \Bigr) \\\\
            & + \frac{1}{p_2} \rBrackets{\psibar_2 F_0 h(0, T)}^{p_2} \xi \biggl(T, 0, 1, -\frac{\sigma_F p_2}{\gamma}, 0, \rBrackets{\psibar_2 F_0 h(0, T) / \Vunder}^{\frac{\gamma}{\sigma_F}}; \gamma, r \biggr)\\
            & \quad \cdot \mathbbm{1}_{\cBrackets{\widetilde{\lambda}^\ast_2 \geq \rBrackets{\psibar_2 F_0 h(0, T)}^{p_2 - 1}}} \\
            & + \frac{1}{p_2} \rBrackets{\widetilde{\lambda}^\ast_2}^{\frac{p_2}{p_2 - 1}} \xi \biggl(T, 0, 1, \frac{p_2}{p_2 - 1}, 0, \Vunder^{p_2 -1} / \widetilde{\lambda}^\ast_2 ; \gamma, r \biggr) \\
            & \quad  \cdot \mathbbm{1}_{\cBrackets{\widetilde{\lambda}^\ast_2 < \rBrackets{\psibar_2 F_0 h(0, T)}^{p_2 - 1}}};
        \end{align*}
        \begin{align*}
            \text{Case III:}\, \, \widetilde{\mathcal{V}}_{2}^{III}(\xtil_2(\psibar_2)) = & \frac{1}{p_2}\Vunder^{p_2} \xi \Bigl( T, 0, 1, 0, \max \Bigl\{ \rBrackets{{\psibar_2 F_0 h(0, T)}/{\Vunder}}^{\frac{\gamma}{\sigma_F}}, \Vunder^{p_2 -1} / \widetilde{\lambda}^\ast_2 \Bigr\}, +\infty; \gamma, r \Bigr) \\\\
            & + \frac{1}{p_2} \rBrackets{\psibar_2 F_0 h(0, T)}^{p_2} \xi \biggl(T, 0, 1, -\frac{\sigma_F p_2}{\gamma}, 0, \\
            & \min \biggl\{ \rBrackets{\rBrackets{\psibar_2 F_0 h(0, T)}^{p_2 - 1} / \widetilde{\lambda}^\ast_2}^{1 / D(\gamma,\sigma_F, p_2)}, \rBrackets{\psibar_2 F_0 h(0, T) / \Vunder}^{\frac{\gamma}{\sigma_F}} \biggr\}; \gamma, r \biggr) \\
            & + \frac{1}{p_2} \rBrackets{\widetilde{\lambda}^\ast_2}^{\frac{p_2}{p_2 - 1}} \xi \biggl(T, 0, 1, \frac{p_2}{p_2 - 1}, \\
            & \rBrackets{\rBrackets{\psibar_2 F_0 h(0, T)}^{p_2 - 1} / \widetilde{\lambda}^\ast_2}^{1 / D(\gamma,\sigma_F, p_2)}, \Vunder^{p_2 -1} / \widetilde{\lambda}^\ast_2 ; \gamma, r \biggr),
        \end{align*}
        $h(0, T)$ is defined in \eqref{eq:h_t_1_t_2}, and $\widetilde{\lambda}_2^\ast$ is the unique solution to the following equation:
        \begin{equation}\label{eq:UoW_auxiliary_unconstrained_problem_lambda_power}
            \EP\sBrackets{\Ztilde(T) \Xtilast_2(T; \xtil_2(\psibar_2))}  = \xtil_2(\psibar_2),
        \end{equation}
        where $\Xtilast_2(T; \xtil_2)$ is defined in Point (iii) below;
        \item[(iii)] the value of the optimal liquid portfolio at time $\tin$ is given by
        \begin{equation}\label{eq:UoW_XStar_power}
            X_2^\ast(t;v_2) = X_B(t; x_B\rBrackets{\psi_2^\ast}) + \Xtilast_2(t; \xtil_2\rBrackets{\psi_2^\ast}),\,\forall\tin,
        \end{equation}
        and the terminal value of the optimal total portfolio is given by
        \begin{equation}\label{eq:UoW_VStar_power}
            V^\ast_2(T; v_2) = X_2^\ast(T;v_2) + \psi_2^\ast F(T) = \max \cBrackets{\Vunder, \psi_2^\ast F(T), \rBrackets{ \widetilde{\lambda}^\ast_2 \Ztildet{T} }^{\frac{1}{p_2 - 1}} },
        \end{equation}
        where $X_B(t; x_B\rBrackets{\psiStar_2})$ is given by \eqref{eq:X_t_B_explicit} and $\Xtilast_2(t; \xtil_2\rBrackets{\psi_2^\ast})$ is given by:
        \begin{equation}\label{eq:X_tilde_star_2_t_power}
        \Xtilast_2(t; \xtil_2\rBrackets{\psi_2^\ast}) =
        \left\{
            \begin{aligned}
                & \rBrackets{\widetilde{\lambda}^\ast_2}^{\frac{1}{p_2 - 1}} \rBrackets{\Ztildet{t}}^{-1}\xi \biggl( T, t, \Ztilde(t), \frac{p_2}{p_2 - 1}, 0, \Vunder^{p_2 -1} / \widetilde{\lambda}^\ast_2; \gamma, r \biggr) \\
                & \quad - \Vunder \rBrackets{\Ztildet{t}}^{-1} \xi \biggl( T, t, \Ztilde(t), 1, 0, \Vunder^{p_2 -1} / \widetilde{\lambda}^\ast_2; \gamma, r \biggr),\quad  \text{if } \psibar_2 = 0; \\
                & \widetilde{X}^{I,\ast}_{2}(t; \xtil_2\rBrackets{\psi_2^\ast}) , \qquad \text{if } \psibar_2 \in \left(0, \frac{v_2}{F_0} \right] \text{ and } D(\gamma, \sigma_F, p_2) > 0; \\
                & \widetilde{X}^{II,\ast}_{2}(t; \xtil_2\rBrackets{\psi_2^\ast}), \qquad \text{if } \psibar_2 \in \left(0, \frac{v_2}{F_0} \right] \text{ and } D(\gamma, \sigma_F, p_2) = 0; \\
                & \widetilde{X}^{III,\ast}_{2}(t; \xtil_2\rBrackets{\psi_2^\ast}) , \qquad \text{if } \psibar_2 \in \left(0, \frac{v_2}{F_0} \right] \text{ and } D(\gamma, \sigma_F, p_2) < 0; \\
        \end{aligned}
        \right.
        \end{equation}
        with
        \begin{align*}
            \text{Case I:}\, &\widetilde{X}^{I,\ast}_{2}(t; \xtil_2\rBrackets{\psi_2^\ast}) =  \rBrackets{\widetilde{\lambda}^\ast_2}^{\frac{1}{p_2 - 1}} \rBrackets{\Ztildet{t}}^{-1}\xi \Biggl( T, t, \Ztilde(t), \frac{p_2}{p_2 - 1}, 0,  \\
            &\quad \min \Biggl\{ \Vunder^{p_2 -1} / \widetilde{\lambda}^\ast_2, \rBrackets{ { \rBrackets{ \psiStar_2 F(t) h(t,T) \rBrackets{\Ztilde(t)}^{\frac{\sigma_F}{\gamma}} }^{p_2 - 1} } / {\widetilde{\lambda}^\ast_2}}^{\frac{1}{D(\gamma, \sigma_F, p_2)}} \Biggr\}; \gamma, r \Biggr)  \\
            & - \Vunder \rBrackets{\Ztildet{t}}^{-1} \xi \Biggl( T, t, \Ztilde(t), 1, \Ztilde(t) \rBrackets{{\psiStar_2 F(t) h(t, T)}/{\Vunder}}^{\frac{\gamma}{\sigma_F}}, \\
            & \quad \min \Biggl\{ \Vunder^{p_2 -1} / \widetilde{\lambda}^\ast_2, \rBrackets{ { \rBrackets{ \psiStar_2 F(t) h(t,T) \rBrackets{\Ztilde(t)}^{\frac{\sigma_F}{\gamma}} }^{p_2 - 1} } / {\widetilde{\lambda}^\ast_2}}^{\frac{1}{D(\gamma, \sigma_F, p_2)}} \Biggr\}, \gamma, r \Biggr) \\
            & - \psiStar_2 F(t)  h(t, T) \rBrackets{\Ztilde(t)}^{\frac{\sigma_F}{\gamma} - 1}   \xi \Biggl( T, t, \Ztilde(t),1 - \frac{\sigma_F}{\gamma}, 0,  \\
            &\quad \min \Biggl\{ \Ztilde(t) \rBrackets{{\psiStar_2 F(t) h(t, T)}/{\Vunder}}^{\frac{\gamma}{\sigma_F}}, \Vunder^{p_2 -1} / \widetilde{\lambda}^\ast_2, \rBrackets{ { \rBrackets{ \psiStar_2 F(t) h(t,T) \rBrackets{\Ztilde(t)}^{\frac{\sigma_F}{\gamma}} }^{p_2 - 1} } / {\widetilde{\lambda}^\ast_2}}^{\frac{1}{D(\gamma, \sigma_F, p_2)}}\Biggr\}; \gamma, r \Biggr);
        \end{align*}
        \begin{align*}
            \text{Case II:}\, &\widetilde{X}^{II,\ast}_{2}(t; \xtil_2\rBrackets{\psi_2^\ast}) =  \Biggl( \rBrackets{\widetilde{\lambda}^\ast_2}^{\frac{1}{p_2 - 1}} \rBrackets{\Ztildet{t}}^{-1}\xi \biggl( T, t, \Ztilde(t), \frac{p_2}{p_2 - 1}, 0,  \Vunder^{p_2 -1} / \widetilde{\lambda}^\ast_2; \gamma, r \biggr) \\
            & - \Vunder \rBrackets{\Ztildet{t}}^{-1} \xi \biggl( T, t, \Ztilde(t), 1, \Ztilde(t) \rBrackets{{\psiStar_2 F(t) h(t, T)}/{\Vunder}}^{\frac{\gamma}{\sigma_F}}, \Vunder^{p_2 -1} / \widetilde{\lambda}^\ast_2; \gamma, r \biggr) \\
            & - \psiStar_2 F(t) h(t, T) \rBrackets{\Ztilde(t)}^{\frac{\sigma_F}{\gamma} - 1}  \xi \Biggl( T, t, \Ztilde(t),1 - \frac{\sigma_F}{\gamma}, 0,  \\
            &\quad \min \Biggl\{ \Ztilde(t) \rBrackets{{\psiStar_2 F(t) h(t, T)}/{\Vunder}}^{\frac{\gamma}{\sigma_F}}, \Vunder^{p_2 -1} / \widetilde{\lambda}^\ast_2\Biggr\}; \gamma, r \Biggr) \Biggr) \mathbbm{1}_{\cBrackets{ \widetilde{\lambda}^\ast_2 \leq \rBrackets{ \psiStar_2 F(t) h(t,T) \rBrackets{\Ztilde(t)}^{\frac{\sigma_F}{\gamma}} }^{p_2 - 1} } } ;   
        \end{align*}
        \begin{align*}
            \text{Case III:}\, &\widetilde{X}_2^{III, \ast}(t; \xtil_2\rBrackets{\psi_2^\ast}) =  \rBrackets{\widetilde{\lambda}^\ast_2}^{\frac{1}{p_2 - 1}} \rBrackets{\Ztildet{t}}^{-1}\xi \Biggl( T, t, \Ztilde(t), \frac{p_2}{p_2 - 1},  \\
            &\quad  \rBrackets{ { \rBrackets{ \psiStar_2 F(t) h(t,T) \rBrackets{\Ztilde(t)}^{\frac{\sigma_F}{\gamma}} }^{p_2 - 1} } / {\widetilde{\lambda}^\ast_2}}^{\frac{1}{D(\gamma, \sigma_F, p_2)}}, \Vunder^{p_2 -1} / \widetilde{\lambda}^\ast_2 ; \gamma, r \Biggr)  \\
            & - \Vunder \rBrackets{\Ztildet{t}}^{-1} \xi \Biggl( T, t, \Ztilde(t), 1, \max \Biggl\{\Ztilde(t) \rBrackets{{\psiStar_2 F(t) h(t, T)}/{\Vunder}}^{\frac{\gamma}{\sigma_F}}, \\
            & \quad \quad \rBrackets{ { \rBrackets{ \psiStar_2 F(t) h(t,T) \rBrackets{\Ztilde(t)}^{\frac{\sigma_F}{\gamma}} }^{p_2 - 1} } / {\widetilde{\lambda}^\ast_2}}^{\frac{1}{D(\gamma, \sigma_F, p_2)}} \Biggr\}, \Vunder^{p_2 -1} / \widetilde{\lambda}^\ast_2 , \gamma, r \Biggr) \\
            & - \psiStar_2 F(t)  h(t, T) \rBrackets{\Ztilde(t)}^{\frac{\sigma_F}{\gamma} - 1}  \xi \Biggl( T, t, \Ztilde(t),1 - \frac{\sigma_F}{\gamma}, \rBrackets{ { \rBrackets{ \psiStar_2 F(t) h(t,T) \rBrackets{\Ztilde(t)}^{\frac{\sigma_F}{\gamma}} }^{p_2 - 1} } / {\widetilde{\lambda}^\ast_2}}^{\frac{1}{D(\gamma, \sigma_F, p_2)}},  \\
            &\quad \min \Biggl\{  \Ztilde(t) \rBrackets{{\psiStar_2 F(t) h(t, T)}/{\Vunder}}^{\frac{\gamma}{\sigma_F}} ,\Vunder^{p_2 -1} / \widetilde{\lambda}^\ast_2 \Biggr\}; \gamma, r \Biggr);
        \end{align*}
        \item[(iv)] the optimal relative portfolio process is given by
        \begin{equation}\label{eq:UoW_piStar_power}
            \pi_2^\ast(t;v_2) = \frac{\pi_B(t; x_B\rBrackets{\psi_2^\ast}) X_B(t;x_B\rBrackets{\psi_2^\ast}) + \pitilast_2(t;\xtil_2\rBrackets{\psi_2^\ast}) \Xtilast_2\rBrackets{t;\xtil_2\rBrackets{\psi^\ast_2}}}{X_2^\ast(t;v_2)},
        \end{equation}
        where $\pi_B(t; x_B\rBrackets{\psi_2^\ast})$ is given by \eqref{eq:pi_B_explicit} and
        \begin{equation}\label{eq:UoW_piTildeStar_power}
            \pitilast_2(t;\xtil_2\rBrackets{\psi_2^\ast}) = -\frac{\gamma \Ztilde(t)}{\sigma g_2(t, \Ztilde(t))} \frac{\partial g_2(t, \Ztilde(t))}{\partial \widetilde{z}},
        \end{equation}
        where $g_2(t, \Ztilde(t))$ is the right-hand side in \eqref{eq:X_tilde_star_2_t_power} represented as a function of $t$ and $\Ztilde(t)$. 
    \end{enumerate}
\end{corollary}

Finally, we apply Proposition \ref{prop:solution_to_DS2016_with_V_c_lower_bounds} to obtain the solution to \eqref{OP:DS2016_with_V_c_lower_bounds_power}, which is summarized in the following corollary.

\begin{corollary}[Solution to Problem \eqref{OP:DS2016_with_V_c_lower_bounds_power}]\label{cor:merging_subproblems_power_U}
    Assume that $v_0 \geq v_0^{\min}:= v_1^{\min} + v_2^{\min}$. Then in Problem \eqref{OP:DS2016_with_V_c_lower_bounds_power}:
    \begin{enumerate}
        \item[(i)] the value function is given by
        \begin{equation}\label{eq:value_function_power}
            \mathcal{V}(v_0;K_c, K_V) = \mathcal{V}_1(v_1^\ast; K_c) + \mathcal{V}_2(v_2^\ast; K_V),
        \end{equation}
        where $\mathcal{V}_1(v_1^\ast; K_c)$ is given by:
        \begin{equation}\label{eq:UoC_value_function_power}
            \begin{aligned}
                \mathcal{V}_1(v_1^\ast;K_c) =& \frac{\rBrackets{\lambda_1^\ast}^{\frac{p_1}{p_1 - 1}}}{p_1} 
                \intzeroT \xi\rBrackets{t, 0, 1, \frac{p_1}{p_1 - 1}, 0, \frac{\cunder^{p_1 - 1}}{\lambda_1^\ast}; \gamma, r}\, dt \\
                    & + \frac{\cunder^{p_1}}{p_1}  \intzeroT \xi\rBrackets{t, 0, 1, 0, \frac{\cunder^{p_1 - 1}}{\lambda_1^\ast}, +\infty; \gamma, r}\, dt,
            \end{aligned}
        \end{equation}
        $\mathcal{V}_2(v_2^\ast; K_V) = \widetilde{\mathcal{V}}_2\rBrackets{v_2^{\ast} - \psiStar_2 F_0 - x_B\rBrackets{\psiStar_2}}$ for $\widetilde{\mathcal{V}}_2$ given by \eqref{eq:conditional_VF_2}, $x_B$ defined in \eqref{eq:X_0_B_explicit}, and $v_1^\ast$ as well as $v_2^\ast$ calculated as per Lemma \ref{lem:vStar_1_vStar_2};
        \item[(ii)] the value of the optimal liquid portfolio at $\tin$ is given by
        \begin{equation}\label{eq:XStar_t_power}
            X^\ast(t; v_0) = X_1^\ast(t;v^{\ast}_1) + X_2^{\ast}(t;v^{\ast}_2), \quad \forall\tin,
        \end{equation}
        where $X_1^\ast(t;v^{\ast}_1)$ is given by \eqref{eq:UoC_XStar_power} and $X_2^{\ast}(t;v^{\ast}_2)$ by \eqref{eq:UoW_XStar_power};
        \item[(iii)] the optimal consumption rate at $\tin$ is given by
        \begin{equation}\label{eq:cStar_power}
            c^\ast(t; v_0) = c^\ast_1(t;v_1^\ast), \quad \forall\tin
        \end{equation}
        with $c^\ast_1(t;v_1^\ast)$ be given by \eqref{eq:UoC_cStar_power};
        \item[(iv)] the optimal position in the fixed-term asset is given by
        \begin{equation}\label{eq:psiStar_power}
            \psi^{\ast} = \psi^{\ast}_2(v_2^\ast);
        \end{equation}
        with $\psi^{\ast}_2(v_2^\ast)$ be given by \eqref{eq:UoW_psiStar_power};
        \item[(v)] the optimal portfolio strategy at $\tin$ is given by
        \begin{equation}\label{eq:piStar_power}
            \pi^{\ast}(t; v_0) = \frac{\pi^{\ast}_1(t;v_1^\ast) \cdot X^{\ast}_1(t;v^\ast_1) + \pi^{\ast}_2(t;v_2^\ast)\cdot X^{\ast}_2(t;v^\ast_2)}{X^{\ast}(t; v_0)}, \quad \forall\tin,
        \end{equation}
        where $\pi^{\ast}_1(t;v_1^\ast)$ is given by \eqref{eq:UoC_piStar_power} and $\pi^{\ast}_2(t;v_2^\ast)$ is given by \eqref{eq:UoW_piStar_power};
        \item[(vi)] the terminal value of the optimal total portfolio equals:
        \begin{equation}\label{eq:VStar_T_power}
            V^\ast(T; v_0) = \max \cBrackets{\Vunder, \psi_2^\ast F(T), \rBrackets{ \widetilde{\lambda}^\ast_2 \Ztildet{T} }^{\frac{1}{p_2 - 1}} },
        \end{equation}
        where $\widetilde{\lambda}^\ast_2$ is the unique solution to \eqref{eq:UoW_auxiliary_unconstrained_problem_lambda_power} with $\xtil_2(\psiStar_2) = v_2^\ast - x_B(\psiStar_2) - \psiStar_2 F_0$.
    \end{enumerate}
\end{corollary}

\subsection{Deterministic fixed-term asset}\label{subsec:explicit_formulas_for_power_U_deterministic_F}

In this subsection we assume that $\sigma_F = 0$ in \eqref{eq:F_SDE_case_studies}, which implies that $F(T) = F(t)\exp \rBrackets{\mu_F (T - t)}$ for any $\tin$. In this case, the non-redundancy condition \eqref{ineq:general-non-redundancy-of-F} is equivalent to the condition $\mu_F \geq r$.

For $\sigma_F = 0$, Corollary \ref{cor:UoC_problem_solution_power} and Corollary \ref{cor:merging_subproblems_power_U} remain valid and unchanged, whereas Proposition \ref{prop:put_on_psiBar_F} and Corollary \ref{cor:UoW_problem_solution_power} can be simplified due to the deterministic nature of $F$. Equations \eqref{eq:X_0_B_explicit}--\eqref{eq:pi_B_explicit} for pricing and replicating the put option with payoff $(\Vunder - \psibar_2 F(T))^{+}$ become respectively:
\begin{align*}
    x_B &= \exp\rBrackets{-rT}\Bigl(\Vunder - \psibar_2 F(0)\exp \rBrackets{\mu_F T}\Bigr)^{+};\\
    X_B(t; x_B) &=  \exp\rBrackets{-r(T - t)} \Bigl( \Vunder - \psibar_2 F(t) \exp\rBrackets{\mu_F(T - t)} \Bigr)^{+},\quad \forall \tin; \\
    \pi_B(t; x_B) & = 0,\quad \forall \tin.
\end{align*}

Corollary \ref{cor:UoW_problem_solution_power} is simpler for $\sigma_F = 0$ because we have only Case I due to $D(\gamma, 0, p_2) \stackrel{\eqref{eq:D_defining_cases}}{=} 1$ for any $\gamma > 0$ and $p_2 \in (-\infty, 1) \setminus \cBrackets{0}$. Via straightforward calculations involving Lemma \ref{lem:F_Ztilde_relation}, which is still valid for $\sigma_F = 0$, we obtain the conditional value function (counterpart of \eqref{eq:conditional_VF_2}):
\begin{equation*}
   \begin{aligned}
       \widetilde{\mathcal{V}}_2&(\xtil_2(\psibar_2)) =  \frac{1}{p_2}\rBrackets{\max \cBrackets{\Vunder, \psibar_2 F_0 \exp\rBrackets{\mu_F T}} }^{p_2} \xi \Bigl( T, 0, 1, 0, \rBrackets{\max \cBrackets{\Vunder, \psibar_2 F_0 \exp\rBrackets{\mu_F T}} }^{p_2 - 1} / \widetilde{\lambda}^\ast_2 , +\infty; \gamma, r \Bigr) \\ 
        & + \frac{1}{p_2} \rBrackets{\widetilde{\lambda}^\ast_2}^{\frac{p_2}{p_2 - 1}} \xi \biggl(T, 0, 1, \frac{p_2}{p_2 - 1}, 0, \rBrackets{\max \cBrackets{\Vunder, \psibar_2 F_0 \exp\rBrackets{\mu_F T}} }^{p_2 - 1} / \widetilde{\lambda}^\ast_2; \gamma, r \biggr).
   \end{aligned}
\end{equation*}

The optimal auxiliary liquid portfolio wealth (counterpart of \eqref{eq:X_tilde_star_2_t_power}):
        \begin{align*}
            \widetilde{X}^{\ast}_{2}&(t; \xtil_2\rBrackets{\psi_2^\ast}) =  \rBrackets{\widetilde{\lambda}^\ast_2}^{\frac{1}{p_2 - 1}} \rBrackets{\Ztildet{t}}^{-1}\xi \Biggl( T, t, \Ztilde(t), \frac{p_2}{p_2 - 1}, 0,  \rBrackets{\max \cBrackets{\Vunder, \psibar_2 F_0 \exp\rBrackets{\mu_F T}} }^{p_2 - 1} / \widetilde{\lambda}^\ast_2; \gamma, r \Biggr)  \\
            & - \max \cBrackets{\Vunder, \psiStar_2  F_0 \exp\rBrackets{\mu_F T} } \rBrackets{\Ztildet{t}}^{-1} \xi \Biggl( T, t, \Ztilde(t), 1, 0, \rBrackets{\max \cBrackets{\Vunder, \psibar_2 F_0 \exp\rBrackets{\mu_F T}} }^{p_2 - 1} / \widetilde{\lambda}^\ast_2, \gamma, r \Biggr).
        \end{align*}

\section{Potential extensions of the financial market}\label{sec:discussion_on_extensions}

In this section, we discuss several potential directions for generalizing the financial market in our framework. We list them in the level of anticipated theoretical difficulty from lowest to highest.

It is straightforward to extend our results to a complete Black-Scholes market with many risky assets as in \cite{Lakner2006}. Furthermore, it seems to be relatively uncomplicated to relax the assumption of constant parameters of liquid risky assets to the assumption of deterministic time-dependent parameters as in \cite{Korn2005}. 

The introduction of multiple fixed-term assets $F_1,\dots, F_{n_F}$ for $n_F \in \N$ should be also possible. Denote by $\psi_j$ the units of $F_j$ held by the investor, $j \in \cBrackets{1,\dots, n_F}$. The separation of the original problem into UoC subproblem and UoW subproblem should be still feasible. For each subproblem $i = \cBrackets{1, 2}$, one could apply the two-step approach, but now in the first step one would need to fix an admissible vector of fixed-term investment strategies $\bm{\psibar}_i =\rBrackets{\psibar_{i, 1}, \dots, \psibar_{i, n_F}}^\top$ and solve the respective subproblem. In the second step, one would maximize the conditional value function $\mathcal{\overline{V}}_i\rBrackets{v_i - \bm{\psibar}_i^\top \bm{F}_0}$ w.r.t. $\bm{\psibar}_i$ on the set of admissible values of $\bm{\psibar}_i$, where we defined $\bm{F}_0:= \rBrackets{F_1(0),\dots, F_{n_F}(0)}^\top$. Since the set of admissible $\bm{\psibar}_i$ is compact and $\mathcal{\overline{V}}_i\rBrackets{v_i - \bm{\psibar}_i^\top \bm{F}_0}$ is continuous w.r.t. $\bm{\psibar}_i$, the maximization problem has a solution, which is unique if $\mathcal{\overline{V}}_i\rBrackets{v_i - \bm{\psibar}_i^\top \bm{F}_0}$ is strictly concave w.r.t.  $\bm{\psibar}_i$.

If the liquid market is incomplete, but the fixed-term asset $F$ is hedgeable by liquid assets, then the derivation of the optimal strategies might be possible with the help of the fictitious completion approach developed in \cite{Karatzas1991}. However, it must be extended to the case where utility functions have random components (due to $F$) and the agent has constraints on intermediate consumption and terminal wealth. The question of the analytical tractability of such a problem for general utility functions is open. Similarly to the results in \cite{Karatzas1991}, it may be that a logarithmic utility function or power-utility functions allow the derivation of optimal strategies in semi-closed form. 

If the fixed-term asset has a source of risk that cannot be hedged by trading liquid assets, the analytical tractability of the problem dramatically decreases (cf. Footnote 11 in \cite{Chen2023}). In general, \cite{Desmettre2016} assume that $F$ may not be spanned by liquid assets, i.e., $F(T; \omega) = F(T; M(\omega), N(\omega))$, where $M$ denotes marketable (hedgeable via liquid assets) risk drivers and $N$ denotes non-marketable (non-hedgeable via liquid assets) sources of risk. To solve the asset-allocation problem for an agent maximizing expected utility of terminal wealth without consumption and without terminal-wealth constraint (which corresponds to $U_1(\cdot) \equiv 0, U(\cdot):= U_2(\cdot), \cunder \downarrow 0, \Vunder \downarrow 0$ in our framework), \cite{Desmettre2016} define a random utility function:
\begin{align}
    \widehat{U}_{\omega}(\xbar) :=  \EP\sBrackets{U\rBrackets{\xbar + \psi F(T; m, N)}}_{m = M(\omega)},\,\,\xbar \in (0, +\infty),\,\, \omega \in \Omega, \label{eq:DS2016_U_hat_omega}
\end{align}
where the expectation in \eqref{eq:DS2016_U_hat_omega} operates only w.r.t. $N$, while $m$ is known. The solution involves the inverse marginal utility of \eqref{eq:DS2016_U_hat_omega}. When $N$ is absent, then one can proceed as we described in Section \ref{subsec:UoW}, since the random variable in $\widehat{U}_{\omega}(\xbar)$ is $\mathcal{F}(T)$-measurable and $\widehat{U}_{\omega}(\xbar) = \widetilde{U}_{\omega}(\xbar)$, where $\widetilde{U}_{\omega}(\xbar)$ equals $\widetilde{U}(\xbar)$ defined in \eqref{eq:def_U_tilde} for $\omega  \in \Omega$ and $\Vunder = 0$. However, when $N$ is present, then in general \eqref{eq:DS2016_U_hat_omega} is not known in closed form\footnote{\cite{Desmettre2022} contains a few examples where the function defined by \eqref{eq:DS2016_U_hat_omega} and its inverse can be derived in closed form in the context of surplus optimization for an insurance company with index-linked and performance-linked liabilities that can be only partially hedged.}. In such cases, for any arbitrarily fixed admissible $\psi$, this function and its inverse would be computed numerically (e.g., with Monte-Carlo simulations), which would be time-consuming and prone to estimation errors. These errors would further propagate to the numerical computation of the optimal Lagrange multiplier found by numerically solving a non-linear equation that involves the aforementioned inverse function. This would have to be done for many $\psi$ in order to determine the optimal $\psi^\ast$ by maximizing the conditional value function. It would be challenging to make the overall numerical routine accurate enough on a laptop with average computational performance. If one had both $U_1(\cdot) \not \equiv 0$ and $U_2(\cdot) \not \equiv 0$, then the numerical routine for finding optimal strategies would be more prone to errors and slower, because $v_1^\ast$ would have to be found numerically according to Lemma \ref{lem:vStar_1_vStar_2}.

\section{Numerical studies}\label{sec:numerical_studies}
In this section, we focus on the optimal strategy with respect to the fixed-term asset and answer the following questions:
\begin{enumerate}
    \item What is the sensitivity of the optimal proportion of wealth invested in the illiquid asset at the initial time? (Subsection \ref{subsec:sensitivity_analysis})
    \item What is the subjective value of the illiquid asset in different scenarios? (Subsection \ref{subsec:SVF_analysis})
    \item What is the advantage of the illiquid asset in terms of the investor's terminal-wealth guarantee in different scenarios? (Subsection \ref{subsec:GEUG_analysis})
\end{enumerate}

In our case studies, we look at an economic agent, without specifying further details whether he/she/it is an insurance company, a pension fund, a retail investor, etc. The decision maker's preferences are modeled by power-utility functions, for which Section \ref{sec:explicit_formulas_for_power_U} provides formulas in semi-explicit form. We assume that the model parameters in the base case have the values indicated in Table \ref{tab:base_values_of_parameters}. 

\begin{table}[h!]
\centering
\begin{tabular}{lccccccccccc}
\toprule
& \multicolumn{2}{c}{Investment} & \multicolumn{2}{c}{Utilities} & \multicolumn{5}{c}{Assets} & \multicolumn{2}{c}{Liabilities} \\
\cmidrule(lr){2-3} \cmidrule(lr){4-5} \cmidrule(lr){6-10} \cmidrule(lr){11-12}
Parameters & $T$ & $v_0$ & $p_1$ & $p_2$ & $r$ & $\mu_1$ & $\sigma_1$ & $\mu_F$ & $\sigma_F$ &  $\cunder$ & $\Vunder$  \\
\midrule
Value & $3$ & $100$ & $-2$ & $-1$ & $3\%$ & $8\%$  & $25\%$ & $10\%$ & $25\%$ &  $v_0 \cdot 3\%$ & $v_0 \cdot 80\%$ \\
\bottomrule
\end{tabular}
\caption{Base values of model parameters}
  \label{tab:base_values_of_parameters}
\end{table}

To compute the optimal Lagrange multipliers $\lambda_1^\ast$ and $\lambda_2^\ast$ (see \eqref{eq:UoC_lambdaStar_power} and \eqref{eq:UoW_auxiliary_unconstrained_problem_lambda_power} respectively), we apply the bisection method. To approximate the derivatives of value functions in \eqref{eq:vBar_1}, we use the central finite difference. To calculate $\psiStar$ according to \eqref{eq:UoW_psiStar_power}, we use grid search.

\subsection{Proportion of illiquid capital}\label{subsec:sensitivity_analysis}
In this subsection, we conduct the sensitivity analysis of the \textbf{o}ptimal \textbf{s}hare of \textbf{i}lliquid \textbf{w}ealth (OSIW) with respect to $T$, $r$, $p_1$, $p_2$, $\Delta^{FS} \mu := \mu_F - \mu$, $\Delta^{SF} \sigma = \sigma - \sigma_F$, $\cunder$, and $\Vunder$. Mathematically, the proportion of initial wealth invested in $F$ is given by:
\begin{equation*}
    \text{OSIW}:=\rBrackets{\psiStar F_0 / v_0}.
\end{equation*}
To be concise, we write $Y\%$ wealth guarantee (liabilities) when referring to the constraint that the terminal value of portfolio (assets) is $\Prob$-a.s. larger than or equal to $Y\%$ of the investor's initial wealth, $Y \in \cBrackets{70, 80, 90}$.

Figure \ref{fig:sensitivity_psi_F_vs_T_inter_Vunder} illustrates that the longer the investment period, the less is invested in $F$. For example, an agent with an investment period of $1$ year and an $80\%$ wealth guarantee allocates around two thirds of the initial wealth to the illiquid asset, while an increase in the investment period to $4$ years leads to the optimal allocation of only a half of the agent's initial wealth to the illiquid asset. The longer the investment period, the more important the liquidity of assets for the investor. This is consistent with the observations of \cite{Desmettre2016}, see Figures 9 and 13 there.

\begin{figure}[!ht]
  \centering
  \includegraphics[width=0.9\textwidth]{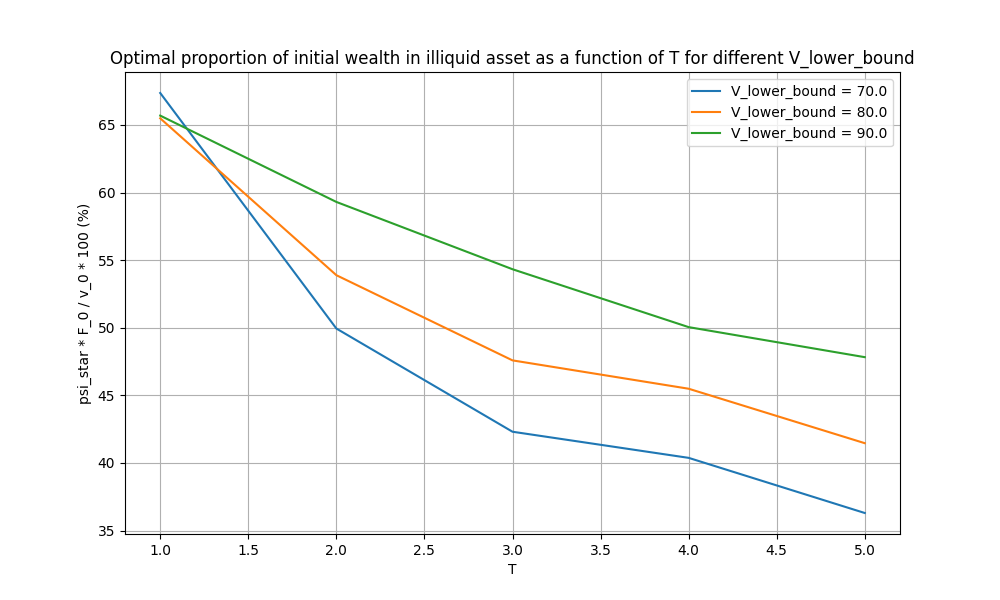}
  \caption{Optimal share of illiquid wealth as a function of $T$ for different $\Vunder$}
  \label{fig:sensitivity_psi_F_vs_T_inter_Vunder}
\end{figure}

\begin{figure}[!ht]
  \centering
  \includegraphics[width=0.9\textwidth]{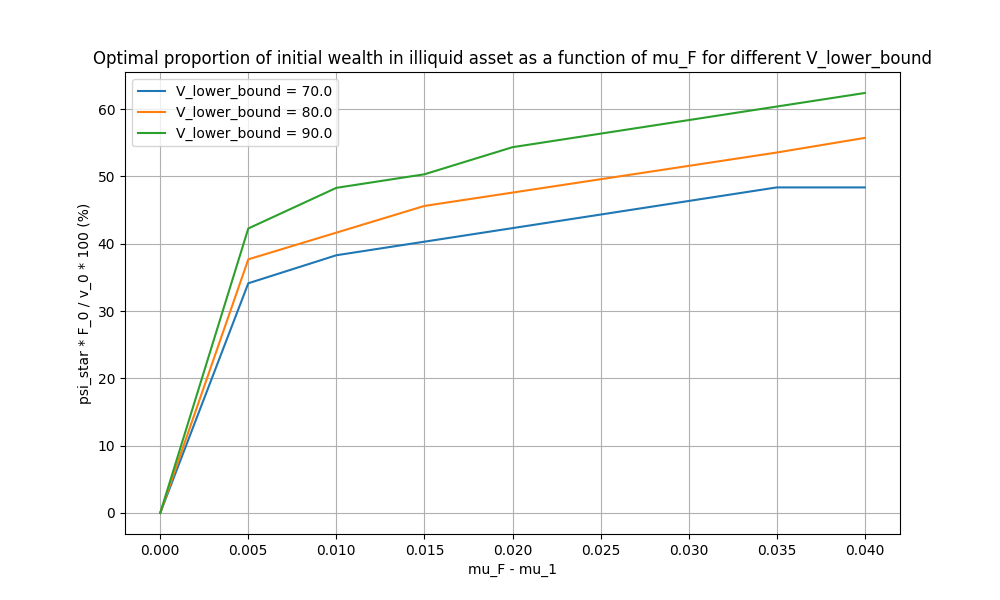}
  \caption{Optimal share of illiquid wealth as a function of $\Delta^{FS}\mu$ for different $\Vunder$}
  \label{fig:sensitivity_psi_F_vs_delta_mu_F_inter_Vunder}
\end{figure}

For three different levels of the terminal-wealth guarantee, Figure \ref{fig:sensitivity_psi_F_vs_delta_mu_F_inter_Vunder} demonstrates that the higher the return differential $\Delta^{FS} \mu = \mu_F - \mu$, the more capital is allocated to the illiquid asset. Consider, for example, an agent with an $80\%$ wealth guarantee. If the liquid risky asset and the fixed-term asset have the same Sharpe ratios, then the agent is indifferent to investing in $F$ and can obtain the maximal expected utility without the illiquid asset, i.e., OSIW is zero. However, an increase in the fixed-term asset's instantaneous rate of return by $1.5\%$, ceteris paribus, leads to a $45\%$ share of initial capital in the illiquid asset. We also observe that for a fixed $\Delta^{FS} \mu$, a larger $\Vunder$ leads to a larger OSIW.

Figure \ref{fig:sensitivity_psi_F_vs_delta_sigma_F_inter_Vunder} shows the dependence of OSIW on $\Delta^{SF} \sigma = \sigma - \sigma_F$, which we call the volatility differential. We see that the lower the volatility of the fixed-term asset, the higher the optimal share of the initial capital invested in $F$. Note that even for $\Delta^{SF} \sigma < 0$ ($\iff \sigma_F > \sigma)$ it is optimal to invest in $F$ as long as the condition \eqref{ineq:case-study-non-redundancy-of-F} is satisfied as a strict inequality. For example, an agent who has a $90\%$ wealth guarantee and access to $F$ with $\Delta^{SF} \sigma = -5\%$ ($\sigma_F = 30\% > 25\% = \sigma$) allocates $40\%$ of the initial capital to $F$, since $\mu_F = 10\% > 8\% = \mu$, which is why the Sharpe ratio of $F$ is strictly higher than the Sharpe ratio of $S_1$. A fixed-term asset with $\Delta^{SF} \sigma = 20\%$ ($\sigma_F = 5\% < 25\% = \sigma$) deserves a significant proportion of the investor's initial capital, namely between two thirds and three quarters depending on the level of terminal-wealth guarantee. Similarly to Figure \ref{fig:sensitivity_psi_F_vs_delta_mu_F_inter_Vunder}, Figure \ref{fig:sensitivity_psi_F_vs_delta_sigma_F_inter_Vunder} also shows that given a fixed volatility differential, a larger $\Vunder$ results in a larger OSIW. 

\begin{figure}[!ht]
  \centering
  \includegraphics[width=0.9\textwidth]{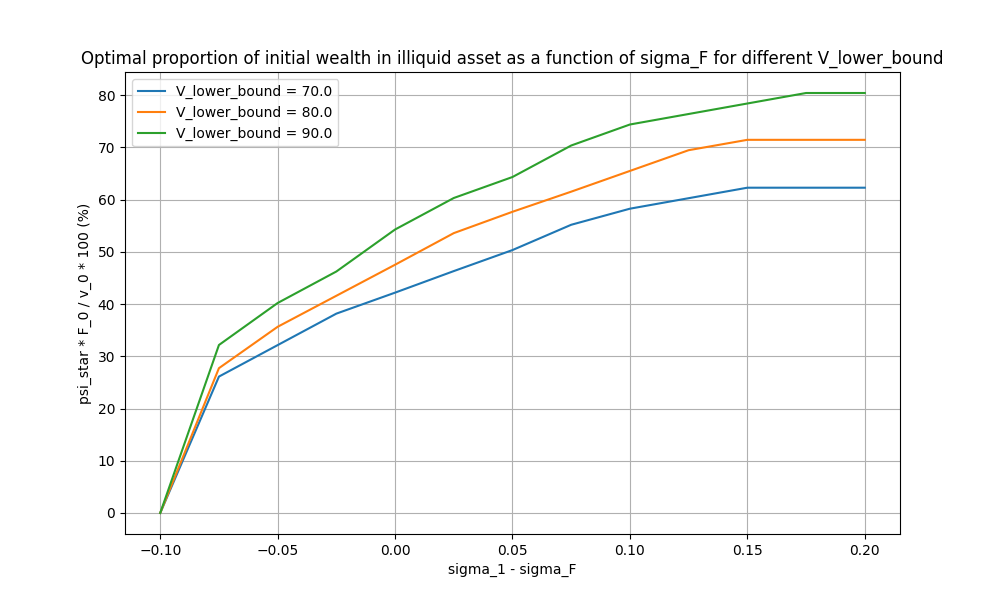}
  \caption{Optimal share of illiquid wealth as a function of $\Delta^{SF}\sigma$ for different $\Vunder$}
  \label{fig:sensitivity_psi_F_vs_delta_sigma_F_inter_Vunder}
\end{figure}

Figures \ref{fig:sensitivity_psi_F_vs_p_1_inter_Vunder} and \ref{fig:sensitivity_psi_F_vs_p_2_inter_Vunder} show the dependence of OSIW on the utility-functions parameters $p_1$ and $p_2$ that represent the decision maker's attitude toward the risk of variability of the consumption rate and the risk of variability of the terminal wealth, respectively. Figure \ref{fig:sensitivity_psi_F_vs_p_1_inter_Vunder} indicates that the smaller the investor's aversion to consumption variation ($p_1 \uparrow$), the smaller the initial capital invested in $F$. Economically, such an agent prefers to maintain a larger amount of liquid capital to dynamically adjust consumption based on evolving market conditions. Mathematically, it is reflected in allocating a larger amount of initial capital to its utility of consumption problem, leaving less capital to invest in $F$, which only matters for the expected utility of terminal wealth. From Figure \ref{fig:sensitivity_psi_F_vs_p_2_inter_Vunder} we learn that the less averse the investor is to variation in terminal wealth, the more\footnote{This is consistent with the findings of \cite{Desmettre2016}, cf. their Figure 12 for $\kappa  = 1$, which corresponds to the case when $F$ does not contain non-marketable risk. Note that our $p_2$ corresponds to their $1 - \rho$.} is invested in $F$. For example, for $80\%$ wealth guarantee, an agent with $p_2 = -2$ (the coefficient of relative risk aversion (RRA) is $3$) allocates approximately $52\%$ of $v_0$ to $F$, whereas $p_2 = -1$ (the RRA-coefficient is $2$) leads to OSIW of approximately $56\%$. 

\begin{figure}[!ht]
  \centering
  \includegraphics[width=0.9\textwidth]{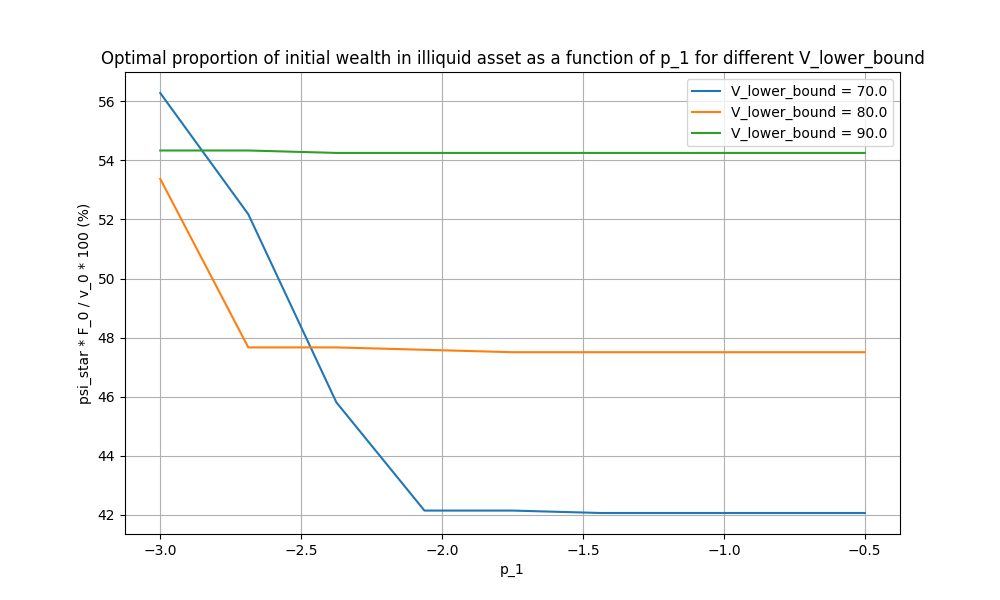}
  \caption{Optimal share of illiquid wealth as a function of $p_1$ for different $\Vunder$}
  \label{fig:sensitivity_psi_F_vs_p_1_inter_Vunder}
\end{figure}

\begin{figure}[!ht]
  \centering
  \includegraphics[width=0.9\textwidth]{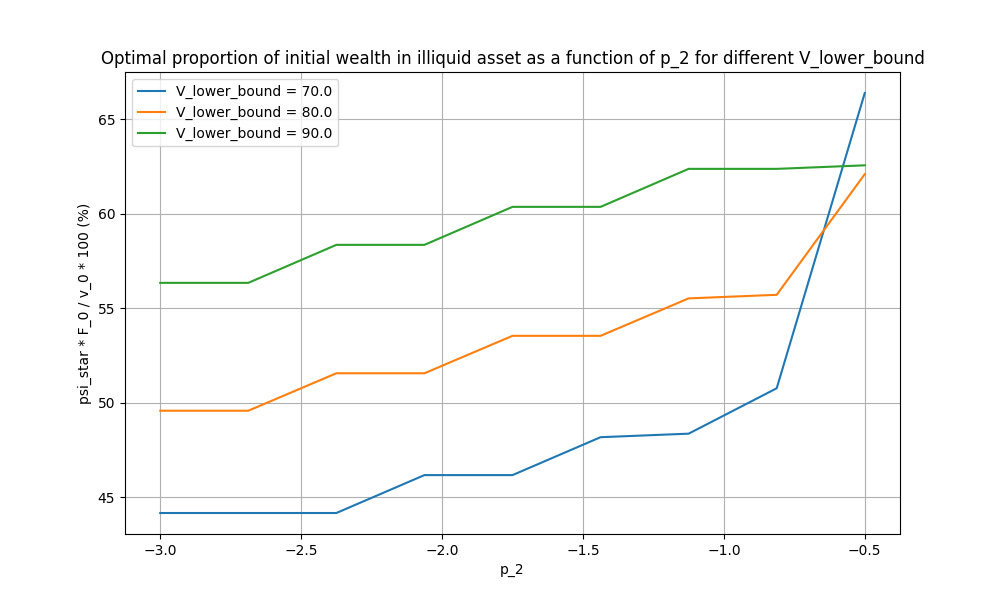}
  \caption{Optimal share of illiquid wealth as a function of $p_2$ for different $\Vunder$}
  \label{fig:sensitivity_psi_F_vs_p_2_inter_Vunder}
\end{figure}

Finally, we comment on the sensitivity of OSIW w.r.t. the interest rate and the lower bounds on consumption and terminal wealth. Regarding the impact of $r$, the higher the interest rate, the less is invested in the fixed-term asset, since $S_0$ becomes more attractive having a similar rate of return and no risk. The sensitivity of OSIW w.r.t. $\cunder$ is very low, since $\cunder$ affects only $v_1^{\min}$. The dependence of OSIW on $\Vunder$ does not show ubiquitous monotonicity patterns. For most parameter values, we observe that a larger $\Vunder$ leads to a larger share of illiquid wealth. However, this dependence changes for short investment horizons (see Figure \ref{fig:sensitivity_psi_F_vs_T_inter_Vunder} for $T \in (0, 1.25]$), relatively small $p_1$ (see Figure \ref{fig:sensitivity_psi_F_vs_p_1_inter_Vunder} for $p_1 < -2.5$), and relatively large $p_2$ (see Figure \ref{fig:sensitivity_psi_F_vs_p_2_inter_Vunder} for $p_2 
\in [-0.6, 1)$). So investors with short investment horizons, high aversion to the consumption variation, or low aversion to the terminal-wealth variation do not necessarily invest more in the illiquid asset if they want to have a larger terminal-wealth guarantee.  

\subsection{Subjective value of a fixed-term asset}\label{subsec:SVF_analysis}
In this subsection, we analyze the monetary value that the agent subjectively assigns to the fixed-term asset. We call it the subjective value of $F$ and abbreviate it by SVF. Being reported in percent, SVF indicates by how many percentage points the initial wealth of an agent without access to $F$ should increase so that the agent gets the same expected utility as when the agent has the original amount of initial capital and can invest in $F$.

Let $\widehat{\mathcal{V}}(v_0; K_c, K_V)$ be the value function of an agent who can trade only $S_0$ and $S_1$. This corresponds to the setting of \cite{Lakner2006}. Then SVF is defined as the solution $\alpha^{\ast}$ to the following non-linear equation w.r.t. $\alpha$:
\begin{equation*}
    \mathcal{V}(v_0; K_c, K_V) = \widehat{\mathcal{V}}(v_0\cdot (1 + \alpha); K_c, K_V).
\end{equation*}

Figure \ref{fig:SVF_vs_p_1_p_2} illustrates the dependence of SVF on the risk-aversion parameters. Subfigure \ref{sfig:SVF_vs_p_1} shows that the more averse the investor is to the risk of consumption variation, the larger is the subjective value of $F$. This is consistent with Figure \ref{fig:sensitivity_psi_F_vs_p_1_inter_Vunder}, where we observed that the decision maker allocates more capital to $F$ for smaller $p_1$, other things being equal. For example, for $p_1 = -2$, the agent trading only liquid assets needs about $0.6\%$ more initial capital to obtain the same expected utility as the expected utility in the case where $F$ belongs to the investment universe. On the contrary, Subfigure \ref{sfig:SVF_vs_p_2} demonstrates that the less averse the investor is to the risk of terminal-wealth variation, the larger is the subjective value of $F$, which is consistent with Figure \ref{fig:sensitivity_psi_F_vs_p_2_inter_Vunder}. For instance, the monetary value that the agent subjectively assigns to $F$ equals approximately $3.25\% \cdot v_0 = 3.25$.

\begin{figure}[!ht] 
    \centering 
    \begin{subfigure}[b]{0.48\textwidth}
        \centering
        \includegraphics[width=\linewidth]{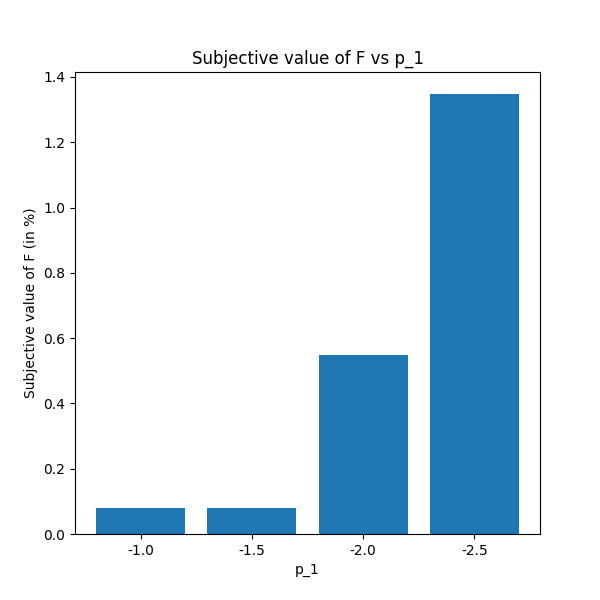} 
        \caption{SVF for different $p_1$}
        \label{sfig:SVF_vs_p_1}
    \end{subfigure}
    \hfill
    \begin{subfigure}[b]{0.48\textwidth} 
        \centering
        \includegraphics[width=\linewidth]{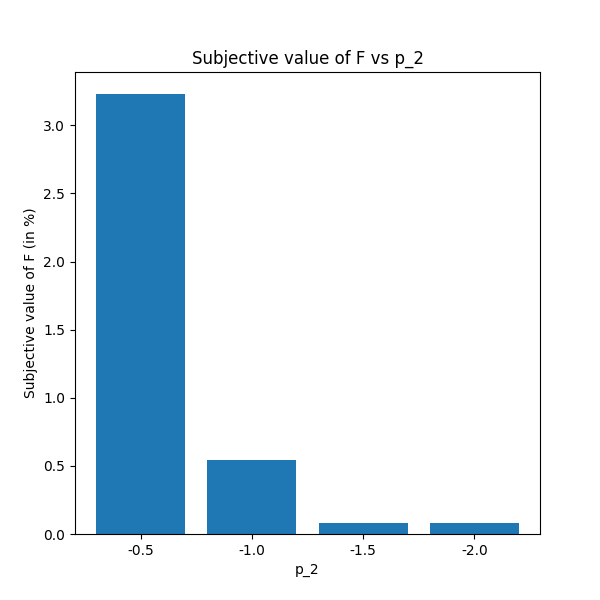} 
        \caption{SVF for different $p_2$}
        \label{sfig:SVF_vs_p_2}
    \end{subfigure}
    \caption{Subjective value of fixed-term asset for different risk appetites}
    \label{fig:SVF_vs_p_1_p_2}
\end{figure}

Figure \ref{fig:SVF_vs_mu_F_sigma_F} illustrates the dependence of SVF on the risk-return profile of $F$. Subfigure \ref{sfig:SVF_vs_mu_F} demonstrates that the larger the return differential $\Delta^{FS}\mu$, the more valuable $F$ for a decision maker. For example, SVF is $0.2\%$ for $\Delta^{FS} \mu  = 1\%$ and increases to $1\%$ for $\Delta^{FS} \mu  = 4\%$. Subfigure \ref{sfig:SVF_vs_sigma_F} shows that the smaller the volatility of the fixed-term asset, the larger the subjective value of $F$. Note that SVF can be positive even if $\Delta^{SF} \sigma < 0$ ($\iff \sigma_F > \sigma$), because \eqref{ineq:case-study-non-redundancy-of-F} can be still satisfied as a strict inequality. When the agent is indifferent to the presence of $F$ (\eqref{ineq:case-study-non-redundancy-of-F} is satisfied as an equality), SVF is $0\%$, as can be seen in Subfigure \ref{sfig:SVF_vs_sigma_F} for $\Delta^{SF} \sigma  = -10\%$.

\begin{figure}[!ht] 
    \centering 
    \begin{subfigure}[b]{0.48\textwidth}
        \centering
        \includegraphics[width=\linewidth]{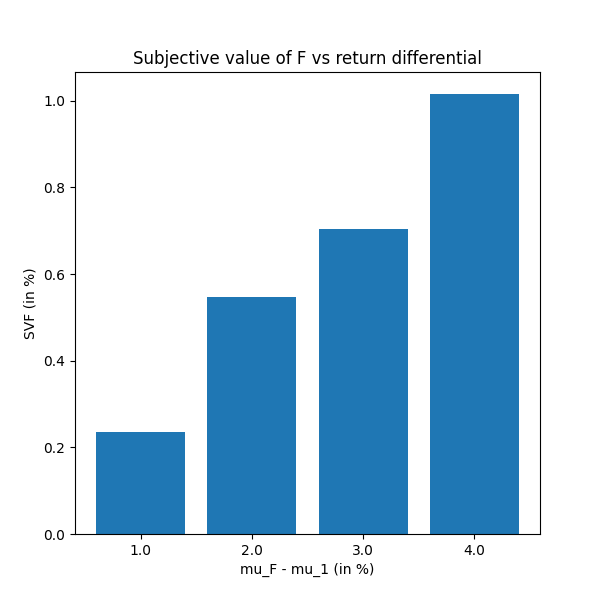} 
        \caption{SVF for different $\Delta^{FS} \mu$}
        \label{sfig:SVF_vs_mu_F}
    \end{subfigure}
    \hfill
    \begin{subfigure}[b]{0.48\textwidth} 
        \centering
        \includegraphics[width=\linewidth]{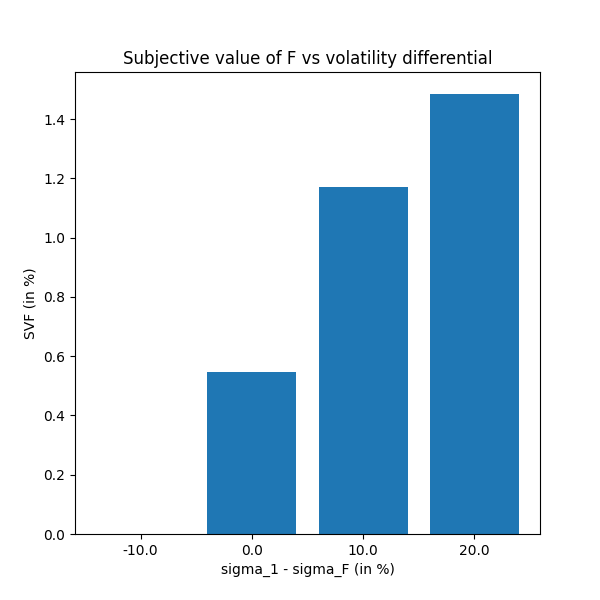} 
        \caption{SVF for different $\Delta^{SF} \sigma$}
        \label{sfig:SVF_vs_sigma_F}
    \end{subfigure}
    \caption{Subjective value of fixed-term asset for different risk-return profiles of $F$}
    \label{fig:SVF_vs_mu_F_sigma_F}
\end{figure}

Next we summarize the impact of other model parameters on SVF and indicate the range of computed SVF without showing the plots. The longer the investment period, the smaller the SVF, decreasing from $1.3\%$ to $0.2\%$. This monotonous behavior was also observed in \cite{Desmettre2016}, cf. Figure 14 there. As the interest rate increased from $1\%$ to $4\%$, SVF decreased from $0.54\%$ to $0\%$.  We also observed that a larger $\Vunder$ leads to a smaller SVF, which ranged between $0.8\%$ and $0.1\%$. In contrast, $\cunder$ had a very little impact on SVF, which is consistent with the impact of $\cunder$ on OSIW. 

In the final part of this subsection, we briefly comment on a related practical application of our framework, which could be beneficial for life insurers and other institutional investors offering investment-linked products to their customers, e.g., unit-linked insurance products. When consulting a client on the most suitable financial product, the company can learn about clients' attitudes to risk by first asking them about the proportion $\Pi$ of the initial capital $v_0$ they would prefer to invest in a stylized fixed-term asset $F$ in some scenario. Second, assuming that the insurer knows the parametric form of the clients' utility functions and that the clients' responses reflect their rational behavior, the insurer could find the respective risk-aversion parameters that would lead to OSIW reported by clients. In an example with power-utility functions, this would mean computing (combinations of) $p_1^\ast$ and $p_2^\ast$ that lead to $\psiStar =  \Pi \cdot v_0 / F_0$ as the optimal investment in the fixed-term asset. This could be done for as many scenarios as needed and provide additional information on the risk appetites of customers, which can be beneficial for recommending tailor-made products to them.

\subsection{Benefit of a fixed-term asset for potential increase of liabilities}\label{subsec:GEUG_analysis}
In this subsection, we examine in different scenarios the benefit of the investor's access to the fixed-term asset in terms of a potential increase in the terminal-wealth guarantee (terminal liabilities). This is done using a so-called guarantee-equivalent utility gain (GEUG), cf. \cite{EscobarAnel2022}. Reported in percents, GEUG indicates by how many percentage points $\Vunder$ can be increased so that the maximal expected utility of the decision maker with the correspondingly higher $\Vunder$ and the access to $F$ is equal to the maximal expected utility of the agent without $F$ in the investment universe. Mathematically, GEUG is the solution $\beta^{\ast}$ to the following non-linear equation w.r.t. $\beta$:
\begin{equation*}
    \mathcal{V}(v_0; K_c, K_V^{\beta}) = \widehat{\mathcal{V}}(v_0; K_c, K_V),
\end{equation*}
where $K_V^{\beta} := \cBrackets{V^{v_0, (\pi, \psi, c)}(T) \,|\, V^{v_0, (\pi, \psi, c)}(T) \geq \underline{V} \cdot (1 + \beta)}$.

Figure \ref{fig:V_GEUG_vs_p} illustrates the dependence of GEUG on risk-aversion parameters. We observe in Subfigure \ref{sfig:V_GEUG_vs_p_1} that the more averse the agent is to the variation of consumption, the more beneficial is the presence of the fixed-term asset for a potential increase of the agent's terminal-wealth guarantee without sacrificing the expected utility. For example, consider a decision maker with $p_1 = -2.5$, $80\%$ wealth guarantee and no possibility of investing in illiquid assets. If such an agent gets the opportunity to buy $F$, then the agent can achieve a terminal wealth that is $\Prob$-a.s. larger than or equal to $81.76 = 80 \cdot (1 + 2.2 / 100)$ without loss in expected utility. According to Subfigure \ref{sfig:V_GEUG_vs_p_2}, the less averse the investor is towards the risk of terminal-wealth variation, the greater the potential benefit of having access to illiquid $F$. Consider, for example, an agent with $p_2 = -0.5$, $\Vunder = 80$ and an investment universe consisting only of liquid assets. If the decision maker gets the opportunity to invest in illiquid $F$, the agent can have a terminal-wealth guarantee $87.2 = 80 \cdot (1 + 9 / 100)$ while achieving the same expected utility as before gaining access to $F$. 

\begin{figure}[!ht]
    \centering
    \begin{subfigure}[b]{0.48\textwidth}
        \includegraphics[width=\textwidth]{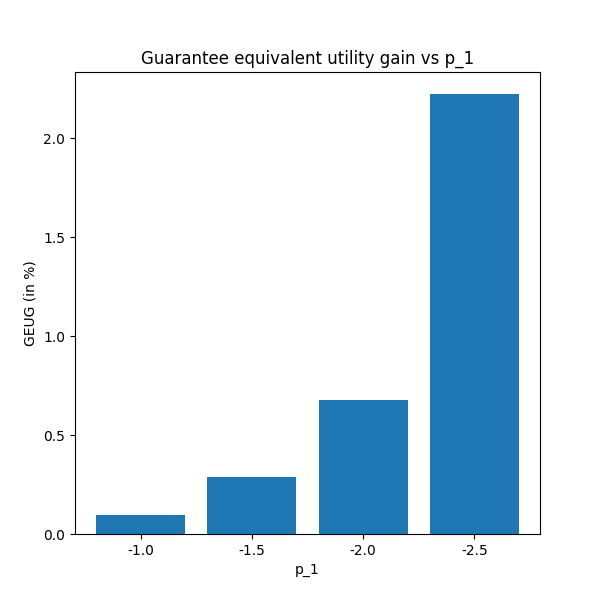}
        \caption{GEUG for different $p_1$}
        \label{sfig:V_GEUG_vs_p_1}
    \end{subfigure}
    \hfill 
    \begin{subfigure}[b]{0.48\textwidth}
        \includegraphics[width=\textwidth]{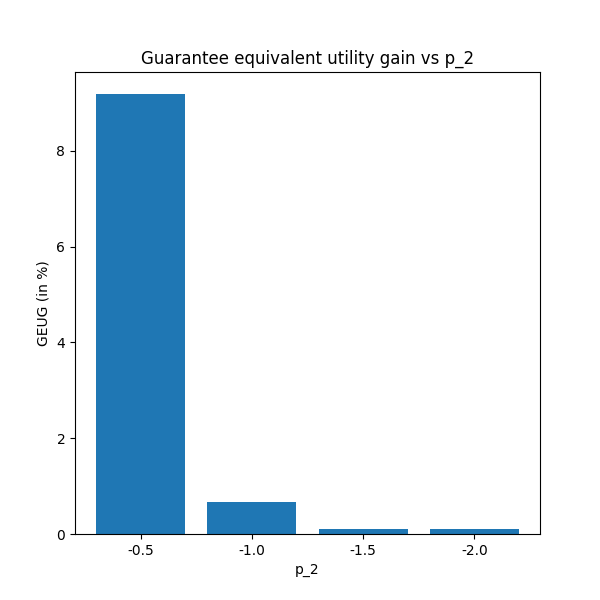}
        \caption{GEUG for different $p_2$}
        \label{sfig:V_GEUG_vs_p_2}
    \end{subfigure}
    \caption{Guarantee equivalent utility gain for different risk appetites}
    \label{fig:V_GEUG_vs_p}
\end{figure}

Figure \ref{fig:V_GEUG_vs_mu_F_and_sigma_F} depicts how GEUG depends on the risk-return profile of the fixed-term asset. Subfigure \ref{sfig:V_GEUG_vs_mu_F} shows that a larger return differential $\Delta^{FS}\mu$ leads to a larger GEUG. For example, consider an agent with access to $F$ and $\Delta^{FS}\mu = 3\%$ ($\mu_F = 11\%, \mu = 8\%$). Then, compared to the situation in which the agent trades only liquid assets $S_0$ and $S_1$, the presence of illiquid $F$ in the investment universe allows the decision maker to increase $\Vunder$ by approximately $1\%$ without loss in expected utility. In Subfigure \ref{sfig:V_GEUG_vs_mu_F}, we see that the lower the volatility of $F$, the larger the GEUG. For instance, access to $F$ with $\sigma_F = 5\%$  ($\Delta \sigma = 20\%$) enables the agent to increase the terminal-wealth guarantee $\Vunder$ by approximately $2\%$ (from $80$ to about $81.6$) while obtaining the same expected utility as the one corresponding to the case when only liquid assets are traded. For $\Delta^{SF} \sigma = -10\%$ ($ \sigma = 25\%, \sigma_F = 35\%$), GEUG is zero, because the inequality \eqref{ineq:case-study-non-redundancy-of-F} is satisfied with equality, which means that the agent is indifferent to the presence of the illiquid asset. 

\begin{figure}[!ht]
    \centering
    \begin{subfigure}[b]{0.48\textwidth}
        \includegraphics[width=\textwidth]{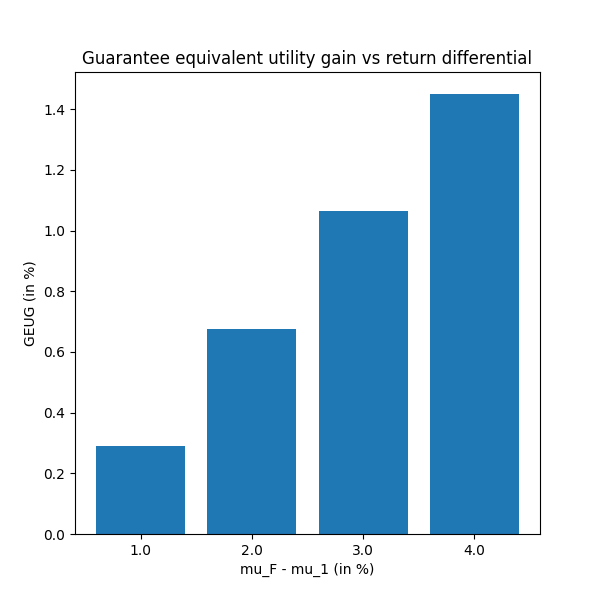}
        \caption{GEUG for different $\Delta^{FS} \mu$}
        \label{sfig:V_GEUG_vs_mu_F}
    \end{subfigure}
    \hfill 
    \begin{subfigure}[b]{0.48\textwidth}
        \includegraphics[width=\textwidth]{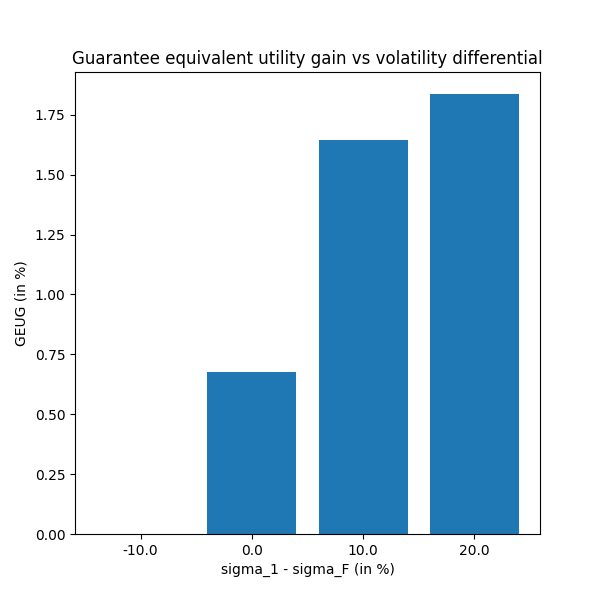}
        \caption{GEUG for different $\Delta^{SF} \sigma$}
        \label{sfig:V_GEUG_vs_sigma_F}
    \end{subfigure}
    \caption{Guarantee equivalent utility gain for different risk-return profiles of $F$}
    \label{fig:V_GEUG_vs_mu_F_and_sigma_F}
\end{figure}

\begin{figure}[!ht]
    \centering
    \begin{subfigure}[b]{0.48\textwidth}
        \includegraphics[width=\textwidth]{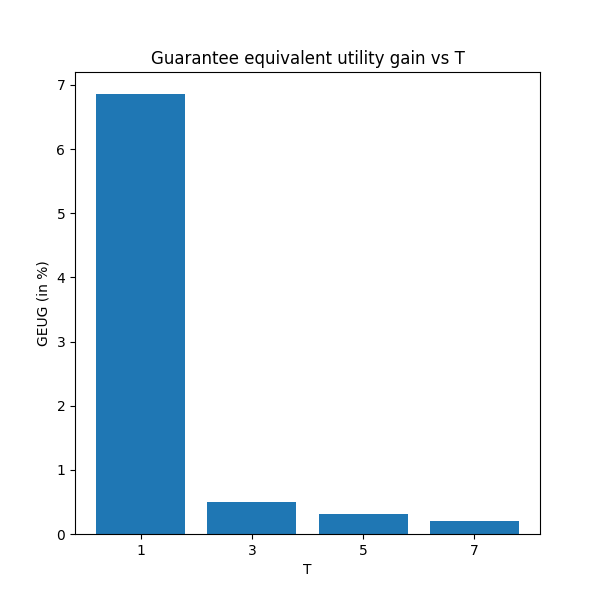}
        \caption{GEUG for different $T$}
        \label{sfig:V_GEUG_vs_T}
    \end{subfigure}
    \hfill 
    \begin{subfigure}[b]{0.48\textwidth}
        \includegraphics[width=\textwidth]{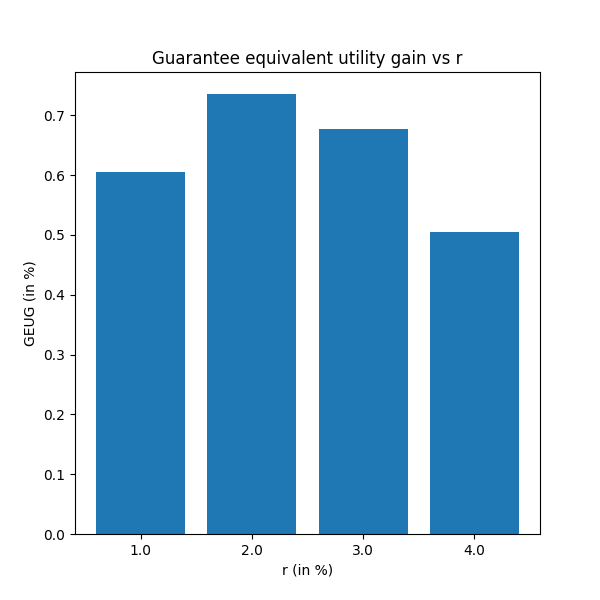}
        \caption{GEUG for different $r$}
        \label{sfig:V_GEUG_vs_r}
    \end{subfigure}
    \caption{Guarantee equivalent utility gain for investment horizons and interest rates}
    \label{fig:V_GEUG_vs_r_and_T}
\end{figure}

Figure \ref{fig:V_GEUG_vs_r_and_T} shows how the length of an investment period and the interest rate influence GEUG. Subfigure \ref{sfig:V_GEUG_vs_T} indicates that the benefit of $F$ for having a larger $\Vunder$ is decreasing with longer investment horizons. For example, for a one-year period, GEUG is almost $7\%$, whereas it drops to approximately $0.6\%$ for a three-year investment period. This is consistent with Figure \ref{fig:sensitivity_psi_F_vs_T_inter_Vunder}, where we saw that the larger the investment horizon $T$, the larger share of initial capital is invested in liquid assets. According to Subfigure \ref{sfig:V_GEUG_vs_r}, the relationship between $r$ and GEUG is not monotonous. GEUG increases in $r$ for interest rates $\lessapprox 2\%$ and decreases in $r$ otherwise.

Finally, we summarize the impact of the liability-side parameters $\cunder$ and $\Vunder$ on GEUG. The former parameter has very little impact on GEUG, which also aligns with what we observed for its influence on OSIW and SVF. As for $\Vunder$, the larger it is, the smaller is the benefit of $F$ for increasing this lower bound without loss of the agent's utility. For values $\Vunder \in [75, 90]$, GEUG decreased from approximately $1.4\%$ to $0.1\%$.  The reason for this relationship is the volatility of $F$, which forces the agent to invest more in the risk-free asset to ensure that the terminal value of assets is $\Prob$-a.s. larger than or equal to the terminal value $\Vunder$ of liabilities. 

\section{Conclusion}\label{sec:conclusion}
We propose an asset-liability management framework for investors whose investment universe includes both liquid assets (e.g., cash, stocks) and fixed-term assets (e.g., deposits, infrastructure investments). Mathematically, our framework is a utility maximization problem of a decision maker who chooses investment-consumption strategies such that the consumption rate and the terminal portfolio are bigger than or equal to their respective lower bounds, which are determined upfront and represent the liability aspect of our modeling approach. The proposed framework is analytically tractable and generalizes the setting of \cite{Lakner2006} by allowing the decision maker to invest in fixed-term assets modeled as in \cite{Desmettre2016}. Our approach to the derivation of optimal decisions also relies on and slightly extends the idea of \cite{Korn2005} to decompose portfolio optimization problems with stochastic lower bounds on terminal wealth into two subproblems, namely a lower-bound replication problem and an auxiliary unconstrained portfolio optimization problem.

To demonstrate the analytical tractability of our framework, we derive in semi-closed form the optimal strategies of agents with power-utility functions in two cases: when the fixed-term asset is deterministic and when it has randomness. In the respective numerical studies for stochastic fixed-term asset, we show how model parameters influence the optimal share of initial capital invested in the illiquid asset, the monetary value that the decision maker subjectively assigns to the fixed-term asset, and the potential of the illiquid asset to increase terminal liabilities without loss in the agent's utility.  

There are different directions for extending our framework. One direction is considering more sophisticated financial markets, e.g., with non-headgeable risks in liquid or illiquid assets, with random time points when fixed-term assets can be traded, etc. Another direction is including behavioral economics components such as S-shaped utility functions or probability distortions. It also seems interesting and relevant to embed borrowing constraints into the liquid market of our framework and analyze their impact on the optimal strategies.

\section*{Acknowledgments}
The work leading to this publication was supported by the PRIME programme of the German Academic Exchange Service (DAAD) with funds from the German Federal Ministry of Education and Research (BMBF).

Yevhen Havrylenko also acknowledges the support of Johannes Kepler University Linz, where he was as a visiting researcher during the final stage of this research.

The author of this publication expresses his gratitude to Sascha Desmettre, Ralf Korn, Thai Nguyen, Frank Seifried, Mogens Steffensen, the participants of the 27th International Congress on Insurance: Mathematics and Economics, the Scandinavian Actuarial Conference 2024, the European Actuarial Journal Conference 2024, the IVW-ifa Workshop 2024 at Ulm University, and the 17th German Probability and Statistics Days for their feedback on this paper.

\section*{Declaration of generative AI and AI-assisted technologies in the writing process} During the preparation of this work, the author used Gemini models to improve the language of the manuscript and to improve the author's code written for numerical studies. After using Gemini models, the author reviewed and edited the content as needed and takes full responsibility for the content of the published article.

\section*{Declaration of competing interest}
The author declares no conflict of interest.

\bibliographystyle{chicago}
\bibliography{main}

\begin{thebibliography}{}

\bibitem[\protect\citeauthoryear{Albrecher, Bauer, Embrechts, Filipovi\'c, Koch-Medina, Korn, Loisel, Pelsser, Schiller, Schmeiser, and Wagner}{Albrecher et~al.}{2018}]{Albrecher2018}
Albrecher, H., D.~Bauer, P.~Embrechts, D.~Filipovi\'c, P.~Koch-Medina, R.~Korn, S.~Loisel, A.~Pelsser, F.~Schiller, H.~Schmeiser, and J.~Wagner (2018).
\newblock Asset-liability management for long-term insurance business.
\newblock {\em European Actuarial Journal\/}~{\em 8\/}(1), 9 -- 25.

\bibitem[\protect\citeauthoryear{Ang, Papanikolaou, and Westerfield}{Ang et~al.}{2014}]{Ang2014}
Ang, A., D.~Papanikolaou, and M.~M. Westerfield (2014).
\newblock Portfolio choice with illiquid assets.
\newblock {\em Management Science\/}~{\em 60\/}(11), 2737--2761.

\bibitem[\protect\citeauthoryear{Basak and Shapiro}{Basak and Shapiro}{2001}]{Basak2001}
Basak, S. and A.~Shapiro (2001).
\newblock {Value-at-Risk-Based Risk Management: Optimal Policies and Asset Prices}.
\newblock {\em The Review of Financial Studies\/}~{\em 14\/}(2), 371--405.

\bibitem[\protect\citeauthoryear{Belak and Christensen}{Belak and Christensen}{2019}]{Belak2019}
Belak, C. and S.~Christensen (2019).
\newblock Utility maximisation in a factor model with constant and proportional transaction costs.
\newblock {\em Finance and Stochastics\/}~{\em 23\/}(1), 29--96.

\bibitem[\protect\citeauthoryear{Bichuch and Guasoni}{Bichuch and Guasoni}{2018}]{Bichuch2018}
Bichuch, M. and P.~Guasoni (2018).
\newblock Investing with liquid and illiquid assets.
\newblock {\em Mathematical Finance\/}~{\em 28\/}(1), 119--152.

\bibitem[\protect\citeauthoryear{Boonen}{Boonen}{2017}]{Boonen2017}
Boonen, T.~J. (2017).
\newblock {Solvency II solvency capital requirement for life insurance companies based on expected shortfall}.
\newblock {\em European Actuarial Journal\/}~{\em 7}, 405--434.

\bibitem[\protect\citeauthoryear{Boyle and Tian}{Boyle and Tian}{2007}]{Boyle2007}
Boyle, P. and W.~Tian (2007).
\newblock Portfolio management with constraints.
\newblock {\em Mathematical Finance\/}~{\em 17\/}(3), 319--343.

\bibitem[\protect\citeauthoryear{Chen, Nguyen, and Stadje}{Chen et~al.}{2018a}]{Chen2018b}
Chen, A., T.~Nguyen, and M.~Stadje (2018a).
\newblock {Optimal investment under VaR-Regulation and Minimum Insurance}.
\newblock {\em Insurance: Mathematics and Economics\/}~{\em 79\/}(C), 194--209.

\bibitem[\protect\citeauthoryear{Chen, Nguyen, and Stadje}{Chen et~al.}{2018b}]{Chen2018c}
Chen, A., T.~Nguyen, and M.~Stadje (2018b, October).
\newblock {Risk management with multiple VaR constraints}.
\newblock {\em Mathematical Methods of Operations Research\/}~{\em 88\/}(2), 297--337.

\bibitem[\protect\citeauthoryear{Chen, Li, and Zeng}{Chen et~al.}{2023}]{Chen2023}
Chen, Z., Z.~Li, and Y.~Zeng (2023).
\newblock Portfolio choice with illiquid asset for a loss-averse pension fund investor.
\newblock {\em Insurance: Mathematics and Economics\/}~{\em 108}, 60--83.

\bibitem[\protect\citeauthoryear{Desmettre and Seifried}{Desmettre and Seifried}{2016}]{Desmettre2016}
Desmettre, S. and F.~T. Seifried (2016).
\newblock Optimal asset allocation with fixed-term securities.
\newblock {\em Journal of Economic Dynamics and Control\/}~{\em 66}, 1 -- 19.

\bibitem[\protect\citeauthoryear{Desmettre, Wahl, and Zagst}{Desmettre et~al.}{2022}]{Desmettre2022}
Desmettre, S., M.~Wahl, and R.~Zagst (2022).
\newblock Dynamic surplus optimization with performance- and index-linked liabilities.
\newblock {\em European Actuarial Journal\/}~{\em 12\/}(2), 607 -- 645.

\bibitem[\protect\citeauthoryear{Escobar-Anel, Havrylenko, Kschonnek, and Zagst}{Escobar-Anel et~al.}{2022}]{EscobarAnel2022}
Escobar-Anel, M., Y.~Havrylenko, M.~Kschonnek, and R.~Zagst (2022).
\newblock {Decrease of capital guarantees in life insurance products: Can reinsurance stop it?}
\newblock {\em Insurance: Mathematics and Economics\/}~{\em 105}, 14--40.

\bibitem[\protect\citeauthoryear{Geman, El~Karoui, and Rochet}{Geman et~al.}{1995}]{Geman1995}
Geman, H., N.~El~Karoui, and J.-C. Rochet (1995).
\newblock Changes of num{\'e}raire, changes of probability measure and option pricing.
\newblock {\em Journal of Applied Probability\/}~{\em 32\/}(2), 443–458.

\bibitem[\protect\citeauthoryear{Gârleanu}{Gârleanu}{2009}]{Garleanu2009}
Gârleanu, N. (2009).
\newblock Portfolio choice and pricing in illiquid markets.
\newblock {\em Journal of Economic Theory\/}~{\em 144\/}(2), 532--564.

\bibitem[\protect\citeauthoryear{Havrylenko}{Havrylenko}{2023}]{Havrylenko2023}
Havrylenko, Y. (2023).
\newblock {\em Risk limitation and risk sharing in investment and insurance}.
\newblock Ph.\ D. thesis, Technische Universität München.

\bibitem[\protect\citeauthoryear{Karatzas, Lehoczky, Shreve, and Xu}{Karatzas et~al.}{1991}]{Karatzas1991}
Karatzas, I., J.~P. Lehoczky, S.~E. Shreve, and G.-L. Xu (1991, May).
\newblock Martingale and duality methods for utility maximization in an incomplete market.
\newblock {\em SIAM Journal. Control and Optimization\/}~{\em 29\/}(3), 702--730.

\bibitem[\protect\citeauthoryear{Karatzas and Shreve}{Karatzas and Shreve}{1998}]{Karatzas1998}
Karatzas, I. and S.~Shreve (1998).
\newblock {\em Methods of mathematical finance\/} (4th ed.).
\newblock Princeton, New Jersey: Springer-Verlag New York.

\bibitem[\protect\citeauthoryear{Korn}{Korn}{1997}]{Korn1997}
Korn, R. (1997).
\newblock {\em Optimal Portfolios}.
\newblock WORLD SCIENTIFIC.

\bibitem[\protect\citeauthoryear{Korn}{Korn}{2005}]{Korn2005}
Korn, R. (2005).
\newblock Optimal portfolios with a positive lower bound on final wealth.
\newblock {\em Quantitative Finance\/}~{\em 5\/}(3), 315--321.

\bibitem[\protect\citeauthoryear{Korn}{Korn}{2014}]{Korn2014}
Korn, R. (2014).
\newblock {\em Moderne Finanzmathematik -- Theorie und praktische Anwendung}.
\newblock Springer Spektrum.

\bibitem[\protect\citeauthoryear{Lakner and Nygren}{Lakner and Nygren}{2006}]{Lakner2006}
Lakner, P. and L.~M. Nygren (2006).
\newblock Portfolio optimization with downside constraints.
\newblock {\em Mathematical Finance\/}~{\em 16\/}(2), 283--299.

\bibitem[\protect\citeauthoryear{Lichtenstern, Shevchenko, and Zagst}{Lichtenstern et~al.}{2021}]{Lichtenstern2021}
Lichtenstern, A., P.~Shevchenko, and R.~Zagst (2021).
\newblock Optimal life-cycle consumption and investment decisions under age-dependent risk preferences.
\newblock {\em Mathematics and Financial Economics\/}~{\em 15}, 275–313.

\bibitem[\protect\citeauthoryear{Merton}{Merton}{1971}]{Merton1971}
Merton, R. (1971).
\newblock {Optimum consumption and portfolio rules in a continuous-time model}.
\newblock {\em Journal of Economic Theory\/}~{\em 4\/}(3), 373--413.

\bibitem[\protect\citeauthoryear{Merton}{Merton}{1969}]{Merton1969}
Merton, R.~C. (1969).
\newblock {Lifetime Portfolio Selection under Uncertainty: The Continuous-Time Case}.
\newblock {\em The Review of Economics and Statistics\/}~{\em 51\/}(3), 247.

\bibitem[\protect\citeauthoryear{Munk}{Munk}{2000}]{Munk2000}
Munk, C. (2000).
\newblock Optimal consumption/investment policies with undiversifiable income risk and liquidity constraints.
\newblock {\em Journal of Economic Dynamics and Control\/}~{\em 24\/}(9), 1315--1343.

\bibitem[\protect\citeauthoryear{Teplá}{Teplá}{2001}]{Tepla2001}
Teplá, L. (2001).
\newblock Optimal investment with minimum performance constraints.
\newblock {\em Journal of Economic Dynamics and Control\/}~{\em 25\/}(10), 1629 -- 1645.

\bibitem[\protect\citeauthoryear{Ye}{Ye}{2008}]{Ye2008}
Ye, J. (2008).
\newblock Optimal life insurance, consumption and portfolio: A dynamic programming approach.
\newblock In {\em 2008 American Control Conference}, pp.\  356--362.

\end{thebibliography}

\appendix
\section{Proofs for Section \ref{sec:methodology}}\label{app:proofs_general}
\begin{proof}[Proof of Lemma \ref{lem:v_1_min}]
    First, we prove that $v_1 \geq v_1^{\min}$ implies that the set of strategies admissible for \eqref{OP:UoC_problem} is non-empty.  
    The minimal consumption stream acceptable for the investor is given by the lower bound in \eqref{eq:consumption_constraint}. The budget needed for it can be calculated as follows:
    \begin{equation*}
        \EQ\sBrackets{\intzeroT \exp \rBrackets{-r t} \cunder dt} \stackrel{(i)}{=} \cunder \intzeroT \exp \rBrackets{-r t}  dt = \cunder \rBrackets{1 - \exp\rBrackets{ - r T}} / r = v_1^{\min},
    \end{equation*}
    where we use in (i) that $\cunder$ and $r$ are constants. So an admissible consumption stream $c(t) = \cunder, \forall \tin,$ can be financed by buying and holding $v_1^{\min} / S_0(0) = v_1^{\min}$ units of the risk-free asset $S_0$ and not investing in the risky asset $S_1$, i.e., $\pi_1(t) = 0$ for $\forall \tin$. Thus, for $v_1 \geq v_1^{\min}$, the set of admissible investment-consumption strategies is non-empty, since, e.g., the strategy $\rBrackets{\pi_1 \equiv 0, \psi_1 = 0, c_1 \equiv \cunder}$ is admissible. 
    
    Second, we prove the other direction, i.e., if the set of strategies admissible for \eqref{OP:UoC_problem} is non-empty, then $v_1 \geq v_1^{\min}$. We show that by contradiction. Assume that there exists a strategy $\rBrackets{\widehat{\pi}_1, \widehat{\psi}_1, \widehat{c}_1}$ that is admissible for \eqref{OP:UoC_problem} and that requires an initial capital $\widehat{v}_1 < v_1^{\min}$. This implies that $\widehat{c}_1(t; \widehat{v}_1) \geq \cunder\,\,\forall\tin$ and, thus:
    \begin{equation}\label{eq:v_1_tilde_contradiction_step_1}
       \EQ\sBrackets{\intzeroT \exp \rBrackets{-r t}  \widehat{c}_1(t; \widehat{v}_1)\,dt} \geq \EQ\sBrackets{\intzeroT  \exp \rBrackets{-r t} \cunder\,dt} = v_1^{\min}.
    \end{equation}
    However, from Equation (2.10) in \cite{Lakner2006} and the admissibility condition \linebreak $X^{\widehat{v}_1 - \widehat{\psi}_1 F_0, \rBrackets{\widehat{\pi}_1, \widehat{c}_1}}(T)\geq 0$ (see \eqref{eq:X_regularity_conditions}), it follows that:
    \begin{equation}\label{eq:v_1_tilde_contradiction_step_2}
        \widehat{v}_1 \geq \EQ\sBrackets{\intzeroT \exp \rBrackets{-r t}  \widehat{c}_1(t; \widehat{v}_1)\,dt}.
    \end{equation}
    Due to \eqref{eq:v_1_tilde_contradiction_step_1} and \eqref{eq:v_1_tilde_contradiction_step_2}, we get that $\widehat{v}_1 \geq v_1^{\min}$, which contradicts the assumption $\widehat{v}_1 < v_1^{\min}$ made at the beginning of our proof by contradiction. Therefore, that assumption was wrong, i.e., there does not exist a strategy $\rBrackets{\widehat{\pi}_1, \widehat{\psi}_1, \widehat{c}_1}$ that is admissible for \eqref{OP:UoC_problem} and that requires $\widehat{v}_1 < v_1^{\min}$.
\end{proof}

\begin{proof}[Proof of Proposition \ref{prop:UoC_problem_solution}]
First, we consider the case where $v_1 < v_1^{\min}$. By Lemma \ref{lem:v_1_min}, the problem does not have a solution, as the investor's initial capital is insufficient to finance the minimal consumption stream satisfying \eqref{eq:consumption_constraint}.

Second, we consider the case where $v_1 = v_1^{\min}$. Investing $v_1 = v_1^{\min}$ in risk-free asset only and consuming $c(t) = \cunder$ $\forall \tin$, the investor matches his/her/their minimal consumption. Since the constraint \eqref{eq:consumption_constraint} must hold $\Prob$-almost surely, the realized return of $S_1$ can be below $r$, there is no other option to match the minimal consumption. Thus, $\piStar_1(t;v_1) = 0$, $\psiStar_1 = 0$, $\cStar_1(t;v_1) = \cunder$. Plugging $\piStar_1$ and $c_1^\ast$ in \eqref{eq:SDE_wealth_process}, we obtain the following ordinary differential equation (ODE) and initial condition:
\begin{align*}
     dX_1^\ast(t) & = X_1^\ast(t)rdt - \cunder dt, \qquad X_1^\ast(0) = v_1 = v_1^{\min}.
\end{align*}
Multiplying both sides of the above ODE by $\exp\rBrackets{-rt}$, we get:
\begin{align*}
     \exp\rBrackets{-rt}dX_1^\ast(t) -  \exp\rBrackets{-rt}X_1^\ast(t)rdt &=  - \exp\rBrackets{-rt} \cunder dt; \\
    d \rBrackets{  \exp\rBrackets{-rt} X_1^\ast(t)} &=  - \exp\rBrackets{-rt} \cunder dt \\
   \int \limits_0^t d \rBrackets{  \exp\rBrackets{-rs} X_1^\ast(s)}& = \int \limits_0^t \rBrackets{ - \exp\rBrackets{-rt} \cunder dt} ; \\
   \exp\rBrackets{-rt} X_1^\ast(t) - \exp\rBrackets{- r \cdot 0} X_1^\ast(0) &= - \cunder \rBrackets{- \frac{1}{r}\exp\rBrackets{-r t} + \frac{1}{r}\exp\rBrackets{-r \cdot 0}} ;\\
  \exp \rBrackets{-rt} X_1^\ast(t) &=  v_1^{\min} +  \cunder \rBrackets{\exp\rBrackets{-r t} - 1} / r ;\\
   X_1^\ast(t) &=  v_1^{\min}\exp \rBrackets{rt} +  \cunder \rBrackets{1 - \exp \rBrackets{rt}} / r. 
\end{align*}

Using the definition $v_1^{\min} = \cunder\rBrackets{1 - \exp\rBrackets{ - r T}} / r$, we get an alternative representation of the optimal wealth:
\begin{equation*}
    X_1^\ast(t) = \cunder \rBrackets{1 - \exp\rBrackets{ - r (T - t)}} / r, \quad \tin.
\end{equation*}

Third, we consider the case where $v_1 > v_1^{\min}$. We derive optimal decisions and related objects in two steps:
\begin{enumerate}
    \item solve \eqref{OP:UoC_problem} for an arbitrarily fixed admissible $\psibar_1$;
    \item optimize w.r.t. $\psibar_1$ the value function from Step 1.
\end{enumerate}
\textit{Step 1.} Fix an arbitrary admissible $\psibar_1 \in \sBrackets{0, \frac{v_1}{F(0)}}$. Then the investor's initial capital left to invest in liquid assets equals $x_1:=x_1(\psibar_1) = v_1 - \psibar_1 \cdot F(0)$. Consider the following problem:
    \begin{align}
        \max_{(\pi, c)}  \,\EP \sBrackets{\int \limits_{t = 0}^{T} U_1\rBrackets{t, c(t)}dt} \,\,\text{s.t.} \,\, (\pi, \psibar_1, c) \in \mathcal{A}_u^{\pi, \psi, c}(x_1),\,c \in K_c. 
        \label{OP:UoC_problem_fixed_psi} \tag{$\overline{P}^{\pi, \psibar_1, c}_{1}$} 
    \end{align}
Using Proposition 3.1 from \cite{Lakner2006}, we can obtain the investor's optimal consumption as per \eqref{eq:UoC_cStar}, the optimal Lagrange multiplier as per \eqref{eq:UoC_lambdaStar} (whose existence and uniqueness is guaranteed by \eqref{eq:UoCS_assumption_regarding_lambda}),  and the optimal liquid portfolio value as per \eqref{eq:UoC_psiStar}, but with $v_1$ replaced by $x_1$ in all formulas.

The existence of the optimal investment strategy w.r.t. liquid risky asset is guaranteed by the completeness of the financial market. If $X_1^\ast(t;x_1) = g_1(t, W^{\Prob}(t))$ for some $g_1(t, w) \in \mathcal{C}^{1,2}\rBrackets{[0, T] \times \R}$ with $g_1(0, 0) = x_1$, then Theorem 21 (\enquote{Comparison of coefficients}) in \cite{Korn1997} is applicable and yields \eqref{eq:UoC_piStar}. Alternatively, one can use Theorem 23 (\enquote{Feedback representation in the Markovian case}) in \cite{Korn1997} under additional assumptions.

\textit{Step 2.} Denote by $\overline{\mathcal{V}}_1(x_1(\psibar_1))$ the value function of \eqref{OP:UoC_problem_fixed_psi}:
\begin{equation*}
   \overline{\mathcal{V}}_1\rBrackets{x_1\rBrackets{\psibar_1}} = \EP \sBrackets{\int \limits_{0}^{T} U_1\rBrackets{t, I_1\rBrackets{t, \lambda^\ast_1(x_1(\psibar_1)) \Ztilde(t)}}dt}.
\end{equation*}
Since  $U_1(t, \cdot)$ is increasing in the second argument, $I_1\rBrackets{t, \lambda^\ast_1(x_1(\psibar_1)) \Ztilde(t)}$ is decreasing in $\psibar_1$, $\overline{\mathcal{V}}_1\rBrackets{x_1(\psibar_1)}$ is decreasing in $\psibar_1$. Therefore:
\begin{equation*}
    \psi_1^\ast =  \argmax_{\psibar_1 \in \sBrackets{0, \frac{v_1}{F(0)}}} \overline{\mathcal{V}}_1\rBrackets{x_1(\psibar_1)} = 0.
\end{equation*}
Thus, \eqref{eq:UoC_psiStar} and \eqref{eq:UoC_VStar} holds. Setting $\psibar_1 =  \psi_1^\ast = 0$ in the results obtained in \textit{Step 1}, we conclude the proof.
\end{proof}

\begin{proof}[Proof of Proposition \ref{prop:solution_to_UoW_auxiliary_unconstrained}]
Points 1 and 2 follow from Theorem 3 in \cite{Korn2005} and Remark 4 (b) to that theorem. The existence of $\pitilast_2$ is also established there. If $\Xtilast_2(t;\widetilde{x}_2) = \widetilde{g}_2(t, W^{\Prob}(t))$ for some $\widetilde{g}_2(t, w) \in \mathcal{C}^{1,2}\rBrackets{[0, T] \times \R}$ with $\widetilde{g}_2(0, 0) = \widetilde{x}_2$, then Theorem 21 in \cite{Korn1997} is applicable and yields \eqref{eq:UoW_auxiliary_unconstrained_problem_pi}, whereas its functional form can be obtained from the standard delta-hedging arguments.
\end{proof}

\begin{proof}[Proof of Proposition \ref{prop:solution_to_UoW_fixed_psi}]

    First, we consider the case where $x_2 < x_B$. The terminal-wealth constraint on the liquid wealth is given by $X_2^{x_2, \pi} \geq B$, where $B$ is defined in \eqref{eq:B_definition}. As per \eqref{eq:UoW_fixed_psi_minimal_capital}, $x_B$ is the minimal amount of intial capital required to replicate $B$. Since $x_2 < x_B$,  the investor has an insufficient amount of initial capital to satisy the terminal-wealth constraint, which is why 
    \eqref{OP:UoW_fixed_psi} does not have a solution.
    
    Second, we consider the case where $x_2 = x_B$. The investor needs exactly $x_B$ to replicate the lower bound on terminal wealth. As there is no additional budget left, it follows that $\Xbarast_2(t; x_2) = X_B(t;x_B)$ and $\pibarast_2(t;x_2) = \pi_B(t; x_B)$ for $\forall \tin$.
    
    Third, we consider the  case where $x_2 > x_B$ and structure the derivation of the optimal decisions as it is done in the proof of Theorem 5 in \cite{Korn2005}.

    \textit{Step 1.} Any admissible $\pibar$ for \eqref{OP:UoW_fixed_psi} leads to a terminal wealth $X^{x_2, \pibar}(T)$, which can be written as:
    \begin{equation*}
        X^{x_2, \pibar}(T) = B + \rBrackets{X^{x_2, \pibar}(T) - B} = B + \widetilde{B}^{x_2, \pibar}, 
    \end{equation*}
    where:
    \begin{align*}
        B \geq 0,\,\,\widetilde{B}^{x_2, \pibar} \geq 0,\,\,\EP\sBrackets{\Ztilde(T) \widetilde{B}^{x_2, \pibar} } = x_2 - x_B = \xtil_2.
    \end{align*}
    Hence, $\widetilde{B}^{x_2, \pibar}$ is a candidate for a terminal wealth obtained by implementing a suitable portfolio strategy $\pitil \in \overline{\mathcal{A}}(\xtil_2; 0)$.

    \textit{Step 2.} Any admissible $\pitil$ for \eqref{OP:UoW_auxiliary_unconstrained_compact} leads to a terminal wealth $\Xtil^{\xtil_2, \pitil}(T) \geq 0$. Defining
    \begin{equation}\label{eq:Btil_B_relation}
        B^{\xtil_2, \pitil} = \Xtil^{\xtil_2, \pitil}(T) + B,
    \end{equation}
    we get that:
    \begin{equation*}
        B^{\xtil_2, \pitil} \geq B,\,\, \EP\sBrackets{\Ztilde(T)B^{\xtil_2, \pitil}} = \xtil_2 + x_B \stackrel{\eqref{eq:xtil_xB_relation}}{=} x_2.
    \end{equation*}
    This implies that $B^{\xtil_2, \pitil}$ is a candidate for a terminal wealth obtained by executing a suitable strategy $\pibar \in \overline{\mathcal{A}}(x_2; B)$. 
    
    \textit{Step 3.} In Steps 1 and 2, we established that each admissible terminal wealth in \eqref{OP:UoW_fixed_psi} can be obtained from a corresponding admissible terminal wealth in  \eqref{OP:UoW_auxiliary_unconstrained_compact} and vice versa. Furthermore:
    \begin{equation*}
        U_2\rBrackets{T, B^{\xtil_2, \pitil} + \psibar_2F(T)} \stackrel{\eqref{eq:Btil_B_relation}}{=} U_2\rBrackets{T, \Xtil^{\xtil_2, \pitil}(T) + B + \psibar_2F(T)} \stackrel{(i)}{=}  \widetilde{U}\rBrackets{ \Xtil^{\xtil_2, \pitil}(T)},
    \end{equation*}
    where we use in (i) the definition of $B$ as per \eqref{eq:B_definition} and the definition of $\widetilde{U}(x)$ as per \eqref{eq:def_U_tilde}. Finally, since the terminal wealth optimal for \eqref{OP:UoW_auxiliary_unconstrained_compact} exists and can be found as indicated in Proposition \ref{prop:solution_to_UoW_auxiliary_unconstrained}, we conclude that
    \eqref{eq:UoW_fixed_psi_X} holds. From risk-neutral valuation and the absence of arbitrage arguments, we have:
    \begin{align*}
         \Xbarast_2(t; x_2) = \EP\sBrackets{B^\ast \Ztilde(T) / \Ztilde(t)|\mathcal{F}(t)} & \stackrel{\eqref{eq:UoW_fixed_psi_X}}{=} \EP\sBrackets{B \Ztilde(T) / \Ztilde(t)|\mathcal{F}(t)} + \EP\sBrackets{\widetilde{B}^\ast \Ztilde(T) / \Ztilde(t)|\mathcal{F}(t)} \\
         & \stackbin[\eqref{eq:UoW_auxiliary_unconstrained_problem_X}]{\eqref{eq:contingent_claim_B_fair_value}}{=} X_B(t; x_B) + \Xtilast_2\rBrackets{t;\xtil_2},
    \end{align*}
    i.e., \eqref{eq:UoW_auxiliary_unconstrained_problem_X_t} holds.

    \textit{Step 4.} From \eqref{eq:UoW_auxiliary_unconstrained_problem_X_t}, we obtain:
    \begin{equation}\label{eq:SDE_links_for_Xbarast_2}
        d\Xbarast_2(t; x_2) = d X_B(t; x_B) + d \Xtilast_2\rBrackets{t;\xtil_2}.
    \end{equation}

    Plugging the following SDEs
     \begin{align*}
        d\Xbarast_2(t; x_2) &= \Xbarast_2(t; x_2)\rBrackets{\rBrackets{r + \pibarast_2(t; x_2)(\mu - r)}dt + \pibarast_2(t; x_2)  \sigma dW^{\Prob}(t)};\\
        dX_B(t; x_B) &= X_B(t; x_B)\rBrackets{\rBrackets{r + \pi_B(t; x_B)(\mu - r)}dt + \pi_B(t; x_B) \sigma dW^{\Prob}(t)};\\
        d\Xtilast_2\rBrackets{t;\xtil_2} &= \Xtilast_2\rBrackets{t;\xtil_2}\rBrackets{\rBrackets{r + \pitilast_2(t; \xtil_2)(\mu - r)}dt + \pitilast_2(t; \xtil_2) \sigma dW^{\Prob}(t)};
    \end{align*}
    into \eqref{eq:SDE_links_for_Xbarast_2} and equating the coefficients next to $\sigma dW^{\Prob}(t)$, we get
    \begin{equation*}
           \pibarast_2(t; x_2) \Xbarast_2(t; x_2) = \pi_B(t; x_B) X_B(t; x_B) + \pitilast_2(t; \xtil_2) \Xtilast_2\rBrackets{t;\xtil_2},
    \end{equation*}
    from which \eqref{eq:UoW_fixed_psi_pi} follows due to $\Xbarast_2(t; x_2) > 0$ $\Prob$-a.s. for all $\tin$.
\end{proof}

\begin{proof}[Proof of Lemma \ref{lem:v_2_min}]
    First, we prove that if $v_2 \geq v_2^{\min}$, then the set of admissible strategies is non-empty. For $v_2 \geq v_2^{\min}$, the strategy $\rBrackets{\pi_2 \equiv 0, \psi_2 = 0, c_2 \equiv 0}$ is admissible, since $V^{v_2, (0,0,0)}(T) = v_2 \exp(rT) \geq v_2^{\min}\exp(rT) =  \Vunder$ and other admissibility conditions are trivially satisfied.

    Second, we prove that if the set of admissible strategies is non-empty, then $v_2 \geq v_2^{\min}$. We show it by contradiction. Assume that there exists a strategy $\rBrackets{\widehat{\pi}_2, \widehat{\psi}_2, \widehat{c}_2}$ admissible for \eqref{OP:UoW_original} and it requires an initial capital $\widehat{v}_2 <  v_2^{\min}$. It means that $V^{\widehat{v}_2, \rBrackets{\widehat{\pi}_2, \widehat{\psi}_2, \widehat{c}_2}}(T) = X^{\widehat{v}_2 - \widehat{\psi}_2 F_0, \rBrackets{\widehat{\pi}_2, \widehat{c}_2}}(T) + \widehat{\psi}_2 F(T)\geq \Vunder$ and $X^{\widehat{v}_2 - \widehat{\psi}_2 F_0, \rBrackets{\widehat{\pi}_2, \widehat{c}_2}}(T) \geq 0$ $\Prob$-a.s., which is equivalent to (cf. \eqref{eq:B_definition}):
    \begin{equation}\label{eq:v_2_min_contradiction_step_1}
        X^{\widehat{v}_2 - \widehat{\psi}_2 F_0, \rBrackets{\widehat{\pi}_2, \widehat{c}_2}}(T)  \geq \rBrackets{\Vunder - \widehat{\psi}_2 F(T)}^{+}.
    \end{equation}

    From risk-neutral pricing, we get that \eqref{eq:v_2_min_contradiction_step_1} implies:
    \begin{equation}\label{eq:v_2_min_contradiction_step_2}
        \widehat{v}_2 - \widehat{\psi}_2 F_0 \geq \EP\sBrackets{ \Ztilde(T) \rBrackets{\Vunder - \widehat{\psi}_2 F(T)}^{+} } =:x_B\rBrackets{\widehat{\psi}_2}
    \end{equation}
    
    Adding $\widehat{\psi}_2 F_0$ to both sides of \eqref{eq:v_2_min_contradiction_step_2}, we obtain:
\begin{equation}\label{eq:v_2_min_optimization}
    \widehat{v}_2 \geq  x_B\rBrackets{\widehat{\psi}_2} + \widehat{\psi}_2 F_0 \geq \min_{\psibar_2 \in \sBrackets{0, \frac{v_2}{F(0)}}}\Bigl(\underbrace{x_B\rBrackets{\psibar_2} + \psibar_2 F(0)}_{:=h(\psibar_2)}\Bigr). 
\end{equation}
To find the minimum of $h(\psibar_2)$, we will find its derivative and show that it is strictly negative for any admissible $\psibar_2$, which will imply that $v_2^{\min} = h(0)$. 

Note that by definitions \eqref{eq:B_definition} and \eqref{eq:UoW_fixed_psi_minimal_capital} we have:
\begin{equation*}
    x_B(\psibar_2) = \EP\sBrackets{\Ztilde(T)\rBrackets{\Vunder - \psibar_2 F(T)}^{+}} = \EQ\sBrackets{\exp\rBrackets{-rT}\rBrackets{\Vunder - \psibar_2 F(T)}^{+}}.
\end{equation*}
Since $F(T)$ is $\mathcal{F}$-measurable and $F(T) > 0$ $\Prob$-a.s. by assumption, we can represent it as $F(T) = F_0 \cdot \exp(Y)$, where $Y$ is an $\mathcal{F}$-measurable random variable. Defining $\overline{F}_0 := \psibar_2 F_0$, we consider the above put option as a function of strike and the initial value of the underlying asset:
\begin{equation*}
    Put(\Vunder, \overline{F}_0) = \EP\sBrackets{\Ztilde(T)\rBrackets{\Vunder - \overline{F}_0\exp(Y)}^{+}}.
\end{equation*}

Since for any $\alpha > 0$ $Put(\alpha \Vunder, \alpha \overline{F}_0) = \alpha Put(\Vunder, \overline{F}_0)$ due to the linearity of the expectation operator, $Put(\cdot, \cdot)$ is a homogeneous function of order one. Thus, by Euler's homogeneous function theorem, we obtain: 
\begin{equation}\label{eq:put_expression_via_Euler}
    Put(\Vunder, \overline{F}_0) =  \Vunder \cdot \frac{\partial Put(\Vunder, \overline{F}_0)}{\partial \Vunder} + \overline{F}_0 \cdot \frac{\partial Put(\Vunder, \overline{F}_0)}{\partial \overline{F}_0}.
\end{equation}
Simultaneously, we have the following representation of $Put(\Vunder, \overline{F}_0)$:
\begin{align}
    Put(\Vunder, \overline{F}_0) &= \EQ\sBrackets{\exp\rBrackets{-r T}\rBrackets{\Vunder - \overline{F}_0 \exp(Y)}\mathbbm{1}_{\cBrackets{\Vunder > \overline{F}_0 \exp(Y)}}} \notag\\
    & \stackrel{(i)}{=} \exp \rBrackets{-rT} \Vunder \EQ\sBrackets{\mathbbm{1}_{\cBrackets{\Vunder > \overline{F}_0 \exp(Y)}}} - \exp \rBrackets{-r T} \overline{F}_0\EQ\sBrackets{\exp(Y) \mathbbm{1}_{\cBrackets{\Vunder > \overline{F}_0 \exp(Y)}}} \notag\\
    & \stackrel{(ii)}{=} \exp \rBrackets{-rT} \Vunder \Q\rBrackets{\Vunder > \overline{F}_0 \exp(Y)} - \exp\rBrackets{-r T} \overline{F}_0 \underbrace{\mathbb{M}\rBrackets{\Vunder > \psibar_2 F_0 \exp(Y)}}_{=:m(\psibar_2)}, \label{eq:put_expression_via_probabilities}
\end{align}
where we use in (i) the linearity of the expectation operator, in (ii) the definition of $\overline{F}_0$ and the change of measure so that $\exp\rBrackets{Y}$ cancels out under the new measure $\mathbb{M}$ (see \cite{Geman1995} for details). Thus, we obtain by comparing \eqref{eq:put_expression_via_Euler} and \eqref{eq:put_expression_via_probabilities}:
\begin{equation}\label{eq:put_delta}
    \frac{\partial Put(\Vunder, \overline{F}_0)}{\partial \overline{F}_0} = - \exp\rBrackets{-r T} m(\psibar_2).
\end{equation}
Now we can calculate the derivative of $h(\psibar_2)$ as follows:
\begin{align*}
    \frac{d h(\psibar_2)}{d \psibar_2}&= \frac{d}{d \psibar_2} \rBrackets{x_B\rBrackets{\psibar_2} + \psibar_2 F(0)} = \frac{d}{d \psibar_2} \rBrackets{Put(\Vunder, \overline{F}_0)} + \frac{d}{d \psibar_2}\rBrackets{\psibar_2 F(0)}
    \stackrel{(i)}{=} \frac{d Put(\Vunder, \overline{F}_0)}{d \overline{F}_0} \cdot \frac{d \overline{F}_0}{d \psibar_2} + F_0\\
    & \stackrel{(ii)}{=} - \exp\rBrackets{-r T} m(\psibar_2) F_0 + F_0 = F_0\rBrackets{1 - \exp\rBrackets{-r T} m(\psibar_2)} \stackrel{(iii)}{\geq} 0.
\end{align*}
where we use in (i) the chain rule, in (ii) \eqref{eq:put_delta} and the definition of $\overline{F}_0$, in (iii) that $m(\psibar_2)\in [0,1]$ and $r > 0$. Therefore, $h(\psibar_2)$ is strictly decreasing in $\psibar_2$ and we obtain from \eqref{eq:v_2_min_optimization} that
\begin{align*}
    \min_{\psibar_2 \in \sBrackets{0, \frac{v_2}{F(0)}}}\rBrackets{x_B\rBrackets{\psibar_2} + \psibar_2 F(0)} = \rBrackets{x_B\rBrackets{0} + 0 \cdot F(0)} =  \EQ\sBrackets{\exp\rBrackets{-r T}(\Vunder)^{+}} = \exp\rBrackets{-r T}\Vunder =: v_2^{\min}.
\end{align*}
Plugging this result into \eqref{eq:v_2_min_optimization}, we get the inequality:
\begin{equation*}
    \widehat{v}_2 \geq v_2^{\min},
\end{equation*}
which contradicts our assumption of $\widehat{v}_2 < v_2^{\min}$ at the beginning of the proof by contradiction. Therefore, the assumption was wrong, i.e., the does not exist an admissible for \eqref{OP:UoW_original} strategy that requires less initial capital than $v_2^{\min}$.
\end{proof}

\begin{proof}[Proof of Proposition \ref{prop:UoW_problem_solution}]
First, we consider the case where $v_2 < v_2^{\min}$. By Lemma \ref{lem:v_2_min}, the problem does not admit a solution.

Second, we consider the case where $v_2 = v_2^{\min}$. From the proof of Lemma \ref{lem:v_2_min}, we know that $v_2^{\min}$ is the amount of initial capital sufficient to solve \eqref{OP:UoW_fixed_psi} with $\psibar_2 = 0$ and, due to the strict increasingness of $h(\psibar_2) = x_B(\psibar_2) +  \psibar_2 F(0)$ (see the proof of Lemma \ref{lem:v_2_min}), the investor needs more initial capital to solve \eqref{OP:UoW_fixed_psi} for any $\psibar_2 > 0$ . Thus, for $v_2 = v_2^{\min}$ it must be $\psibar_2^\ast = 0$. It is straightforward to verify that in this case the investor has just enough initial capital to replicate the lower bound $\Vunder$ on the terminal wealth by investing only in risk-free asset, i.e., $\pi_2^\ast(t;v_2)=0$, $\psi_2^\ast = 0$, $c_2^\ast(t;v_2) = 0$, $X_2^\ast(t;v_2) = \exp\rBrackets{ -r (T - t)}\Vunder = v_2 \exp\rBrackets{r t}$ for $\forall \tin$, and $V^\ast_2(T; v_2) = \Vunder$.

Third, we consider the case where $v_2 > v_2^{\min}$. For an arbitrarily fixed admissible $\psibar_2 \in \sBrackets{0, \frac{v_2}{F(0)}}$, Proposition \ref{prop:solution_to_UoW_fixed_psi} yields the solution to \eqref{OP:UoW_fixed_psi}. Similarly to Lemma A.3 in \cite{Desmettre2016}, we can show that $\overline{\mathcal{V}}_2\rBrackets{x_2\rBrackets{\psibar_2}}$ is strictly concave in $\psibar_2$. Thus, the problem of maximizing VF associated with \eqref{OP:UoW_fixed_psi} over admissible positions $\psibar_2$ in the fixed-term asset
\begin{equation*}
    \mathcal{V}_2(v_2) = \max_{\psibar_2 \in \sBrackets{0, \frac{v_2}{F(0)}}} \overline{\mathcal{V}}_2\rBrackets{x_2\rBrackets{\psibar_2}} = \max_{\psibar_2 \in \sBrackets{0, \frac{v_2}{F(0)}}} \widetilde{\mathcal{V}}_2\rBrackets{\xtil_2\rBrackets{\psibar_2}}
\end{equation*}
has a unique solution \eqref{eq:UoW_psiStar}, which is the optimal position in the fixed-term asset in \eqref{OP:UoW_original}. Denoting the maximizer by $\psi^\ast_2$ and inserting it in the results of Proposition \ref{prop:solution_to_UoW_fixed_psi}, we obtain \eqref{eq:UoW_XStar} and \eqref{eq:UoW_piStar}. 

Finally, we obtain \eqref{eq:UoW_VStar} from the following equalities:
\begin{align*}
    V^\ast_2(T; v_2) & = X_2^\ast(T;v_2) + \psi_2^\ast F(T) \\
    & \stackrel{\eqref{eq:UoW_auxiliary_unconstrained_problem_X_t}}{=} X_B\rBrackets{ T; x_B\rBrackets{\psiStar_2} } + \Xtilast_2\rBrackets{ T; \xtil_2\rBrackets{\psiStar_2} }  + \psiStar_2 F(T)\\
    & \stackbin[\eqref{eq:contingent_claim_B_fair_value} ]{\eqref{eq:UoW_auxiliary_unconstrained_problem_X_T}}{=} \rBrackets{\Vunder - \psiStar_2 F(T)}^{+}  + \widetilde{I}_2\rBrackets{\widetilde{\lambda}_2^\ast(\widetilde{x}_2)\Ztilde(T)} + \psi_2^\ast F(T) \\
    & \stackrel{\eqref{eq:def_I_2_tilde}}{=} \max\cBrackets{\Vunder, \psiStar_2F(T)} + \max \cBrackets{I_2\rBrackets{\widetilde{\lambda}_2^\ast\Ztilde(T)} - \max\cBrackets{\Vunder, \psiStar_2F(T)}, 0}\\
    & = \max \cBrackets{\Vunder, \psiStar_2 F(T), I_2\rBrackets{\widetilde{\lambda}_2^{\ast} \Ztildet{T}}}.
\end{align*}
\end{proof}

\begin{proof}[Proof of Proposition \ref{prop:maximization_problem_for_vStar_1_vStar_2}]
    In the first step, we show that the left-hand side (LHS) of \eqref{eq:value_functions_merging} is larger than or equal to the RHS of \eqref{eq:value_functions_merging}. In the second step, we show the reverse inequality. As a result, we conclude the equality.

    \textit{Step 1.} Let $\rBrackets{\piStar_1(t, v_1), \cStar_1(t,v_1), \psiStar_1(v_1)}$ be the solution to \eqref{OP:UoC_problem}, $\mathcal{V}_1(v_1; K_c)$ be the corresponding value function. Let $\rBrackets{\piStar_2(t, v_2), \cStar_2(t,v_2), \psiStar_2}$ be the solution to \eqref{OP:UoW_original}, $V^\ast_2(t;v_2)$ be the corresponding terminal portfolio value, and $\mathcal{V}_2(v_1; K_V)$ be the corresponding value function. Then:
    \begin{align*}
        \mathcal{V}_1(v_1; K_c) &+ \mathcal{V}_2(v_2; K_V) = \EP
        \sBrackets{\intzeroT U_1(t, \cStar_1(t;v_1))dt + U_2(T, V_2^\ast(T;v_2))}\\
        &\stackrel{}{\leq} \max_{\rBrackets{\pi,c, \psi} \in \mathcal{A}\rBrackets{v_0, K_c, K_V}}\EP
        \sBrackets{\intzeroT U_1(t, c(t))dt + U_2(T, V^{v_0, \rBrackets{\pi,c, \psi}}(T))}\stackrel{}{=}\mathcal{V}_0(v_0; K_c, K_V).
    \end{align*}
    Since the above inequality holds for any admissible $v_1$ and $v_2$ such that $v_1 + v_2 = v_0$, we have proven that:
    \begin{equation*}
        \mathcal{V}(v_0; K_c, K_V) \geq \max_{(v_1, v_2) \in \Lambda(v_0)} \cBrackets{\mathcal{V}(v_1; K_c) + \mathcal{V}(v_2; K_V)}
    \end{equation*}

    \textit{Step 2.} Let $\rBrackets{\piStar, \cStar, \psiStar}$ be the solution to \eqref{OP:DS2016_with_V_c_lower_bounds}, $\mathcal{V}_0(v_0; K_c, K_V)$ be the corresponding value function, $X^\ast(T)$ the terminal value of the corresponding liquid portfolio. Define:
    \begin{equation*}
        v_1:=\EP\sBrackets{\intzeroT \Ztilde(t)\cStar(t)dt} \quad \text{and} \quad v_2:=\EP\sBrackets{\Ztilde(T)X^\ast(T)} + \psiStar F(0).
    \end{equation*}
    Then $v_1 + v_2 = v_0$ and:
    \begin{align*}
        \mathcal{V}&(v_0;K_c, K_V) =  \EP
        \sBrackets{\intzeroT U_1(t, \cStar(t))dt + U_2(T, V^{v_0, \rBrackets{\piStar,\cStar, \psiStar}}(T))} \\
        &\stackrel{}{=} \EP
        \sBrackets{\intzeroT U_1(t, \cStar(t))dt} + \EP
        \sBrackets{U_2(T, V^{v_0, \rBrackets{\piStar,\cStar, \psiStar}}(T))} \\
        &\stackrel{}{\leq} \max_{(\pi_1, c_1, \psi_1) \in \mathcal{A}(v_1; K_c)} \EP
        \sBrackets{\intzeroT U_1(t, c_1(t))dt} + \max_{(\pi_2, c_2, \psi_2) \in \mathcal{A}(v_2; K_V)} \EP \sBrackets{U_2(T, V^{v_2, \rBrackets{\pi_2, c_2, \psi_2}}(T))} \\
        &\stackrel{}{=} \mathcal{V}_1(v_1; K_c) + \mathcal{V}_2(v_2; K_V) \stackrel{}{\leq} \max_{(v_1, v_2) \in \Lambda(v_0)}\left\{\mathcal{V}_1(v_1; K_c) + \mathcal{V}_2(v_2; K_V)\right\}.
    \end{align*}
    So we have proven that:
    \begin{equation*}
        \mathcal{V}(v_0;K_c, K_V) \leq \max_{(v_1, v_2) \in \Lambda(v_0)}\left\{\mathcal{V}_1(v_1; K_c) + \mathcal{V}_2(v_2; K_V)\right\}.
    \end{equation*}

    Combining the inequalities proven in Step 1 and Step 2, we conclude the equality \eqref{eq:value_functions_merging}. 
\end{proof}

\begin{proof}[Proof of Lemma \ref{lem:vStar_1_vStar_2}]
    According to Proposition \ref{prop:maximization_problem_for_vStar_1_vStar_2}, the optimal $v_1$ and $v_2$ are found by solving \eqref{eq:value_functions_merging}. Due to the constraint set $\Lambda(v_0)$ in this maximization problem, we express $v_2= v_0 - v_1$ and consider an equivalent problem:
    \begin{equation}\label{OP:maximization_problem_for_vStar_1}
        \max_{v_1 \in \sBrackets{v_1^{\min}, v_0 - v_2^{\min}}} \left\{\mathcal{V}_1(v_1; K_c) + \mathcal{V}_2(v_0 - v_1; K_V)\right\},
    \end{equation}
    which has a solution since $\mathcal{V}_1$ and $\mathcal{V}_2$ are continuous in $v_1$ and the constraint set $\overline{\Lambda}(v_0) := \sBrackets{v_1^{\min}, v_0 - v_2^{\min}}$ is compact. Moreover, the solution is unique due to the concavity of the objective function:
    \begin{equation*}
        \frac{\partial^2 }{\partial v_1^2} \rBrackets{\mathcal{V}_1(v_1; K_c)  + \mathcal{V}_2(v_0 - v_1; K_V)} = \mathcal{V}_1''(v_1; K_c) + \mathcal{V}_2''(v_0 - v_1; K_V) \stackrel{(i)}{<}0, 
    \end{equation*}
    where we use in (i) the strict concavity of $\mathcal{V}_1$ and $\mathcal{V}_2$ as functions of initial wealth.

    Therefore, $v_1^\ast$ is either a critical point or the point on the boundary of the constraint set $\overline{\Lambda}(v_0)$. The equation for finding the critical point $\bar{v}_1$ is given by \eqref{eq:vBar_1}. Thus, $v_1^\ast$ is given by \eqref{eq:vBar_1} and $v_2^\ast$ by \eqref{eq:vStar_2}.
\end{proof}

\begin{proof}[Proof of Proposition \ref{prop:solution_to_DS2016_with_V_c_lower_bounds}]
    \eqref{eq:value_function} follows from Proposition \ref{prop:maximization_problem_for_vStar_1_vStar_2} and Lemma \ref{lem:vStar_1_vStar_2}.
    
    Due to the decomposition of \eqref{OP:DS2016_with_V_c_lower_bounds} in \eqref{OP:UoC_problem} and \eqref{OP:UoW_original} and the differentiation between the intitial capital for liquid investments and the initial capital for illiquid investment, we trivially have \eqref{eq:XStar_t}. Furthermore, we have $\psiStar = \psiStar_1 + \psiStar_2 \stackrel{\eqref{eq:UoC_psiStar}}{=} \psiStar_2$, i.e., \eqref{eq:psiStar} holds, and $V^{\ast}(T; v_0) = V^{\ast}_1(T, v^{\ast}_1) + V^{\ast}_2(T, v^{\ast}_2) = \XStar_1(T; v_1^{\ast}) + \psiStar_1 F(T) + \XStar_2(T;v_2^{\ast}) + \psiStar_2 F(T) \stackbin[\psiStar_1 = 0]{X_1^\ast(T; v_1^{\ast}) = 0}{=} \XStar_2(T; v_2^{\ast}) + \psiStar_2 F(T)$, i.e., \eqref{eq:VStar_T} holds.
    
    As for the liquid portfolio value, we have:
    \begin{equation}\label{eq:X_SDE_compact}
        d\XStar(t; v_0) = d\XStar_1(t; v_1^{\ast}) + d\XStar_2(t; v_2^{\ast}).
    \end{equation}
    Inserting the respective optimal controls in \eqref{eq:SDE_wealth_process}, we obtain three SDEs:
    \begin{align*}
        d\XStar(t; v_0) &= \XStar(t; v_0)\rBrackets{\rBrackets{r + \piStar(t; v_0)(\mu - r)}dt +\piStar(t; v_0)  \sigma dW^{\Prob}(t)} - \cStar(t)dt;\\
        d\XStar_1(t; v_1^{\ast}) &= \XStar_1(t; v_1^{\ast})\rBrackets{\rBrackets{r + \piStar_1(t; v_1^{\ast})(\mu - r)}dt +\piStar_1(t; v_1^{\ast})  \sigma dW^{\Prob}(t)} - \cStar_1(t)dt;\\
        d\XStar_2(t; v_2^{\ast}) &= \XStar_2(t; v_2^{\ast})\rBrackets{\rBrackets{r + \piStar_2(t; v_2^{\ast})(\mu - r)}dt +\piStar_2(t; v_2^{\ast})  \sigma dW^{\Prob}(t)} - \underbrace{\cStar_2(t)}_{\stackrel{\eqref{eq:UoW_cStar}}{=}0}dt.
    \end{align*}

    Inserting the above three SDEs into \eqref{eq:X_SDE_compact} and equating the terms multiplied by $\sigma dW^{\Prob}(t)$, we get \eqref{eq:piStar}. Comparing the terms multiplied by $dt$, we get \eqref{eq:cStar}.
\end{proof}

\section{Proofs for Section \ref{sec:explicit_formulas_for_power_U}}\label{app:proofs_power_U}
\begin{proof}[Proof of Lemma \ref{lem:auxiliary_E_of_Z_T}]
    It follows from Lemma A.2.1 in Appendix A in \cite{Havrylenko2023}, where we replace $T$ by $s$.
\end{proof}

\bigskip

\begin{proof}[Proof of Lemma \ref{lem:F_Ztilde_relation}]
    Take any $t_1, t_2$ such that $0 \leq t_1 \leq t_2 \leq T$. It follows from \eqref{eq:explicit_pricing_kernel} that:
\begin{equation*}
\rBrackets{\Ztildet{t_2}}^{-\frac{\sigma_F}{\gamma}} =  \rBrackets{\Ztildet{t_1}}^{-\frac{\sigma_F}{\gamma}}\exp \left( \frac{\sigma_F}{\gamma}\left(r +  \frac{1}{2} \gamma^2 \right)(t_2 - t_1) + \sigma_F \rBrackets{W^{\Prob}(t_2) - W^{\Prob}(t_1)} \right).
\end{equation*}
Using the above relation and \eqref{eq:F_SDE_case_studies}, we get:
    \begin{align}
        F(t_2) &= F(t_1) \exp \rBrackets{\rBrackets{\mu_F - \frac{1}{2}\sigma_F^2} (t_2 - t_1) + \sigma_F \rBrackets{W^{\Prob}(t_2) - W^{\Prob}(t_1)} }\notag  \\
        &= F(t_1) \frac{\rBrackets{\Ztilde(t_1)}^{-\frac{\sigma_F}{\gamma}}}{ \rBrackets{\Ztilde(t_1)}^{-\frac{\sigma_F}{\gamma}} } \exp \Bigl(\rBrackets{\mu_F - \frac{1}{2}\sigma_F^2} (t_2 - t_1) + \sigma_F \rBrackets{W^{\Prob}(t_2) - W^{\Prob}(t_1)} \notag \\
        & \quad -\frac{\sigma_F}{\gamma}\rBrackets{r + \frac{1}{2}\gamma^2} (t_2 - t_1) + \frac{\sigma_F}{\gamma}\rBrackets{r + \frac{1}{2}\gamma^2}(t_2 - t_1) \Bigr)\notag  \\
        & =  \frac{F(t_1)}{ \rBrackets{\Ztilde(t_1)}^{-\frac{\sigma_F}{\gamma}} } \exp \rBrackets{\rBrackets{\mu_F - \frac{1}{2}\sigma_F^2 - \frac{\sigma_F r}{\gamma} - \frac{1}{2}\sigma_F \gamma} (t_2 - t_1)} \rBrackets{\Ztilde(t_2)}^{-\frac{\sigma_F}{\gamma}}, \notag
    \end{align}
for any $t_1, t_2$ such that $0 \leq t_1 \leq t_2 \leq T$. 
\end{proof}

\bigskip

\begin{proof}[Proof of Corollary \ref{cor:UoC_problem_solution_power}]
    Points (i)-(iii) follow from straightforward calculations based on Proposition \ref{prop:UoC_problem_solution} and Lemma \ref{lem:auxiliary_E_of_Z_T}. To prove Point (iv), let $\pi_1^\ast(t;v_1)$ denote the investment strategy driving the optimal liquid portfolio and let $g_1(t, \Ztilde(t))$ denote the right-hand side of \eqref{eq:UoC_XStar_power} as a function of $t$ and $\Ztilde(t)$.. From \eqref{eq:SDE_wealth_process}, we know that:
    \begin{equation*}
        dX_1^\ast(t;v_1) = X_1^\ast(t;v_1)\rBrackets{\rBrackets{r + \pi_1^\ast(t;v_1)(\mu - r)}dt + \pi_1^\ast(t;v_1)  \sigma dW^{\Prob}(t)}.
    \end{equation*}
    From Ito{\^o}'s lemma and the SDE \eqref{eq:SDE_pricing_kernel}, we know that:
    \begin{equation*}
        dX_1^\ast(t;v_1) = d g_1(t, \Ztilde(t)) = \rBrackets{\dots} dt - \gamma \Ztildet{t} \frac{\partial g_1(t, \Ztilde(t))}{\partial \widetilde{z}}dW^{\Prob}(t).
    \end{equation*}
    Since the above two SDEs describe the same stochastic process, the SDEs must be equal. This implies the equality of diffusion terms of those SDEs, which is why \eqref{eq:UoC_piStar_power} immediately follows.
\end{proof}

\bigskip

\begin{proof}[Proof of Proposition \ref{prop:put_on_psiBar_F}]
    The initial budget \eqref{eq:X_0_B_explicit} follows from:
    \begin{align*}
                x_B &\stackrel{(a)}{=} \EP\sBrackets{ \max \cBrackets{\Vunder - \psibar_2 F(T), 0} \Ztilde(T) / \Ztilde(0)|\mathcal{F}(0)}\\
                & \stackrel{(b)}{=} \EP\sBrackets{  \Ztilde(T) \Vunder \mathbbm{1}_{\cBrackets{\Vunder \geq \psibar_2 F(T)}}|\mathcal{F}(0)} + \EP\sBrackets{  \Ztilde(T) \psibar_2 F(T) \mathbbm{1}_{\cBrackets{\Vunder \geq \psibar_2 F(T)}}|\mathcal{F}(0)}\\
                &\stackrel{(c)}{=} \Vunder \xi \rBrackets{T, 0, 1, 1, \rBrackets{\psibar_2 F_0 h(0, T) / \Vunder}^{\frac{\gamma}{\sigma_F}}, +\infty; \gamma, r} \\
                & \quad - \psibar F_0 h(0, T) \xi \rBrackets{T, 0, 1, 1 - \frac{\sigma_F}{\gamma}, \rBrackets{\psibar_2 F_0 h(0, T) / \Vunder}^{\frac{\gamma}{\sigma_F}}, +\infty; \gamma, r},
    \end{align*}
    where we use in (a) the definition of $B$ and \eqref{eq:contingent_claim_B_fair_value}, in (b) the linearity of expectation, in (c) Lemma \ref{lem:auxiliary_E_of_Z_T}. Analogous calculations yield \eqref{eq:X_t_B_explicit}. \eqref{eq:pi_B_explicit} is shown analogously to the proof of \eqref{eq:UoC_piStar_power}.
\end{proof}

\bigskip

\begin{proof}[Proof of Corollary \ref{cor:UoW_problem_solution_power}]
    All results of the proposition can be shown via straightforward but lengthy calculations based on Proposition \ref{prop:UoW_problem_solution}, Lemma \ref{lem:auxiliary_E_of_Z_T} and Lemma \ref{lem:F_Ztilde_relation}. When $\psibar_2 > 0$ ($\psiStar_2 > 0$), three different cases depending on the value of $D(\gamma, \sigma_F, p_2)$ emerge due to the terms that involve:
    \begin{equation*}
        \rBrackets{\widetilde{\lambda}^\ast_2 \Ztilde(T)}^{ \frac{1}{p_2 - 1} } \lesseqgtr \psibar_2 F(T)\Leftrightarrow \rBrackets{\Ztilde(T)}^{ D(\gamma, \sigma_F, p_2)} \lesseqgtr \rBrackets{\psibar_2 F(t) h(t, T) \rBrackets{\Ztildet{t}}^{\sigma_F / \gamma}}^{p_2 - 1} / \widetilde{\lambda}^\ast_2.
    \end{equation*}
\end{proof}

\bigskip

\begin{proof}[Proof of Corollary \ref{cor:merging_subproblems_power_U}]
    The proof is based on straightforward but lengthy calculations that rely on Proposition \ref{prop:solution_to_DS2016_with_V_c_lower_bounds}, Corollary \ref{cor:UoC_problem_solution_power} and Corollary \ref{cor:UoW_problem_solution_power}.
\end{proof}

\end{document}